\documentclass[11pt]{article}

\usepackage[utf8]{inputenc}
\usepackage[T1]{fontenc}
\usepackage[margin=1in]{geometry}
\usepackage{enumitem}
\usepackage{amsmath,amsfonts,amssymb,amsthm}
\usepackage{mathtools}
\usepackage{scrlfile}
\usepackage{thmtools}
\usepackage{thm-restate}
\usepackage[font=small]{caption}

\makeatletter
\newcommand*\rel@kern[1]{\kern#1\dimexpr\macc@kerna}
\newcommand*\widebar[1]{%
  \begingroup
  \def\mathaccent##1##2{%
    \rel@kern{0.8}%
    \overline{\rel@kern{-0.8}\macc@nucleus\rel@kern{0.2}}%
    \rel@kern{-0.2}%
  }%
  \macc@depth\@ne
  \let\math@bgroup\@empty \let\math@egroup\macc@set@skewchar
  \mathsurround\z@ \frozen@everymath{\mathgroup\macc@group\relax}%
  \macc@set@skewchar\relax
  \let\mathaccentV\macc@nested@a
  \macc@nested@a\relax111{#1}%
  \endgroup
}
\makeatother
\renewcommand{\bar}{\widebar}

\newtheorem{theorem}{Theorem}

\newtheorem{lemma}[theorem]{Lemma}
\newtheorem{proposition}[theorem]{Proposition}
\newtheorem{corollary}[theorem]{Corollary}
\newtheorem{definition}[theorem]{Definition}
\newtheorem{conjecture}[theorem]{Conjecture}
\newtheorem{fact}[theorem]{Fact}
\newtheorem{claim}[theorem]{Claim}
\newtheorem{open}[theorem]{Open Problem}

\newcommand{\be}{\begin{equation}}
\newcommand{\ee}{\end{equation}}

\newcommand{\B}{\{0,1\}}

\newcommand{\OR}{\textsc{OR}}

\newcommand{\NAND}{\textsc{Nand}}

\DeclareMathOperator{\R}{R}
\DeclareMathOperator{\Q}{Q}

\DeclareMathOperator{\Dom}{Dom}

\DeclareMathOperator{\poly}{poly}

\DeclareMathOperator{\bR}{\mathbb{R}}

\DeclareMathOperator{\cost}{cost}
\DeclareMathOperator{\Stab}{Stab}
\DeclareMathOperator{\err}{err}

\newcommand{\triv}{\textsc{Triv}}
\newcommand{\kFaultDirectTrees}{\textsc{1-Fault Direct Trees}}

\DeclareMathOperator{\argmax}{argmax}

\newcommand{\bN}{\mathbb{N}}
\newcommand{\bE}{\mathbb{E}}

\usepackage[hidelinks,breaklinks]{hyperref}
\hypersetup{
    pdftitle={}, 
    pdfauthor={}, 
    colorlinks=true, 
    linkcolor=blue, 
    citecolor=blue, 
    urlcolor=blue, 
   	final=true, 
}

\newcommand{\eq}[1]{\hyperref[eq:#1]{(\ref*{eq:#1})}}
\renewcommand{\sec}[1]{\hyperref[sec:#1]{Section~\ref*{sec:#1}}}
\newcommand{\thm}[1]{\hyperref[thm:#1]{Theorem~\ref*{thm:#1}}}
\newcommand{\lem}[1]{\hyperref[lem:#1]{Lemma~\ref*{lem:#1}}}
\newcommand{\defn}[1]{\hyperref[def:#1]{Definition~\ref*{def:#1}}}
\newcommand{\prop}[1]{\hyperref[prop:#1]{Proposition~\ref*{prop:#1}}}
\newcommand{\cor}[1]{\hyperref[cor:#1]{Corollary~\ref*{cor:#1}}}
\newcommand{\fig}[1]{\hyperref[fig:#1]{Figure~\ref*{fig:#1}}}
\newcommand{\tab}[1]{\hyperref[tab:#1]{Table~\ref*{tab:#1}}}
\newcommand{\alg}[1]{\hyperref[alg:#1]{Algorithm~\ref*{alg:#1}}}
\newcommand{\app}[1]{\hyperref[app:#1]{Appendix~\ref*{app:#1}}}
\newcommand{\conj}[1]{\hyperref[conj:#1]{Conjecture~\ref*{conj:#1}}}
\newcommand{\chap}[1]{\hyperref[chap:#1]{Chapter~\ref*{chap:#1}}}
\newcommand{\clm}[1]{\hyperref[clm:#1]{Claim~\ref*{clm:#1}}}
\newcommand{\fct}[1]{\hyperref[fct:#1]{Fact~\ref*{fct:#1}}}
\newcommand{\opn}[1]{\hyperref[opn:#1]{Open Problem~\ref*{opn:#1}}}

\usepackage{tikz,tkz-graph}
\usepackage{enumitem}
\usepackage{algorithm}
\usepackage{algcompatible}
\usepackage[capitalize,nameinlink]{cleveref}

\newcommand{\entr}{\textsc{entrance}}
\newcommand{\exit}{\textsc{exit}}
\newcommand{\rooot}{\textsc{root}}

\DeclareMathOperator{\Ex}{E}

\newcommand{\eps}{\varepsilon}

\def\candy#1#2{
\begin{scope}[shift={#1},scale={#2}]
    \draw
    (0,-1) -- (0,1) -- (1,0) -- (2,1) -- (3,0) -- (4,1) -- (4,-1) -- (3,0) -- (2,-1) -- (1,0) -- (0,-1);
\end{scope}
}

\def\doublebowtie#1#2{
\begin{scope}[shift={#1},scale={#2}]
    \draw
    (0,-1) -- (0,1) -- (1,0) -- (1.9,1) -- (1.9,-1) -- (1,0) -- (0,-1);
    \draw
    (2.1,-1) -- (2.1,1) -- (3,0) -- (4,1) -- (4,-1) -- (3,0) -- (2.1,-1);
\end{scope}
}

\begin{document}

\title{Symmetries, graph properties, and quantum speedups}

\author{
Shalev Ben{-}David$^{1}$ \quad
Andrew M. Childs$^{2,3}$ \quad
Andr\'as Gily\'en$^{4,5}$ \\[2pt]
William Kretschmer$^6$ \quad
Supartha Podder$^7$ \quad
Daochen Wang$^{3,8}$ \\[10pt]
\small$^1$ Cheriton School of Computer Science, University of Waterloo \\
\small$^2$ Department of Computer Science and Institute for Advanced Computer Studies, University of Maryland \\
\small$^3$ Joint Center for Quantum Information and Computer Science, University of Maryland \\
\small$^4$ Institute for Quantum Information and Matter, California Institute of Technology \\
\small$^5$ Simons Institute for the Theory of Computing, University of California, Berkeley \\
\small$^6$ Department of Computer Science, University of Texas at Austin \\
\small$^7$ Department of Mathematics and Statistics, University of Ottawa \\
\small$^8$ Department of Mathematics, University of Maryland
}

\date{}
\maketitle

\begin{abstract}
Aaronson and Ambainis (2009) and Chailloux (2018) showed that fully
symmetric (partial) functions do not admit exponential quantum query speedups. This raises a natural question:
how symmetric must a function be before it cannot
exhibit a large quantum speedup?

In this work, we prove that hypergraph symmetries in the adjacency matrix model allow at most a polynomial separation between randomized and quantum query complexities. We also show that, remarkably, permutation groups constructed out of these symmetries are essentially the \emph{only} permutation groups that prevent super-polynomial quantum speedups.
We prove this by fully characterizing the primitive
permutation groups that allow super-polynomial quantum
speedups.

In contrast, in the adjacency list model for bounded-degree graphs---where graph symmetry is manifested differently---we exhibit a property testing problem that shows an exponential quantum speedup. These results resolve open questions posed by Ambainis, Childs, and Liu (2010) and Montanaro and de Wolf (2013).
\end{abstract}

\clearpage
{\small\tableofcontents}
\clearpage

\section{Introduction}

One of the most fundamental problems in the field of quantum
computing is the question of when quantum algorithms
substantially outperform classical ones.
While polynomial quantum speedups are known in
many settings, super-polynomial quantum speedups are known
(or even merely conjectured) for only a few select problems. Crucially, exponential
quantum speedups only occur for certain ``structured'' problems such as period-finding (used in Shor's factoring algorithm
\cite{Shor94}) and Simon's problem \cite{Sim97}, in which the input
is known in advance to have a highly restricted form. In contrast,
for ``unstructured'' problems such as black-box search or
$\mathsf{NP}$-complete problems,
only polynomial speedups are known (and in some models,
it can be formally shown that only polynomial speedups are possible).

In this work, we are interested in formalizing and characterizing
the structure necessary for fast quantum algorithms. In particular,
we study the \emph{types of symmetries} a function can have
while still exhibiting super-polynomial quantum speedups.

\subsection{Prior work}

Despite the strong intuition in the field that structure is necessary
for super-polynomial quantum speedups, only a handful of works have attempted
to formalize this and characterize the required structure.
All of them study the problem in the query complexity (black-box) model
of quantum computation, which is a natural framework in which
both period-finding and Simon's problem can be formally shown to
give exponential quantum speedups (see \cite{BdW02} for a survey
of query complexity,
or \cite{Cle04} for a formalization of period-finding specifically).

In the query complexity model, the goal is to compute
a Boolean function $f\colon \Sigma^n\to\B$ using as few queries to
the bits of the input $x\in\Sigma^n$ as possible, where $\Sigma$ is some
finite alphabet. Each query specifies an index $i\in[n] := \{1,2,\ldots,n\}$
and receives the response $x_i\in\Sigma$. A query algorithm,
which may depend on $f$ but not on $x$,
must output $f(x)$ (with bounded error in the worst-case) after as few
queries as possible. Quantum query algorithms are allowed
to make queries in superposition, and we are interested in
how much advantage this gives them over randomized classical
algorithms (for formal definitions of these notions, see \cite{BdW02}).

Beals, Buhrman, Cleve, Mosca, and de Wolf \cite{BBC+01} showed
that all \emph{total Boolean functions} $f\colon \Sigma^n\to\B$
have a polynomial relationship
between their classical and quantum query complexities
(which we denote $\R(f)$ and $\Q(f)$, respectively).
This means that super-polynomial speedups are not possible in
query complexity unless we impose a promise on the input: that is,
unless we define $f\colon P\to\B$ with $P\subset\Sigma^n$, and allow
an algorithm computing $f$ to behave arbitrarily on inputs
outside of the promise set $P$. For such promise problems
(also called partial functions), provable exponential quantum speedups
are known. This is the setting in which Simon's problem and
period-finding reside.

The question, then, is what we can say about the structure
necessary for a partial Boolean function $f$ to exhibit
a super-polynomial quantum speedup. Towards this end,
Aaronson and Ambainis \cite{AA14} showed
that \emph{symmetric} functions do not allow super-polynomial quantum
speedups, even with a promise. Chailloux \cite{Cha18} improved
this result by reducing the degree of the polynomial
relationship between randomized and quantum algorithms for symmetric
functions, and removing a technical requirement on the symmetry
of those functions.\footnote{
Aaronson and Ambainis required the function
to be symmetric both under permuting the $n$ bits of the input,
and under permuting the alphabet symbols in $\Sigma$;
Chailloux showed that the latter is not necessary.}

Other work attempted to characterize the structure necessary for
quantum speedups in alternative ways. Ben-David \cite{Ben16} showed that certain types of symmetric promises do not admit any function with a super-polynomial quantum speedup, a generalization of \cite{BBC+01}
(who showed this when the promise set is $\Sigma^n$).
Aaronson and Ben-David \cite{AB16} showed that small promise sets, which contain
only $\poly(n)$ inputs out of $|\Sigma|^n$, also do not admit
functions separating quantum and classical algorithms by
more than a polynomial factor.

A class of problems with significant symmetry, though much less than full permutation symmetry, is the class of graph properties. For such problems, the input describes a graph, and the output depends only on the isomorphism class of that graph. Thus the vertices can be permuted arbitrarily, but such a permutation induces a structured permutation on the edges, about which queries provide information.

The setting of graph property \emph{testing} provides
a natural class of partial graph properties. Here we are promised that the input graph either has a property, or is $\epsilon$-far from having the property, meaning that we must change at least an $\epsilon$ fraction of the edges to make the property hold.
Graph property testing has been extensively studied since its introduction by Goldreich, Goldwasser, and Ron \cite{GGR98}.

The behavior of classical graph property testers can differ substantially depending on the model in which the input graph is specified. For example, in the adjacency matrix model, Alon and Krivelevich \cite{AK02} proved that bipartiteness can be tested in $\tilde{O}(1/\epsilon^2)$ queries, which is surprisingly independent of the  size of the input graph. In contrast, in the adjacency list model for bounded-degree graphs, Goldreich and Ron \cite{GR97} proved that $\Omega(\sqrt{n})$ queries are needed to test bipartiteness of $n$-vertex graphs.

Quantum algorithms for testing properties of bounded-degree graphs in the adjacency list model were studied by Ambainis, Childs, and Liu \cite{ACL11}. They gave upper bounds of $\tilde{O}(n^{1/3})$ for testing bipartiteness and expansion, demonstrating polynomial quantum speedups. Furthermore, they showed that at least $\Omega(n^{1/4})$ quantum queries are required to test expansion, ruling out the possibility of an exponential quantum speedup. This work naturally raises the question (also highlighted by Montanaro and de Wolf \cite{MdW13}) of whether there can ever be exponential quantum speedup for graph property testing, or for graph properties more generally.

\subsection{Our contributions}

In this work, we extend the results of Aaronson-Ambainis and
Chailloux to other symmetry groups. We characterize general symmetries of functions as follows.

\begin{definition}\label{def:symmetric_function}
Let $f\colon P\to\B$ be a function with $P\subseteq\Sigma^n$,
where $\Sigma$ is a finite alphabet and $n\in\bN$. We say
that $f$ is \emph{symmetric} with respect to a permutation group $G$
acting on domain $[n]$ if for all $x\in P$ and all $\pi\in G$,
the string $x\circ\pi$ defined by $(x\circ\pi)_i\coloneqq x_{\pi(i)}$
satisfies $x\circ\pi\in P$ and $f(x\circ\pi)=f(x)$.
\end{definition}

The case where $G=S_n$, the fully-symmetric
permutation group, is the scenario handled in Chailloux's work \cite{Cha18}:
he showed that $\R(f)=O(\Q(f)^3)$ if $f$ is symmetric under $S_n$.
Aaronson and Ambainis \cite{AA14}
required an even stronger symmetry property.
We note that when $\Sigma$ is large, say $|\Sigma|=n$ or larger,
the class of functions symmetric under $S_n$ is already highly
nontrivial: among others,
it includes functions such as $\textsc{Collision}$,
an important function whose quantum query complexity
was established in \cite{AS04}; $k$-$\textsc{Sum}$,
whose quantum query complexity required the negative-weight
adversary to establish \cite{BS13}; and
$k$-$\textsc{Distinctness}$, whose quantum query complexity is
still open \cite{BKT18}. Additionally, computational
geometry functions such as $\textsc{ClosestPair}$
(a function studied in recent work by Aaronson,
Chia, Lin, Wang, and Zhang \cite{ACL+19}) are typically symmetric
under $S_n$, as the points are usually represented as alphabet
symbols. If we round the alphabet to be finite, the $G=S_n$
case already shows that no computational geometry problem
of this form can have super-cubic quantum speedups.

In this work, we examine what happens when we
relax the full symmetry $S_n$ to smaller symmetry groups $G$.
We introduce some tools for showing that particular classes
of permutation groups $G$ do not allow super-polynomial quantum speedups;
that is, we provide tools for showing that every $f$ symmetric
with respect to $G$ satisfies $\Q(f)=\R(f)^{\Omega(1)}$.
Our first main result is the following theorem,
in which $G$ is the \emph{graph symmetry}:
the permutation group acting on strings of length $\binom{n}{2}$
(which represent the possible edges of a graph), which includes
all permutations of the edges which are induced by one of the $n!$
relabelings of the $n$ vertices. Functions that take
the adjacency matrix of a graph as input, and whose output
depends only on the isomorphism class of the graph (not on the labeling of its vertices),
are symmetric with respect to the graph symmetry $G$.

\begin{theorem}[{Informal version of \cor{formal_graph}}]\label{thm:main_graph}
Any Boolean function $f$ defined on the adjacency matrix of
a graph (and symmetric with respect to renaming the vertices of the
graph) has
$\R(f) = O(\Q(f)^6)$. This holds even if $f$ is a partial function.
\end{theorem}

This theorem holds even when the alphabet of $f$ is non-Boolean.
We note that this is a strict generalization of the result
of Aaronson and Ambainis \cite{AA14}, since any fully-symmetric
function is necessarily symmetric under the graph
symmetry as well. It is also a generalization of \cite{Cha18},
except that our polynomial degree (power $6$) is larger than
the power $3$ of Chailloux. Our results also extend to other types of graph symmetries, including hypergraph symmetries and bipartite graph symmetries.

Next, we extend \thm{main_graph} to give a more general characterization of which classes of symmetries are inconsistent with super-polynomial quantum speedups.
Our main result in this regard is a dichotomy theorem for \emph{primitive} permutation groups. Primitive permutation groups are sometimes described as the ``building blocks'' of permutation groups, because arbitrary permutation groups can always be decomposed into primitive groups (in a certain formal sense). We give a complete classification of which primitive groups $G$ allow super-polynomial quantum speedups and which do not, in terms of the \emph{minimal base size} $b(G)$. The minimal base size is a key quantity in computational group theory that roughly captures the size of a permutation group $G$ relative to the number of points it acts on. Thus, the following theorem amounts to showing that symmetries of ``large'' primitive permutation groups do not allow large quantum speedups, while symmetries of ``small'' primitive permutation groups do.

\begin{theorem}[{Informal version of \cor{primitive_dichotomy}}]\label{thm:primitive_dichotomy_informal}
Let $G$ be a primitive permutation group acting on $[n]$, and let $b(G)$ denote the size of a minimal base for $G$. If $b(G) = n^{\Omega(1)}$, then $R(f) = Q(f)^{O(1)}$ for every $f$ that is symmetric under $G$. Otherwise, if $b(G) = n^{o(1)}$, then there exists a partial function $f$ that is symmetric under $G$ such that $R(f) = Q(f)^{\omega(1)}$.
\end{theorem}

By decomposing an arbitrary permutation group into primitive groups, we can extend the above theorem to show that if a permutation group $G$ does not allow super-polynomial speedups, then it must be constructed out of a constant number of primitive groups $H$ that satisfy $b(H) = n^{\Omega(1)}$, where $H$ acts on $n$ points. In the other direction, if any of the primitive factors $H$ of $G$ satisfy $b(H) = n^{o(1)}$, or if $G$ contains more than a constant number of primitive factors, then we exhibit a function that is symmetric under $G$ with a super-polynomial quantum speedup.

Remarkably, the proof of \thm{primitive_dichotomy_informal} also includes a strong characterization of what the primitive permutation groups $G$ that satisfy $b(G) = n^{\Omega(1)}$ look like. Indeed, we show that as a consequence of the classification of finite simple groups, all such groups essentially look like hypergraph symmetries or minor extensions thereof. Thus, in some sense, permutation groups built out of hypergraph symmetries are the \emph{only} permutation groups that are inconsistent with super-polynomial speedups. We consider this one of the most surprising consequences of our work; a priori, one could expect there to be many different kinds of symmetries that disallow super-polynomial speedups.

Finally, we return to the topic of graph properties.
While the aforementioned results show that no exponential speedup is possible for graph properties in the adjacency \emph{matrix} model, they do not resolve the original open question of Ambainis, Childs, and Liu \cite{ACL11}, which specifically addressed graph property testing in the adjacency \emph{list} model. Graph symmetries in this model manifest themselves in a different way that is not captured by \defn{symmetric_function}, so \thm{main_graph} does not apply.\footnote{Given a graph $x\colon [n]\times[d]\rightarrow[n]\cup\{*\}$ in the adjacency list model with $n$ vertices and bounded degree $d$, its isomorphism class essentially consists of graphs of the form $\pi^{-1}\circ x\circ (\pi\times \text{id}_{[d]})$ where $\pi\in S_n$ is a relabeling of vertices and $\pi^{-1}$ is defined as the inverse of $\pi$ but also maps $*$ to $*$. A function $f$ is a graph property in this model if and only if its domain $P$ is a union of such isomorphism classes and $f$ is constant on each isomorphism class.} In this model, we show the following,

\begin{theorem}[{Informal version of \thm{adjacency_list}}]\label{thm:main_graph_adjacency}
There exists a graph property testing problem in the adjacency list model for which there is an exponential quantum speedup.
\end{theorem}

This is in stark contrast to the situation in the adjacency matrix case.
Together, \thm{main_graph} and \thm{main_graph_adjacency} fully settle the open question about the possibility of exponential quantum speedup for graph properties, showing that its answer is highly model-dependent: exponential speedup is impossible in the adjacency matrix model, but is possible in the adjacency list model, even in the restrictive setting of graph property \emph{testing}.

\subsection{Techniques}

Here we briefly discuss the techniques used to prove the main results.
\subsubsection{Well-shuffling symmetries do not allow speedups}

Our analysis begins with a basic observation of Chailloux \cite{Cha18}.
Suppose that $f$ is symmetric under the full symmetric group $S_n$ acting on $[n]$.
Zhandry \cite{Zha13} showed that distinguishing
a random element of $S_n$ from a random small-range
function $\alpha\colon [n]\to[n]$ with $|\alpha([n])|=r$ requires
$\Omega(r^{1/3})$ quantum queries.\footnote{Actually,
we show that a version of Zhandry's result that is sufficient
for our purposes follows easily from the collision lower bound,
so his techniques are not necessary for our results.}
Now, if $Q$ is a quantum
algorithm solving $f$ using $T$ queries, then $Q$ also outputs $f(x)$
on input $x\circ \pi$ (the input $x$ with bits shuffled according to
$\pi$) for any $\pi\in S_n$, since $f$ is symmetric under $S_n$.
In particular, $Q(x\circ\pi)$ for a random $\pi\in S_n$ still
outputs $f(x)$. However, the $T$-query algorithm $Q$ cannot
distinguish a random $\pi\in S_n$ from a random function
$\alpha\colon [n]\to[n]$ with range $r\approx T^3$. Therefore $Q$
must output $f(x)$ with at most constant error even when run on
$x\circ\alpha$ for a random small-range function $\alpha$.
This property can be used to simulate $Q$ classically: a classical
algorithm $R$ can simply sample a small-range function $\alpha$,
explicitly query the entire string $x\circ \alpha$ (using only $O(T^3)$ queries since $\alpha$ has range $O(T^3)$),
and then simulate $Q$ on the string $x\circ\alpha$. This is
an $O(T^3)$-query classical algorithm for computing $f$,
created from a $T$-query quantum algorithm for $f$.

The above trick can be generalized from the fully-symmetric group $S_n$ to other permutation groups $G$, provided we can show
that it is hard for a $T$-query quantum algorithm to distinguish $G$ from a function with range $\poly(T)$.
The question of whether there exists an arbitrary symmetric function
$f$ with a quantum speedup is therefore reduced to the question
of whether the concrete task of distinguishing $G$ from
the set of all small-range functions can be done quickly
using a quantum algorithm. That is, if $D_{n,r}$ is the set
of all strings in $[n]^n$ that use only $r$ unique symbols,
then we care about the quantum query cost of distinguishing
$D_{n,r}$ from $G$; if this cost is $r^{\Omega(1)}$, then
no function that is symmetric under $G$ can exhibit a super-polynomial
quantum speedup. In this case, we call $G$ \emph{well-shuffling}.

We show that the well-shuffling property is preserved under various
operations one might perform on a permutation group, which allows us to prove
that many permutation groups are well-shuffling
simply by reduction to $S_n$. As an example, for undirected graph properties, the relevant transformation is from a permutation group $G$ acting on a set $[n]$ to the induced action $H$ on the $\binom{n}{2}$ unordered pairs of $[n]$. We then show that if $G$ is well-shuffling, then $H$ is also well-shuffling, because the problem of distinguishing $G$ from a small range function has a reduction to the problem of distinguishing $H$ from a small range function. More generally, this reduction works for ordered pairs or for tuples of any constant length.

\subsubsection{Classification of symmetries with and without speedups}

To characterize which permutation groups allow super-polynomial quantum speedups, we first show that large quantum speedups exist for any sufficiently ``small'' permutation group. In particular, if the minimal base size $b(G)$ of a permutation group $G$ acting on $n$ points satisfies $b(G) = n^{o(1)}$, then we exhibit a partial function that is symmetric under $G$ with a super-polynomial quantum speedup. We construct such a function by starting with an arbitrary promise problem that has a sufficiently large separation between randomized and quantum query complexities (e.g., Simon's problem \cite{Sim97} or Forrelation \cite{AA15}). We then modify it to be symmetric under $G$ in a way that preserves the super-polynomial separation between randomized and quantum query complexities.

One might wonder whether a sort of converse holds, i.e., whether $b(G) = n^{\Omega(1)}$ implies that $G$ is well-shuffling, and therefore that all partial functions symmetric under $G$ admit at most a polynomial quantum speedup. We show this for primitive permutation groups $G$, essentially because the structure of ``large'' primitive groups is very well-understood. We base our proof on a theorem due to Liebeck \cite{Lie84} that characterizes minimal base sizes of primitive groups via the classification of finite simple groups. Liebeck's theorem allows us to show that primitive groups $G$ with $b(G) = n^{\Omega(1)}$ can all be constructed out of the symmetric group via transformations that preserve the well-shuffling property. As a result, the quantity $b(G)$ completely characterizes whether a primitive permutation group admits a super-polynomial speedup. Moreover, the transformations involved are the same transformations we use to show that hypergraph symmetries are well-shuffling, so primitive groups that satisfy $b(G) = n^{\Omega(1)}$ all essentially look like hypergraph symmetries.

To go beyond primitive groups, we use the fact that an arbitrary permutation group can be decomposed into transitive groups (by looking at its orbits), which can then be decomposed into primitive groups (by looking at its nontrivial block systems). We show that for a group $G$, if any of the primitive factors in such a decomposition of $G$ are consistent with super-polynomial speedups, or if the number of primitive factors is $\omega(1)$, then there exists a partial function that is symmetric under $G$ with a super-polynomial quantum speedup.
These functions are generally easy to construct, though in one case that involves tree-like symmetries, we make essential use of the $\kFaultDirectTrees$ problem of Zhan, Kimmel, and Hassidim \cite{ZKH12}: a partial function that has a super-polynomial quantum speedup, constructed by adding a promise to a Boolean evaluation tree in a way that preserves the symmetries of the tree. Interestingly, these results also show that the quantity $b(G)$ does \emph{not} determine the well-shuffling property in general, because there exist permutation groups $G$ acting on $n$ points with $b(G) = \Omega(n)$ that are consistent with exponential quantum speedups. Nevertheless, we can at least say that if a permutation group $G$ is inconsistent with super-polynomial quantum speedups, then it must be decomposable into $O(1)$ well-shuffling primitive groups.

\subsubsection{Graph property testing in the adjacency list model}

Finally, we show that that the possibility of exponential quantum speedups for graph properties depends strongly on the input model: while the adjacency matrix model allows at most polynomial speedups, we give an example of an exponential speedup in the adjacency list model for bounded-degree graphs. In particular, we show that such a separation holds even in the more restrictive setting of property testing, where ``no'' instances are promised to be far from ``yes'' instances.

The main idea for this separation comes from work by Childs et al.~\cite{CCD+03}, demonstrating an exponential algorithmic speedup by quantum walk. More specifically, they design a ``welded trees'' graph that consists of two depth-$k$ binary trees welded together by an alternating cycle on the leaves of the two trees, with the two degree-$2$ vertices (roots) referred to as ``\entr'' and ``\exit''. The main result of \cite{CCD+03} is that a quantum computer, given the \entr, can find the \exit\ in time $\poly(k)$, whereas a classical computer needs exponential time $2^{\Omega(k)}$ to find the \exit. However, it is not immediately clear how to lift this separation to the realm of graph property testing: it is not even a graph property (since the \entr\ is given and the \exit\ can be labeled differently in isomorphic graphs), and it is still more challenging to find a related classically-hard property \emph{testing} problem.

Our approach is to leverage the quantum advantage for finding \entr-\exit\ pairs to test whether the given graph has a certain form derived from welded trees. The difficulty is that it is unclear how finding \entr-\exit\ pairs helps a quantum computer test the structure of a welded trees. One of our main observations is that once the vertices in the weld of the binary trees are marked, it becomes easy to test the entire structure: we can separately test the two binary trees, and then test the weld itself. However, marking the weld vertices also enables a classical computer to find \entr-\exit\ pairs, so it seems that this approach has no quantum advantage.

Our resolution is to use many copies of the welded trees structure and randomly assign a bit to every root (degree-$2$) vertex. Then we mark vertices by connecting every non-root vertex $v$ to a root $r$ via an ``advice edge'', such that
\begin{enumerate}[noitemsep]
    \item if $v$ is in the weld, then the sum of the bits assigned to $r$ and its pair is odd, and
    \item if $v$ is not in the weld, then the sum of the bits assigned to $r$ and its pair is even.
\end{enumerate}
To keep the degree bounded, instead of directly connecting advice edges to the roots, we attach a binary tree ``antenna'' to the roots and connect the advice edges there, as shown in \cref{fig:modified_glued_tree,fig:schematic_full_graph}.

Still, it might seem that this construction has a fatal chicken-or-egg problem: we need to mark the weld vertices to test the structure, but then marking the weld vertices enables also a classical computer to find the root pairs, leaving no quantum advantage. However, this is a feature, not a bug. The chicken-or-egg problem has two fixed points, with neither a chicken nor an egg (corresponding to classical strategies for our problem) or with both a chicken and an egg (corresponding to the quantum case). The key point is that the marking of the weld vertices can only be read efficiently by a quantum computer.

To show that this property is hard to test classically, we show that it is exponentially difficult to distinguish a randomly selected graph of the form described above, from a graph that looks like the above, except the binary trees are self-welded, i.e., the leaves of the binary trees are connected with a random cycle only containing vertices from the same binary tree (this then disconnects the root pairs, so the advice edges no longer have the original role and are chosen arbitrarily).
The intuition behind our classical lower bound proof is the same as that behind the lower bound proof of~\cite{CCD+03}: given a root of a (self)-welded trees graph, the graph appears to be an exponentially deep binary tree for a (random) classical walker who can only explore the graph locally. Finally, we show that a random welded trees graph is $\Omega(1)$-far from a random self-welded trees graph by observing that the welded trees graph is bipartite, whereas the self-welded trees graph is $\Omega(1)$-far from being bipartite with very high probability.

\subsection{Organization}

The remainder of the paper is organized as follows.

In \sec{prelim} we introduce some basic notions from query complexity, the theory of permutation groups, and symmetric functions.
\sec{wellshuffle} introduces the notion of well-shuffling permutation groups and shows that they do not allow super-polynomial quantum speedups.
In \sec{wsgroups} we present techniques for showing that permutation groups are well-shuffling, establishing in particular that hypergraph symmetries cannot have super-polynomial quantum speedups.
In \sec{wsclassify} we classify well-shuffling permutation groups, showing that hypergraph symmetries are essentially the only groups that are inconsistent with super-polynomial speedups.
\sec{testing} presents an exponential separation between classical and quantum property testing in the adjacency list model for bounded-degree graphs.
Finally, in \sec{open} we briefly discuss some open problems.
The appendices present two technical results: \app{minimax} gives a proof of the quantum minimax lemma, and \app{wreath_tree} establishes a super-polynomial quantum speedup for deep wreath product symmetries.

\section{Preliminaries}
\label{sec:prelim}
\subsection{Query complexity}

We begin by introducing some standard notation from query complexity.
A \emph{Boolean function} is a $\B$-valued function $f$ on
strings of length $n\in\bN$. We use $\Dom(f)$
to denote the domain of $f$, and we always have
$\Dom(f)\subseteq\Sigma^n$ where $\Sigma$ is a finite alphabet.
The function $f$ is called \emph{total} if $\Dom(f)=\Sigma^n$;
otherwise it is called \emph{partial}.

For a (possibly partial) Boolean function $f$, we use
$\R_\epsilon(f)$ to denote its \emph{randomized query complexity
with error $\epsilon$},
as defined in \cite{BdW02}. This is the minimum number of queries
required in the worst case by a randomized algorithm that computes
$f$ with worst-case error $\epsilon$. We use $\Q_\epsilon(f)$ to denote the
\emph{quantum query complexity with error $\epsilon$} of $f$,
also defined in \cite{BdW02}.
This is the minimum number of queries required in the worst case
by a randomized algorithm that computes $f$ with worst-case error
$\epsilon$. When $\epsilon=1/3$, we omit it and simply write
$\R(f)$ and $\Q(f)$.

An important tool for lower bounding query complexity
is the minimax theorem, the original version of which was given
by Yao for zero-error (Las Vegas) randomized algorithms \cite{Yao77}.
Here we use a bounded-error, quantum version
of the minimax theorem. Bounded-error versions of the minimax
theorem can be shown using linear programming duality
(see also \cite{Ver98}
who proved a minimax theorem in the setting where both the error
and the expected query complexity are measured against the same hard
distribution).
A similar technique works for quantum query complexity;
this result is folklore, and we prove it in
\app{minimax}.

\begin{restatable}[Minimax for bounded-error quantum algorithms]{lemma}{minimax}
\label{lem:minimax}
Let $f$ be a (possibly partial) Boolean function with
$\Dom(f)\subseteq\Sigma^n$, and let $\epsilon\in(0,1/2)$.
Then there is a distribution $\mu$ supported on $\Dom(f)$
which is \emph{hard} for $f$ in the following sense:
any quantum algorithm using fewer than $\Q_\epsilon(f)$
quantum queries for computing $f$ must have average error more than
$\epsilon$ on inputs sampled from $\mu$.
\end{restatable}

Note that achieving average error $\epsilon$ against a known distribution
$\mu$ is always easier than achieving worst-case error $\epsilon$; the minimax
theorem says that there is a hard distribution against which achieving average
error $\epsilon$ is just as hard as achieving worst-case error $\epsilon$.

\subsection{Permutation groups}

We review some basic definitions about permutation groups, following Hulpke~\cite{Hul10}.

\begin{definition}[Permutation group]
A \emph{permutation group} is a pair $(D,G)$ where $D$ is a set
and $G$ is a set of bijections $\pi\colon D\to D$, such that
$G$ forms a group under composition (i.e., $G$ contains the identity
function and is closed under composition and inverse of the bijections).
In this case, $G$ is said to \emph{act} on $D$. We will often denote a permutation group simply by $G$, with the domain $D$ being implicit.
\end{definition}

In other words, a permutation group is simply a set of permutations
of a domain $D$ which is closed under composition and inverse.
In this work we will generally take $D=[n]$, where $[n]$ denotes
the set $\{1,2,\dots,n\}$ for $n\in\bN$. The set $[n]$ will
represent the indices of an input string, or equivalently,
the queries an algorithm is allowed to make.

We define orbits, transitivity, and stabilizers of permutation groups, all of which are standard.

\begin{definition}[Orbit]
Let $G$ be a permutation group on domain $D$, and let $i\in D$.
Then the \emph{orbit} of $i$ is the set $\{\pi(i):\pi\in G\}$.
A subset of $D$ is an orbit of $G$ if it is the orbit
of some $i\in D$ with respect to $G$.
\end{definition}

\begin{definition}[Transitivity]
We say that a permutation group $G$ on domain $D$ is \emph{$k$-transitive} if
for all distinct $i_1,i_2,\dots,i_k\in D$ and distinct
$j_1,j_2,\dots,j_k\in D$, there exists some $\pi\in G$ such that
$\pi(i_t)=j_t$ for all $t=1,2,\dots,k$. A group that is $1$-transitive is called simply \emph{transitive}.
\end{definition}

\begin{definition}[Stabilizer]
Let $G$ be a permutation group on domain $D$. The \emph{pointwise stabilizer} of a set $S \subseteq D$ is the subgroup $\Stab_G(S) := \{\pi \in G : \pi(i) = i \  \forall i \in S\}$. The \emph{setwise stabilizer} of a set $S \subseteq D$ is the subgroup $\Stab_G(\{S\}) := \{\pi \in G : \pi(i) \in S \  \forall i \in S\}$.
\end{definition}

We use some facts about the structure of permutation groups. For this, we need some additional definitions, starting with the notions of primitive groups and wreath products.

\begin{definition}[Block system]
Let $G$ be a permutation group on domain $D$. A partition $\mathcal{B} = \{B_1,B_2,\ldots,B_k\}$ of $D$ is called a \emph{block system} for $G$ if $\mathcal{B}$ is $G$-invariant. Formally, this means that $\pi(B_i) \in \mathcal{B}$ for every $\pi \in G$ and $i \in [k]$, where $\pi(B_i) := \{\pi(x) : x \in B_i\}$. The sets $B_i$ are called \emph{blocks} for $G$.
\end{definition}

As an example, consider the action of the group of linear transformations on the nonzero elements of a vector space. The set of lines through the origin (i.e., nonzero scalar multiples of a nonzero vector) form a block system.

It follows from this definition that if $\mathcal{B}$ is a block system for $G$, then for every block $B_i$ and permutation $\pi \in G$, either $\pi(B_i) = B_i$ or $\pi(B_i) \cap B_i = \emptyset$. Moreover, if $G$ acts transitively, then all blocks have the same size.

\begin{definition}[Primitive group]
Let $G$ be a transitive permutation group on domain $D$. $G$ is called \emph{primitive} if the only block systems for $G$ are $\{D\}$ and the partition of $D$ into singletons. A permutation group that is not primitive is called \emph{imprimitive}.
\end{definition}

Primitive groups are sometimes described as the ``building blocks'' of permutation groups, for reasons we will see later.

\begin{definition}
Let $G$ be an abstract group. The \emph{direct power} of $G$ with exponent $m$, denoted $G^{\times m}$, is the direct product of $m$ copies of $G$:
$$G^{\times m} := \underbrace{G \times \cdots \times G}_{m\ \mathrm{times}}.$$
\end{definition}

If $G$ is a permutation group acting on $[n]$, there are two natural ways to construct $G^{\times m}$ as a permutation group. Let $\pi = (\pi_1,\pi_2,\ldots,\pi_m) \in G^{\times m}$. The first action is on $[m] \times [n]$ where we have $\pi(i, j) = \pi_i(j)$. The second action is on $[n]^m$ where we have $\pi(i_1,i_2,\ldots,i_m) = (\pi_1(i_1), \pi_2(i_2),\ldots,\pi_m(i_m))$. These two actions of $G^{\times m}$ give rise to two actions of the \emph{wreath product} of two groups:

\begin{definition}[Wreath product]
\label{def:wreath_product}
Let $G$ and $H$ be permutation groups acting on $[m]$ and $[n]$, respectively. The \emph{wreath product} of $G$ and $H$, denoted $G \wr H$, is one of two permutation groups:
\begin{enumerate}[label=(\roman*)]
\item The \emph{imprimitive action} is the action on $[m] \times [n]$ generated by the natural action of $H^{\times m}$, along with permutations of the form $(i, j) \to (\sigma(i), j)$ for every $\sigma \in G$.
\item The \emph{product action} is the action on $[n]^m$ generated by the natural action of $H^{\times m}$, along with permutations of the form $(i_1, i_2,\ldots,i_m) \to (i_{\sigma(1)},i_{\sigma(2)},\ldots,i_{\sigma(m)})$ for every $\sigma \in G$.
\end{enumerate}
\end{definition}

\begin{fact}
\label{fct:wreath_product_order}
Let $G$ and $H$ be permutation groups acting on $[m]$ and $[n]$, respectively. If $n > 1$, then the imprimitive action and the primitive action of $G \wr H$ are isomorphic as groups, and both have order $|G| \cdot |H|^m$.
\end{fact}

If it is clear from context which of the two actions we are referring to, then we may sometimes simply write $G \wr H$ without specifying the type of action.

This definition differs from the standard definition of the wreath product, which is an abstract group that can be defined even if $H$ is an abstract group. However, the definitions agree when $H$ is a permutation group, which is the only case we consider.

The wreath product imprimitive action carries a useful intuitive interpretation: if $g$ and $h$ are functions that are symmetric under permutation groups $G$ and $H$, respectively, then $G \wr H$ is the (visible)\footnote{In rare cases, depending on $g$ and $h$, there can be more symmetries. For example, let $g = h = \mathsf{OR}_n$, which both have symmetry group $S_n$. The wreath product $S_n \wr S_n$ is \emph{not} the same as $S_{n^2}$, which is the group of symmetries of $g \circ h = \mathsf{OR}_{n^2}$. However, we can say that $g \circ h$ is guaranteed to be \emph{at least} symmetric under $G \wr H$ for any functions $g$ and $h$.} group of symmetries of the block-composed function $g \circ h$. Thus, the wreath product imprimitive action is associative.

Note that the wreath product is not commutative. This can be confusing because some authors instead denote this wreath product as $H \wr G$. We prefer the notation used here because of the correspondence to block composition of functions.

\begin{definition}[Base]
Let $G$ be a permutation group on domain $D$. A \emph{base} for $G$ is a set $S \subseteq D$ such that the pointwise stabilizer $\Stab_G(S)$ is the trivial group. The \emph{minimal base size} for $G$, denoted $b(G)$, is the size of the smallest base for $G$.
\end{definition}

Bases of permutation groups are important in computational group theory for the following reason: if $S$ is a base for $G$, then a permutation $\pi \in G$, viewed as a function $D \to D$, is uniquely determined by its restriction to $S$. This can be useful particularly when $S$ is small compared to $D$.

In many applications, the minimal base size is a useful proxy for the size of a permutation group. This is quantified by the following well-known lemma, which we prove for completeness:

\begin{lemma}\label{lem:base_bounds}
Let $G$ be a permutation group on $[n]$. Then the order of the group $|G|$ and the minimal base size $b(G)$ satisfy $2^{b(G)} \le |G| \le n^{b(G)}$. Equivalently, $\log_n(|G|) \le b(G) \le \log_2(|G|)$.
\end{lemma}

\begin{proof}
Without loss of generality, identify a minimal base for $G$ with $[k]$ where $k = b(G)$. Let $H_i = \Stab_G([i])$ be the pointwise stabilizer of the first $i$ points of the base, with the understanding that $H_0 = G$. We have the inclusion
$$1 = H_k \subseteq H_{k-1} \subseteq \cdots \subseteq H_1 \subseteq H_0 = G.$$
Let $O_i \subseteq [n]$ denote the orbit of $i$ in $H_{i-1}$. By the orbit-stabilizer theorem, we have $|H_{i-1}| = |H_i| \cdot |O_i|$. Thus, we have $|G| = \prod_{i=1}^k |O_i|$. We also know that $2 \le |O_i| \le n$, where the upper bound is trivial, and the lower bound holds because $[k]$ was assumed to be a minimal base. We get the desired bound on $|G|$ by bounding each term in the product.
\end{proof}

\subsection{Symmetric functions}

We introduce some notation that is used throughout the
paper to talk about symmetric functions.

\begin{definition}[Permuting strings]
Let $\pi$ be a permutation on $[n]$, and let $x\in\B^n$.
We write $x\circ \pi$ to denote the string whose characters have been
permuted by $\pi$; that is, $(x\circ\pi)_i\coloneqq x_{\pi(i)}$.
More generally, $x\circ \pi$ is similarly defined when $\pi$
is merely a function $[n]\to[n]$ rather than a permutation.
\end{definition}

Note that if we view a string $x\in\Sigma^n$ as a function
$[n]\to\Sigma$ with $x(i)\coloneqq x_i$, then $x\circ \pi$ is simply
the usual function composition of $x$ and $\pi$.
This notation allows us to easily define symmetric functions.

\begin{definition}[Symmetric function]
Let $G$ be a permutation group on $[n]$, and let $f$ be a (possibly partial)
Boolean function with $\Dom(f)\subseteq\Sigma^n$.
We say $f$ is \emph{symmetric under} $G$
if for all $x\in\Dom(f)$ and all $\pi\in G$
we have $x\circ\pi\in\Dom(f)$ and $f(x\circ\pi)=f(x)$.
\end{definition}

For asymptotic bounds such as $\Q(f)=\R(f)^{\Omega(1)}$
to be well-defined, we need to talk about
\emph{classes} of functions rather than individual functions.
To do that, we need to talk about classes of permutation groups.
We introduce the following definition, which defines, for a class
of permutation groups $\mathcal{G}$, the set of all functions symmetric
under some group in $\mathcal{G}$. We denote this set by $F(\mathcal{G})$.

\begin{definition}[Class of symmetric functions]
Let $\mathcal{G}=\{G_i\}_{i\in I}$ be a (possibly infinite)
set of finite permutation groups, with $G_i$ acting on $[n_i]$
for each $i\in I$. Here $I$ is an arbitrary index set and $n_i\in\bN$
for all $i\in I$. Then define $F(\mathcal{G})$ to be the
set of all (possibly partial) Boolean functions that are symmetric
under some $G_i$. That is, we have $f\in F(\mathcal{G})$
if and only if
$f\colon\Dom(f)\to\B$ is a function with $\Dom(f)\subseteq[m]^n$
for some $n,m\in\bN$, and $f$ is symmetric under $G_i$ for some
$i\in I$ such that $n_i=n$.
(Here $[m]$ represents the alphabet $\Sigma$.)
\end{definition}

\section{Well-shuffling permutation groups}
\label{sec:wellshuffle}

In this section we first define the notion of a
well-shuffling class of permutation groups,
which is a class $\mathcal{G}$ of permutation groups $G$ that are hard
to distinguish from the set of small-range functions via a quantum query algorithm.
We then show that a well-shuffling class of permutation groups does not allow
super-polynomial quantum speedups. This result (\thm{shuffle})
converts the task of showing permutation groups do not allow quantum speedups
into the task of showing those permutation groups are well-shuffling,
a much simpler objective.

We start by defining the set of small-range strings $D_{n,r}$.

\begin{definition}[Small-range strings]
For $n,r\in\bN$, let $D_{n,r}$ be the set of all strings $\alpha$ in $[n]^n$
for which the number of unique alphabet symbols in $\alpha$ is at most $r$.
\end{definition}

We identify a string $\alpha\in[n]^n$ with a function $[n]\to[n]$.
Then $D_{n,r}$ is the set of all functions $[n]\to[n]$ with range size
at most $r$. Next, we define $\cost(G,r)$ as the quantum query complexity
of distinguishing $G$ from $D_{n,r}$ (where $G$ is a permutation group acting
on $[n]$).

\begin{definition}[Cost]
Identify a permutation on $[n]$ with a string in $[n]^n$ in which
each alphabet symbol occurs exactly once. Then a permutation group $G$
on $[n]$ corresponds to a subset of $[n]^n$. For $r<n$, let
$\cost_\epsilon(G,r)$ be the minimum number of quantum queries
needed to distinguish $G$ from $D_{n,r}$ to worst-case error $\epsilon$;
that is, $\cost_\epsilon(G,r)\coloneqq \Q_\epsilon(f)$,
where $f$ has domain $G\cup D_{n,r}\subseteq [n]^n$ and is defined by
$f(x)=1$ if $x\in G$ and $f(x)=0$ if $x\in D_{n,r}$. When $r\ge n$,
we set $\cost_\epsilon(G,r)\coloneqq\infty$. When $\epsilon=1/3$, we omit it
and write $\cost(G,r)$.
\end{definition}

We note that since $\cost_\epsilon(G,r)$ is defined as the worst-case quantum
query complexity of a Boolean function, it satisfies amplification,
meaning that the precise value of $\epsilon$ does not matter so long as
it is a constant in $(0,1/2)$ and so long as we do not care about
constant factors.

We define a well-shuffling class of permutation groups as follows.

\begin{definition}[Well-shuffling permutation groups]
\label{def:shuffling}
Let $\mathcal{G}$ be a collection of permutation groups.
We say $\mathcal{G}$ is \emph{well-shuffling}
if $\cost(G,r)=r^{\Omega(1)}$ for $G\in\mathcal{G}$ and $r\in\bN$.
More explicitly, we say $\mathcal{G}$ is well-shuffling
\emph{with power $a > 0$} if there exists $b > 0$ such that
$\cost(G,r)\ge r^{1/a}/b$ for all $G\in \mathcal{G}$ and all $r\in\bN$.
\end{definition}

We note that $\cost(G,r)\ge r^{1/a}/b$ is always satisfied when $r$
is greater than or equal to the domain size of $G$, since in that case
$\cost(G,r)=\infty$. Hence to show well-shuffling we only need to worry about
$r$ smaller than $n$, the domain size of the permutation group $G$.

The following theorem plays a central role in this work:
it shows that a well-shuffling
collection of permutation groups does not allow super-polynomial quantum speedups.

\begin{theorem}\label{thm:shuffle}
Let $f\colon\Dom(f)\to\B$ be a partial Boolean function on $n\in\bN$ bits,
with $\Dom(f)\subseteq\Sigma^n$
(where $\Sigma$ is a finite alphabet).
Let $G$ be a permutation group on $[n]$, and
suppose that $f$ is symmetric under $G$. Then there is a universal
constant $c\in\bN$ such that
\[\R(f)\le \min\{r\in\bN:\cost(G,r)\ge c\Q(f)\}.\]
Consequently, if
$\mathcal{G}$ is a well-shuffling collection
of permutation groups with power $a$, then for all $f\in F(\mathcal{G})$ we have
$\R(f)=O(\Q(f)^a)$.
\end{theorem}

To prove this theorem, we use the following minimax theorem for the
cost measure.

\begin{lemma}[Minimax for cost]\label{lem:cost_minimax}
Let $r,n\in\bN$ satisfy $r<n$, let $\epsilon\in[0,1/2)$,
and let $G$ be a permutation group on $[n]$. Then there is a distribution
$\mu$ on $D_{n,r}$ that is \emph{hard} in the following sense.
Let $\mu'$ be the uniform distribution on $G\subseteq[n]^n$.
Then any quantum algorithm for distinguishing $G$ from
$D_{n,r}$ which uses fewer than $\cost_\epsilon(G,r)$
queries must either make error $>\epsilon$ on average against $\mu$,
or else make error $>\epsilon$ against $\mu'$
(i.e., it fails to distinguish $\mu$ from the uniform distribution on $G$).
\end{lemma}

\begin{proof}
Let $f$ be the function that asks to distinguish $G$ from $D_{n,r}$
in the worst case. Then by the minimax theorem (\lem{minimax}), there is a hard
distribution $\nu$ for $f$, such that any quantum algorithm
using fewer than $\Q_\epsilon(f)=\cost_\epsilon(G,r)$ queries
must make more than $\epsilon$ error against $\nu$.
Let $\nu'$ be the distribution we get by applying a uniformly
random permutation from $G$ to a sample from $\nu$.
Then $\nu'$ is still a hard distribution for $f$. Indeed,
if it were not a hard distribution, there would be some quantum
algorithm $Q$ solving $f$ against $\nu'$ using too few queries;
but in that case, we could design an algorithm $Q'$ for solving
$f$ against $\nu$ simply by taking the input $x$, implicitly
applying a uniformly random $\pi$ from $G$ to permute the bits of $x$
(this can be done without querying $x$, simply by redirecting
all future queries $i\in[n]$ through the permutation $\pi$),
and then running $Q$ on the permuted string.

Now, note that composing
a uniformly random permutation from $G$ with an arbitrary (fixed)
permutation from $G$ gives a uniformly random permutation from $G$.
This means that $\nu'$ is some mixture of the uniform distribution $\mu'$
on $G\subseteq[n]^n$ and another distribution $\mu$ on $D_{n,r}$.
Then any algorithm that succeeds on both $\mu$ and $\mu'$ with
error at most $\epsilon$ will also succeed on $\nu'$ with error at most $\epsilon$,
from which the desired result follows.
\end{proof}

Using this lemma, we now prove \thm{shuffle}.

\begin{proof}[Proof of \protect\thm{shuffle}]
Let $Q$ be a quantum algorithm for $f$ that uses $\Q(f)$ queries.
Amplify it to $Q'$ by repeating $3$ times and taking the majority
vote; then it uses $3\Q(f)$ queries and makes worst-case error
$7/27$ instead of $1/3$.
Using \lem{cost_minimax}, let $\mu$ be the hard distribution on $D_{n,r}$
which is hard to distinguish from $G$ to error $\epsilon$,
where we pick $r$ later and pick $\epsilon$
to be a constant close to $1/2$. Sample $\alpha$ from $\mu$,
and consider the string $x\circ \alpha$
with $(x\circ\alpha)_i=x_{\alpha(i)}$.

Now, since $f$ is invariant
under $G$, $Q'$ outputs $f(x)$ with
error at most $7/27$ when run on $x\circ\pi$ for each $\pi\in G$.
In particular, consider picking $\pi$ from $G$ uniformly at random,
and running $Q'$ on $x\circ\pi$ where the string $x\in\Dom(f)$ is
fixed. Compare this to the behavior of $Q'$ on $x\circ\alpha$,
where $\alpha$ is sampled from the hard distribution $\mu$ on $D_{n,r}$.

If $Q'$ did not output $f(x)$ on $x\circ \alpha$
with error at most $1/3$, then we could convert $Q'$ to an
algorithm distinguishing $\pi$ from $\alpha$ with constant error. This is because $Q'$ outputs $f(x)$ with error at most $7/27<1/3$
on input $x\circ\pi$; hence $Q'$ behaves differently when run on
$x\circ\alpha$ and on $x\circ \pi$. We can convert $Q'$ to
an algorithm $Q''$ which hard codes the input $x$, and receives
either a random $\pi$ from $G$ or a random $\alpha$ from $\mu$
as input. This algorithm makes only $3\Q(f)$ queries
to $\pi$ or $\alpha$, but its acceptance probability
differs by a constant gap between the two distributions, which
(using some standard re-balancing) we can use to distinguish
$G$ from $\mu$ with at most constant error.

Now, assuming the distribution
$\mu$ was picked to be hard enough (i.e., $\epsilon$ was chosen
sufficiently close to $1/2$), this means that $3\Q(f)$,
the query cost of $Q''$, is at least $\cost_\epsilon(G,r)$.
Since $\cost_\epsilon(G,r)$ is the worst-case quantum query complexity
of a Boolean function, it can be amplified. We conclude that
if $Q'$ failed to output $f(x)$ on input $x\circ\alpha$ (with
$\alpha\leftarrow\mu$) with error at most $1/3$, then we have
$\Q(f)=\Omega(\cost(G,r))$, that is, $\Q(f)>\cost(G,r)/c$
for some universal constant $c$ (from amplification).

Now assume that $\Q(f)\le\cost(G,r)/c$. Then $Q'$ has error at most $1/3$
for computing $f(x)$ when run on $\alpha(x)$, with $\alpha$
chosen from $\mu$. Since $\alpha\in D_{n,r}$
uses at most $r$ alphabet symbols,
a randomized algorithm can simulate $Q'$ simply by picking
$\alpha$ from $\mu$ and querying all the $r$ bits of $x$ used
in the string $x\circ\alpha$, fully determining that string.
This algorithm $R$ uses $r$ queries, and makes at most $1/3$
error, so we conclude that $\R(f)\le r$.

By correctly picking $r$, we conclude that
$\R(f)\le \min\{r\in\bN:\cost(G,r)\ge c\Q(f)\}$,
as desired. Finally, note that if $\cost(G,r)\ge r^{1/a}/b$,
then by picking $r=(bc\Q(f))^a$ we get
$\cost(G,r)\ge c\Q(f)$. From this it follows that
$\R(f)\le (bc\Q(f))^a=O(\Q(f)^a)$, as desired.
\end{proof}

The upshot of \thm{shuffle} is that we can show a class of permutation groups
$\mathcal{G}$ does not allow super-polynomial quantum speedups
simply by showing that it is well-shuffling---that is, by showing
that $G\in\mathcal{G}$ is hard to distinguish from
the set of small-range functions $D_{n,r}$
using a quantum query algorithm.

\section{Showing permutation groups are well-shuffling}
\label{sec:wsgroups}

In this section, we introduce some tools for showing that a collection
of permutation groups is well-shuffling. Due to \thm{shuffle}, a well-shuffling
collection of permutation groups does not allow any super-polynomial
quantum speedups for the class of functions symmetric under it, so these tools
can be directly used to show that certain symmetries are not consistent
with large quantum speedups.

\subsection{The symmetric group}

The first fundamental result is that the class of full symmetric
permutation groups $S_n$ is well-shuffling.
This was shown by Zhandry \cite{Zha13} in a different context, though
we also provide a simpler proof by a reduction from the collision problem.

\begin{restatable}
{theorem}{zhandry}
\label{thm:zhandry}
There is a universal constant $C$ such that any quantum algorithm
distinguishing a permutation in $S_n$ from a string in $D_{n,r}$ must make at least $r^{1/3}/C$ queries.
\end{restatable}

This theorem says that $S_n$ is hard to distinguish from
$D_{n,r}$ (moreover, Zhandry \cite{Zha13} showed
that the hard distribution over $D_{n,r}$
is uniform, but we do not need this fact).

\begin{proof}
When $n$ is a multiple of $r$, then each $(n/r)$-to-$1$
function has range $r$ and each $1$-to-$1$ function is a permutation;
hence distinguishing $(n/r)$-to-$1$ from $1$-to-$1$ functions
is a sub-problem of distinguishing $D_{n,r}$ from $S_n$.
This sub-problem is the collision problem,
from which an $\Omega(r^{1/3})$
lower bound directly follows \cite{AS04,Amb05,Kut05}.
When $n$ is not a multiple of $r$ but $r\le n/2$, we can just set
$n'=r\lceil n/r\rceil$, and then distinguishing $(n'/r)$-to-$1$
from $1$-to-$1$ functions with domain size $n'$ still
reduces to distinguishing $D_{n,r}$ from $S_n$.
\end{proof}

From \thm{zhandry}, the following two corollaries immediately follow
(in light of \thm{shuffle}).

\begin{corollary}
The set of symmetric groups $\mathcal{S}=\{S_n\}_{n\in\bN}$
is well-shuffling with power $3$.
\end{corollary}

\begin{corollary}
\label{cor:zhandry}
All (possibly partial) Boolean functions $f$ that are symmetric
under the full symmetric group $S_n$ satisfy
$\R(f)=O(\Q(f)^3)$.
\end{corollary}

Apart from \thm{zhandry}, the main tools we use to prove
the well-shuffling property are transformations on permutation groups that
approximately preserve $\cost(G,r)$.
We outline several such transformations and invariances.
Since we prove \thm{zhandry} by a reduction from collision,
and since our main tools from here on out are additional reductions,
effectively all lower bounds in this paper work by
reductions from collision.

\subsection{Reduction between similar-looking permutation groups}

We next show that similar-looking permutation groups have similar costs. Here, two groups $G$ and $H$ are ``similar looking'' if a random permutation from $G$ is almost indistinguishable from a random permutation from $H$ when queried in sufficiently few places. More formally, we have the following theorem.

\begin{theorem}[Similar-looking permutation groups have similar costs]
\label{thm:transitivity_transformation}
Suppose $G$ and $H$ are permutation groups on $[n]$
and $k\le n$ is a positive integer such that for each
$i_1,i_2,\dots,i_k,j_1,j_2,\dots,j_k\in[n]$,
\[\left|\Pr_{\pi\leftarrow G}[\forall \ell\;\pi(i_\ell)=j_\ell]
-\Pr_{\pi\leftarrow H}[\forall \ell\;\pi(i_\ell)=j_\ell]\right|
\le n^{-10k}.\]
Then $\cost(H,r)\ge\Omega(\min\{k,\cost(G,r)\})$.
In particular, if $k\ge n^{\Omega(1)}$ and
$\cost(G,r)\ge r^{\Omega(1)}$, then $\cost(H,r)\ge r^{\Omega(1)}$.
\end{theorem}

\begin{proof}
Let $Q$ be a quantum algorithm for distinguishing
$H$ from $D_{n,r}$ which uses $\cost(H,r)$ and achieves worst-case
error $1/3$. If $\cost(H,r)\ge k$, we are done, so assume
$\cost(H,r)<k$. Now, $Q$ can be converted into a polynomial of
degree at most $2\cost(H,r)$ in the variables $z_{ij}$, where
$z_{ij}=1$ if the input $x$ satisfies $x_i=j$ and otherwise $z_{ij}=0$
(see \cite{AS04}).
This polynomial $p$ satisfies $p(x)\in [0,1/3]$ if $x\in D_{n,r}$ and
$p(x)\in[2/3,1]$ if $x\in H$. It has $n^2$ variables and degree
$d=2\cost(H,r)$. We assume it has no monomials that always
evaluate to $0$ (for example, $z_{11}z_{12}$, which is always
$0$ as $x_1$ cannot be both $1$ and $2$), because if it had such
monomials we could just delete them.

We claim that the sum of absolute values of coefficients of $p$ is at most
$n^{3d}$, where $d=2\cost(H,r)$ is its degree. To see this,
first note that there are at most $\binom{n^2}{d}$ monomials of $p$ of degree
$d$; for each such monomial $m$, let $p_m$ be the polynomial consisting
of all terms in $p$ that use a subset of the variables in $m$.
Then the sum of the absolute values of the coefficients of $p$ is at most
$\binom{n^2}{d}$ times the maximum sum of absolute values of the coefficients
in one of the polynomials $p_m$; since $\binom{n^2}{d}\le n^{2d}$, it suffices
to upper bound the sum of absolute values of coefficients of $p_m$
for arbitrary $m$. Now, $m$ consists of $d$ variables $z_{i_tj_t}$ for $t=1,2,\dots, d$,
which equal $1$ when $x_{i_t}=j_t$
and equal $0$ otherwise. Consider feeding into the quantum
algorithm an input string where $x_i=*$ when $i\notin\{i_1,i_2,\dots,i_d\}$,
and $x_{i_t}$ is either $j_t$ or $*$ for $t=1,2,\dots, d$. The quantum algorithm
will accept the string with some probability between $0$ and $1$, which
means the polynomial $p$ computing the acceptance probability of $Q$ will evaluate to
something between $0$ and $1$. But such inputs ``zero out'' all terms
that use variables outside of $m$, and hence turn $p$ into $p_m$.
From this we can conclude that $p_m$ is bounded in $[0,1]$ for all inputs
it receives in $\B^d$. But polynomials bounded in $[0,1]$ on the Boolean hypercube
can have sum of coefficients at most $5^d$ (one way to see this is
to recall that a bounded polynomial in the $\{-1,1\}$ basis has its sum
of squares of coefficients equal to at most $1$, and has at most $2^d$
coefficients, so by Cauchy-Schwartz, the sum of absolute values of coefficients
is at most $2^{d/2}$; converting the $\{-1,1\}$ basis to the $\B$ basis
requires plugging $2z-1$ terms into the variables, which can increase
the sum of absolute values by a factor of at most $3^d$, for a total of at most
$(3\sqrt{2})^d\le 5^d$). Assuming $n\ge 5$, we get an upper bound of $n^{3d}$
on the sum of absolute values of coefficients of $p$.

We have $d\le 2k$, so this sum is also at most $n^{6k}$.
Now, on each input $x$, the expected output of $p(\pi(x))$
when $\pi$ is sampled uniformly from $H$ is a linear combination
of the expectations of the monomials of $p$. For each monomial,
this expectation is just the probability that the monomial
is satisfied, which by the condition on $G$ and $H$
is within $n^{-10k}$ of the expectation under $\pi\leftarrow G$.
It follows that the expectation of $p(\pi(x))$ when $\pi\leftarrow H$
is within $n^{6k}n^{-10k}=n^{-4k}$ of the expectation of $p(\pi(x))$
when $\pi\leftarrow G$. But this expectation is simply the acceptance
probability of $Q$. Hence the acceptance probability of $Q$
on the uniform distribution on $H$ is within $n^{-4k}$
of the acceptance probability of $Q$ on the uniform distribution on $G$.

Since $Q$ distinguishes $H$ from $D_{n,r}$, it distinguishes
the uniform distribution on $H$ from any string in $D_{n,r}$.
Since it does not distinguish the uniform distribution on $H$
from the uniform distribution on $G$, $Q$ must also
distinguish the uniform distribution on $G$ from any input in
$D_{n,r}$ to error $1/3+n^{-4k}$. By amplifying, we can
get this down to error $1/3$, meaning that $\cost(G,r)=O(\cost(H,r))$,
as desired.
\end{proof}

An immediate consequence is that highly transitive permutation groups are well-shuffling:

\begin{corollary}\label{cor:informal_transitive}
If $G$ is $k$-transitive, then $\cost(G,r)=\Omega(\min\{k,r^{1/3}\})$,
where the constant in the big-$\Omega$ is universal.
\end{corollary}

\begin{proof}
This follows directly from \thm{transitivity_transformation},
setting $H$ to be the $k$-transitive permutation group we care about and setting
$G=S_n$. To see this, observe that $k$-transitivity completely determines
$\Pr_{\pi\leftarrow G}[\forall \ell\in[k]\;\pi(i_\ell)=j_\ell]$,
and that both $G$ and $H$ are $k$-transitive; hence this expression
is the same for both $G$ and $H$, and the difference between the two expressions
is exactly $0$ (certainly less than $n^{-10k}$).
\end{proof}

However, \cor{informal_transitive} is less powerful than it seems: it is a consequence of the classification of finite simple groups that the only $6$-transitive permutation groups are the symmetric group $S_n$ for $n \ge 6$ and the alternating group $A_n$ for $n \ge 8$, both in their natural action on $[n]$ \cite{DM12}. Nevertheless, it will be important for us later that the alternating group $A_n$, which is $(n-2)$-transitive, satisfies $\cost(A_n, r) = \Omega(r^{1/3})$.

\subsection{Transformations for graph symmetries}

Next, we introduce some additional transformations on permutation groups
that approximately preserve the cost. The transformations in this section
will allow us to show that graph property permutation groups (and several variants of them)
are well-shuffling.

\subsubsection{Transformation for directed graphs}

We start by defining an extension of a permutation group $G$ on $[n]$ to a permutation group acting on $[n]^\ell$. The notation in the definition below
comes from \cite{Ker13}.

\begin{definition}\label{def:power_action}
Let $G$ be a permutation group on domain $D$, and let $\ell\in\bN$.
Define $G^{(\ell)}$ to be the permutation group that acts on domain $D^\ell$
by $\pi(i_1,i_2,\dots,i_\ell)=(\pi(i_1),\pi(i_2),\dots,\pi(i_\ell))$
for each $\pi\in G$ (so the number of permutations in $G^{(\ell)}$
is the same as the number of permutations in $G$).

Define $G^{\langle\ell\rangle}$ to be the permutation group $G^{(\ell)}$
with domain restricted
to the subset $D^{\langle\ell\rangle}\subseteq D^\ell$
consisting of all distinct $\ell$-tuples of elements of $D$.
\end{definition}

We show that these transformations both preserve the cost, at least when $\ell$
is constant. We start with $G^{(\ell)}$.

\begin{theorem}
\label{thm:exponentiation}
Let $G$ be a permutation group on $[n]$, and let $H$ be
the permutation group $G^{(\ell)}$. Then $\cost(H,r^\ell)\ge\cost(G,r)/\ell$.
\end{theorem}

\begin{proof}
Let $Q$ be an algorithm distinguishing $H$ from $D_{n^\ell,r^\ell}$.
Let $\mu$ be the hard distribution for $G$, such that no algorithm
using fewer than $\cost(G,r)$ queries can distinguish $\mu$ from the
uniform distribution on $G$. Then $\mu$ is a distribution
on $D_{n,r}$. Let $\mu'$ be the distribution on $D_{n^\ell,r^{\ell}}$
that we get by sampling $\alpha\leftarrow\mu$ and returning
$\alpha'$ defined by
$\alpha'(z)=(\alpha(z_1),\alpha(z_2),\dots,\alpha(z_\ell))$
for each $z\in[n]^\ell$ (here we identify $[n^\ell]$ with $[n]^\ell$).
Note that if $\alpha$ has range $r$, then $\alpha'$ has range
at most $r^\ell$.

Then $Q$ distinguishes $\mu'$ from the uniform distribution
on $H$. The latter distribution is the same as obtained when sampling $\pi$ uniformly from $G$ and returning $\pi'$
defined by $\pi'(z)=(\pi(z_1),\pi(z_2),\dots,\allowbreak \pi(z_\ell))$
for $z$ in the domain of $H$.
This means that $Q$ can be used to distinguish $\mu$ from the
uniform distribution on $G$: all we need is to simulate
every query of $Q$ using $\ell$ queries to the input $\alpha$.
The desired result follows.
\end{proof}

To handle $G^{\langle \ell\rangle}$, we first observe that restricting the domain
of a permutation group to some union of its orbits does not decrease its cost.

\begin{lemma}
\label{lem:cutting_into_orbits}
Let $G$ be a permutation group on $[n]$, and let $S\subseteq[n]$
be a union of orbits of $G$. Let $G'$ be the permutation group $G$
acting only on $S$. Then $\cost(G',r)\ge\cost(G,r)$.
\end{lemma}

\begin{proof}
We identify $S$ with $[|S|]$ without loss of generality.
If $Q$ distinguishes $G'$ from $D_{|S|,r}$, then
we can turn it into $Q'$ distinguishing $G$ from $D_{n,r}$ by
having $Q'$ run $Q$ and make queries only from $[|S|]$.
\end{proof}

The fact that $G^{\langle\ell\rangle}$ does not decrease the cost of $G$ too much then follows as a corollary of \thm{exponentiation} and \lem{cutting_into_orbits}.

\begin{corollary}\label{cor:angle_power}
Let $G$ be a permutation group on $[n]$, and let $H$ be the permutation group $G^{\langle\ell\rangle}$. Then $\cost(H,r^\ell)\ge\cost(G,r)/\ell$.
\end{corollary}

\begin{proof}
It suffices to note that $G^{\langle\ell\rangle}$ is the permutation group $G^{(\ell)}$ with domain restricted to $[n]^{\langle\ell\rangle}$,
which is a union of orbits because $\pi\in G^{(\ell)}$ always sends
a tuple with unique entries to another tuple with unique entries
(since a permutation on $[n]$ is applied to each entry).
The desired result then follows from
\thm{exponentiation} and \lem{cutting_into_orbits}.
\end{proof}

We now observe that the transformation $G^{\langle\ell\rangle}$ immediately allows us to show
that directed graph symmetries are well-shuffling.

\begin{corollary}[Directed graph symmetries]\label{cor:dir_graph}
The set $\mathcal{G}=\{G_k\}_{k\in\bN}$ of all directed graph
symmetries is well-shuffling with power $6$. Here the permutation group $G_k$
acts on a domain of size $n=k(k-1)$ representing the possible
arcs of a $k$-vertex directed graph, and $G_k$ consists of
all $k!$ permutations on these arcs that act by relabeling the vertices.
\end{corollary}

\begin{proof}
This immediately follows by observing that $G_k=S_k^{\langle 2\rangle}$.
To see this, note that the domain of $G_k$ is the set of all ordered pairs
$(x,y)\in [k]$ with $x\ne y$, which is precisely $[k]^{\langle 2\rangle}$,
and the permutations in $G_k$ are just those in $S_k$ applied to
both coordinates, which is precisely relabeling the vertices.
\cor{angle_power} then gives $\cost(G_k,r)\ge\cost(S_k,\sqrt{r})/2$,
which is at least $\Omega(r^{1/6})$ by \thm{zhandry}.
\end{proof}

The collection of directed hypergraph symmetries is similarly well-shuffling.

\begin{corollary}[Directed hypergraph symmetries]
The set $\mathcal{G}_p=\{G_k\}_{k\in\bN}$ consisting of all
$p$-uniform directed hypergraph symmetries is well-shuffling
with power $3p$.
\end{corollary}

\begin{proof}
This follows from the same argument as \cor{dir_graph}.
\end{proof}

\subsubsection{Transformation for undirected graphs}

To handle undirected graphs, we introduce yet another operation on
permutation groups which approximately preserves the cost.

\begin{theorem}\label{thm:classes}
Let $G$ be a permutation group, and let
$\mathcal{B} = \{B_1,B_2,\dots,B_k\}$ be a block system for $G$, where the blocks have equal size.
Let $H$ be the permutation group on $[k]$ induced by the action of $G$ on the blocks.
Then $\cost(H,r)\ge\cost(G,r)$.
\end{theorem}

\begin{proof}
Let $Q$ be a quantum algorithm distinguishing $H$ from $D_{k,r}$
using $\cost(H,r)$ queries. We construct
a quantum algorithm $Q'$ for distinguishing $G$ from $D_{n,r}$.
Let $\mu$ be the distribution
on $D_{n,r}$ that is hard to distinguish from the uniform
distribution on $G$. The algorithm $Q'$ fixes a unique
$i_t\in B_t$ for each $t=1,2,\dots,k$.
On input $\alpha$ from $G\cup D_{n,r}$,
the algorithm $Q'$ runs $Q$ as follows:
each query $t\in[k]$ that $Q$ makes is turned into the query
$i_t\in[n]$ for $\alpha$, and the output $\alpha(i_t)$ is converted into the symbol $t'$ such that $\alpha(i_t)\in B_{t'}$
and returned to $Q$. In this way, the algorithm $Q'$ effectively
runs $Q$ on the mapped string $\phi(\alpha)\in[k]^k$, where
$\phi(\alpha)_t$ is the symbol $t'$ such that $\alpha(i_t)\in B_{t'}$.

Now, if $\alpha\in D_{n,r}$, then $\phi(\alpha)\in D_{k,r}$,
while if $\alpha\in G$, we have $\phi(\alpha)\in H$. Since $Q$
distinguishes $H$ from $D_{k,r}$, it follows that $Q'$ distinguishes
$G$ from $D_{n,r}$ using the same number of queries, as desired.
\end{proof}

We are now finally ready to prove the formal version
of \thm{main_graph}, showing that the
collection of (undirected) graph symmetries is well-shuffling.

\begin{definition}[Graph Symmetries]\label{def:graph_sym}
The collection of \emph{graph symmetries} is the set
$\mathcal{G}=\{G_k\}_{k\in \bN}$ of permutation groups
with $G_k$ acting on $[n]$ with $n=k(k-1)/2$,
such that the domain $[n]$ represents the set
of all possible edges in a $k$-vertex graph,
and $G_k$ acts on these edges and permutes them
in a way that corresponds to relabeling the vertices of the underlying graph.
\end{definition}

\begin{corollary}\label{cor:formal_graph}
The set of all graph symmetries is well-shuffling with power $6$.
Hence $\R(f)=O(\Q(f)^6)$ for functions $f$ symmetric under a graph symmetry.
\end{corollary}

\begin{proof}
Let $G$ be a directed graph symmetry on domain size $k(k-1)$,
and partition this domain into $k(k-1)/2$ sets of size $2$ of the form
$\{(x,y),(y,x)\}$ for $x,y\in[k]$. Then the induced permutation group $H$
on these sets (from \thm{classes}) is precisely the undirected graph symmetry
on graphs of size $k$. Since $\cost(H,r)\ge\cost(G,r)$, and since
the directed graph symmetries are well-shuffling with power $6$, it follows
that the undirected graph symmetries are also well-shuffling with power $6$.
\end{proof}

Using similar arguments, we can show a similar result for hypergraphs.
\begin{corollary}
For every constant $p \in \bN$, the collection of all $p$-uniform hypergraph symmetries is well-shuffling with power $3p$.
\end{corollary}

\subsubsection{Transformations for bipartite graphs}

We introduce yet more operations on permutation groups for the case of bipartite graph
symmetries.

\begin{definition}[Product of permutation groups]\label{def:group_product}
Let $G_1$ and $G_2$ be two permutation groups acting on $[n_1]$ and $[n_2]$, respectively.
Then the product permutation group $G_1\times G_2$ is a permutation group acting on $[n_1 n_2]$ such that
for any $(\pi_1, \pi_2) \in G_1\times G_2$, and any $k\in[n_1]$ and $\ell\in [n_2]$, we have
$(\pi_1, \pi_2)(k, \ell) = (\pi_1(k), \pi_2(\ell))$.
(Here we identify $[n_1]\times[n_2]$ with $[n_1n_2]$.)
\end{definition}

\begin{theorem}
\label{thm:group-product}
For all $G_1$ and $G_2$ acting on $[n_1]$ and $[n_2]$, respectively, and for all $r$, $\cost(G_1\times G_2,r^2)\ge \min\{\cost(G_1,r),\cost(G_2,r)\}$.
\end{theorem}

\begin{proof}
Let $H=G_1 \times G_2$ and $m =n_1n_2$.
Let $Q$ be an algorithm distinguishing $H$ from $D_{m,r^2}$.
Let $\mu_1$ be the hard distribution for $G_1$ (on $D_{n_1,r}$), and let $\mu_2$ be
the hard distribution for $G_2$ (on $D_{n_2,r}$).
Let $\mu'$ be the distribution on $D_{m,r}$ that we get by sampling $\alpha_1 \leftarrow\mu_1$, and $\alpha_2\leftarrow\mu_2$ independently, and returning
$\alpha'=(\alpha(z_1),\alpha(z_2))$.
Note that if $\alpha_1$ and $\alpha_2$ have range $r$, then $\alpha'$ has range
at most $r^2$. Now, since $Q$ distinguishes $D_{m,r^2}$ from $G_1\times G_2$,
it must also distinguish $\mu'$ from the uniform distribution over $G_1\times G_2$,
which itself is the product of the uniform distribution on $G_1$ and the
uniform distribution on $G_2$. Let $\nu_1$ be the uniform distribution on $G_1$,
and let $\nu_2$ be the uniform distribution on $G_2$. Consider
the behavior of $Q$ on $\mu_1\times\nu_2$. It must either distinguish
this distribution from $\mu_1\times\mu_2$, or else from $\nu_1\times\nu_2$
(since it distinguishes $\mu_1\times\mu_2$ and $\nu_1\times\nu_2$ from
each other). In the first case, we can construct $Q'$ which
artificially generates the sample from $\mu_1$ and uses $Q$ to distinguish
$\mu_2$ from $\nu_2$. In the second case, we can construct $Q'$ which
artificially generates the sample from $\nu_2$ and uses $Q$ to distinguish
$\mu_1$ from $\nu_1$.
Hence $\cost(G_1\times G_2,r^2)\ge\min\{\cost(G_1,r),\cost(G_2,r)\}$, as desired.
\end{proof}

\begin{corollary}
The collection $\mathcal{G}$
of all bipartite graph symmetries with equal parts is well-shuffling.
\end{corollary}

\begin{proof}
This immediately follow by observing that bipartite graph symmetries are
the symmetries $S_{k}\times S_{k}$.
Then \thm{group-product} and \thm{zhandry} give the desired result.
\end{proof}

\subsubsection{Other transformations}

We introduce one final transformation, which merges two permutation groups into one. This transformation also does not decrease the cost.

\begin{lemma}[Merger]
\label{lem:extension}
Let $G$ and $H$ be two permutation groups on $[n]$, and let
$F=\langle G,H\rangle$ be the permutation group on $[n]$ which is the closure
of $G\cup H$ under composition. Then $\cost(F,r)\ge\cost(G,r)$.
\end{lemma}

\begin{proof}
Since $G$ is a subset of $F$, distinguishing $G$ from $D_{n,r}$
is strictly easier than distinguishing $F$ from $D_{n,r}$.
\end{proof}

\section{Classification of well-shuffling permutation groups}
\label{sec:wsclassify}

In this section, we show that the results from the previous section are qualitatively optimal, in the sense that permutation groups constructed out of hypergraph symmetries are essentially the \emph{only} groups that are not consistent with super-polynomial quantum speedups. Our starting point is an observation that super-polynomial quantum speedups can be constructed out of any permutation group with sufficiently small minimal base size.

\begin{proposition}\label{prop:small_base}
Let $G$ be a permutation group on $[n]$, and let $f$ be a (possibly partial) Boolean function with $\Dom(f) \subseteq \Sigma^n$ for some alphabet $\Sigma$. Then there exists a partial Boolean function $g$ that is symmetric under $G$ such that $Q(g) \le Q(f) + b(G)$ and $R(g) \ge R(f)$.
\end{proposition}

\begin{proof}
Define a promise problem $g$ as follows. A string $x \in (\Sigma \times [n])^n$ is in the promise if there exists a permutation $\pi \in G$ such that for all $i$, $x_{\pi(i)} = (y_i, i)$ where $y \in \Sigma^n$ is in the promise of $f$. In other words, the promise for $g$ contains all reorderings of strings of the form $(y_i, i)$ by permutations that are in G. Naturally, in this case we define $g(x) = f(y)$. Then clearly $R(g) \ge R(f)$, because any randomized algorithm that solves $g$ also solves $g$ restricted to the promise that $\pi$ is the identity permutation, in which case $g$ is equivalent to $f$. Moreover, this problem is clearly symmetric under $G$. On the other hand, $Q(g) \le Q(f) + b(G)$: by querying a minimal base for $G$, the algorithm can learn the unique permutation $\pi$ such that $x_{\pi(i)} = (y_i, i)$. Now, the algorithm just permutes $x$ according to $\pi$ and runs the query algorithm for $f$.
\end{proof}

Put another way, if $b(G) = n^{o(1)}$, then by choosing an arbitrary function $f$ with $Q(f) = n^{o(1)}$ and $R(f) = n^{\Omega(1)}$ (e.g., Simon's problem \cite{Sim97} or Forrelation \cite{AA15}), then one can construct a function $g$ symmetric under $G$ with $Q(g) = n^{o(1)}$ and $R(g) = n^{\Omega(1)}$. Thus, permutation groups that do not allow super-polynomial speedups must have large minimal base size.

Naturally, this raises the question of whether a sort of converse holds, i.e., whether $b(G) = n^{\Omega(1)}$ implies that $G$ is well-shuffling. We will show that this is true for primitive permutation groups (\thm{primitive_large_base}), and moreover that primitive groups with $b(G) = n^{\Omega(1)}$ all look roughly like symmetries of $p$-uniform hypergraphs (\cor{primitive_shuffling_equivalence}). Thus, we completely classify which primitive permutation groups are consistent with super-polynomial quantum speedups, and which are not.

For imprimitive groups, we will see that there exist permutation groups $G$ with $b(G) = n^{\Omega(1)}$ that are consistent with large quantum speedups. So, $b(G) = n^{\Omega(1)}$ is not equivalent to the well-shuffling property in general. Nevertheless, we can show that primitive permutation groups with $b(G) = n^{\Omega(1)}$ are the ``building blocks'' of well-shuffling group actions. Indeed, we show in \cor{imprimitive_classification} that if a collection of imprimitive permutation groups is inconsistent with super-polynomial speedups, then it must be built out of a constant number of well-shuffling primitive groups. Thus, in some sense, all permutation groups that do not allow large quantum speedups must involve symmetries of $p$-uniform hypergraphs.

\subsection{Primitive groups}

For primitive permutation groups, the following theorem due to Liebeck \cite{Lie84} shows that there is a tight connection between minimal base size and the structure of the group. Its proof relies on the classification of finite simple groups.

\begin{theorem}\label{thm:liebeck}
Let $G$ be a primitive permutation group on $[n]$. Then one of the following holds:
\begin{enumerate}[label=(\roman*)]
\item $G$ is a subgroup of $S_\ell \wr S_m$ containing $A_m^{\times \ell}$, where the action of $S_m$ is on $p$-element subsets of $[m]$ and the wreath product has the product action of degree $n = {\binom{m}{p}}^\ell$, or
\item $b(G) < 9 \log_2 n$.
\end{enumerate}
Here, $S_n$ denotes the symmetric group on $[n]$ and $A_n$ denotes the alternating group on $[n]$.
\end{theorem}

Observe that in case (i), when $\ell = 1$, this corresponds precisely to either the set of $p$-uniform undirected hypergraph symmetries, or the index-2 subgroup of even permutations of the vertices. When $p$ is a constant, we know that these groups are well-shuffling. More generally, as long as both $\ell$ and $p$ are constant, then these groups are well-shuffling:

\begin{corollary}\label{cor:primitive_shuffling}
For all constant $\ell$, $p$, the collection of primitive permutation groups that satisfy case (i) of \thm{liebeck} is well-shuffling with power $3\ell p$.
\end{corollary}

\begin{proof}
Suppose $G$, $\ell$, $p$, and $m$ are as in case (i) of \thm{liebeck}. Let $A$ denote the action of the alternating group $A_m$ on $[m]$, and let $H$ denote the action of $A_m$ on size-$p$ subsets of $[m]$. Formally, we have the following chain of inequalities:
\begin{align*}
    \cost(G, r) &\ge \cost(H^{\times \ell}, r) & \textrm{\lem{extension}}\\
    &\ge \cost(H^{(\ell)}, r) & \textrm{\lem{extension}}\\
    &\ge \cost(H, r^{1/\ell})/\ell & \textrm{\thm{exponentiation}}\\
    &\ge \cost(A^{\langle p\rangle}, r^{1/\ell})/\ell & \textrm{\thm{classes}}\\
    &\ge \cost(A, r^{1/\ell p})/\ell p & \textrm{\cor{angle_power}}\\
    &\ge \Omega\bigl(r^{1/3\ell p}\bigr). & \textrm{\cor{informal_transitive}}
\end{align*}
In words, in the first two lines, it suffices to show that a subgroup of $G$ is well-shuffling, and thus that a subgroup of a subgroup is well-shuffling. $H^{(\ell)}$ is the so-called ``diagonal subgroup'' of $H^{\times \ell}$, consisting of the permutations $(\pi_1,\pi_2,\ldots,\pi_\ell)$ such that $\pi_i = \pi_j$ for all $i, j$; this is consistent with the notation from \defn{power_action}. The third inequality appeals to the cost lower bound for power actions. Then, we observe that $H$ is equivalent to the action of $A^{\langle p\rangle}$ on blocks, where $A^{\langle p \rangle}$ is the power action of $A$ restricted to tuples with no repeats, and the blocks are sets of tuples that are reorderings of each other. The next inequality uses the cost lower bound for this restriction of the power action. Finally, we appeal to the fact that the alternating group on $m$ points is $(m-2)$-transitive, and thus well-shuffling.
\end{proof}

In fact, the following theorem shows that a collection of primitive permutation groups is well-shuffling if and only if the permutation groups are as in case (i) of \thm{liebeck}, with $\ell$ and $p$ both constant:

\begin{theorem}\label{thm:primitive_large_base}
For all constant $\epsilon > 0$, there exists $c \in \mathbb{N}$ such that the collection of primitive permutation groups $G$ with $b(G) \ge n^{\epsilon}$ is well-shuffling with power $c$. Moreover, for all sufficiently large $n$, all such $G$ are as in case (i) of \thm{liebeck} with $\ell p \le c/3$.
\end{theorem}

\begin{proof}
For sufficiently large $n$, we cannot be in case (ii) of \thm{liebeck}. This is because $b(G) = O(\log n)$, so by \prop{small_base}, we can construct a function symmetric under $G$ with a super-polynomial gap between quantum and randomized query complexity.

Therefore, suppose we are in case (i) of \thm{liebeck}. If $\binom{m}{p} = 1$, then by the definition of the wreath product (\defn{wreath_product}), $G$ is the trivial group. Otherwise, we suppose $\binom{m}{p} \ge 2$ and therefore $m \ge 2$. The minimal base size of $G$ is upper bounded by
\begin{align*}
    b(G) &\le \log_2(|G|) & \textrm{\lem{base_bounds}}\\
    &\le \log_2(|S_\ell \wr S_m|) & G \subseteq S_\ell \wr S_m\\
    &= \log_2(\ell! \cdot m!^\ell) & \textrm{\fct{wreath_product_order}}\\
    &\le \ell \log_2 \ell + \ell m \log_2 m\\
    &\le 2\ell^2m^2.
\end{align*}
Thus we have
$$\left(\frac{m}{p}\right)^{\ell p\epsilon}\le \binom{m}{p}^{\ell \epsilon} = n^{\epsilon} \le b(G) \le 2\ell^2m^2,$$
so
$$\ell p \epsilon (\log_2 m - \log_2 p) \le 1 + 2 \log_2 \ell + 2 \log_2 m.$$
Consider two cases. First, suppose $m \le p^2$. Without loss of generality, we may always assume $p \le m / 2$, because the action of $S_m$ on $p$-element subsets is equivalent to the action on $(m-p)$-element subsets, by identifying each $p$-element subset with its complement in $[m]$. So we have
$$\ell p\epsilon \le \ell p\epsilon (\log_2 m - \log_2 p) \le 1 + 2 \log_2 \ell + 2 \log_2 m \le 1 + 2 \log_2 \ell + 4\log_2 p.$$
But clearly $\ell p\epsilon \le 1 + 2 \log_2 \ell + 4\log_2 p$ is possible only for bounded values of $\ell p$.

Otherwise, suppose $m \ge p^2$. Then we have
$$\ell p\epsilon \le \frac{1 + 2 \log_2 \ell + 2 \log_2 m}{\log_2 m - \log_2 p} \le 2\frac{1 + 2 \log_2 \ell + 2 \log_2 m}{\log_2 m } = 4 + \frac{2 + 4 \log_2 \ell}{\log_2 m } \le 6 + 4\log_2 \ell.$$
Likewise, $\ell p\epsilon \le 6 + 4\log_2 \ell$ is possible only for bounded values of $\ell p$. So we may suppose that $\ell p\le c/3$ for some constant $c$. Then, the theorem follows from \cor{primitive_shuffling}.
\end{proof}

\begin{corollary}\label{cor:primitive_shuffling_equivalence}
Let $\mathcal{G}$ be a collection of primitive permutation groups. The following are equivalent:
\begin{enumerate}[label=(\roman*)]
    \item $\mathcal{G}$ is well-shuffling.
    \item $b(G) = n^{\Omega(1)}$ for $G \in \mathcal{G}$ acting on $[n]$.
    \item All but finitely many $G \in \mathcal{G}$ satisfy case (i) of \thm{liebeck} with $\ell p = O(1)$.
\end{enumerate}
\end{corollary}

\begin{proof}
(i) implies (ii) by \prop{small_base}, because well-shuffling permutation groups are not consistent with super-polynomial quantum speedups. (ii) implies (iii) by \thm{primitive_large_base}, and (iii) implies (i) by \cor{primitive_shuffling}.
\end{proof}

\begin{corollary}\label{cor:primitive_dichotomy}
Let $\mathcal{G}$ be a collection of primitive permutation groups. Then:
\begin{enumerate}[label=(\roman*)]
    \item If $b(G) = n^{\Omega(1)}$ for $G \in \mathcal{G}$ acting on $[n]$, then $\mathcal{G}$ is well-shuffling, and thus $R(f) = Q(f)^{O(1)}$ for all $f\in F(\mathcal{G})$.
    \item If $b(G) = n^{o(1)}$ for $G \in \mathcal{G}$ acting on $[n]$, then there exists a class of partial functions $f \in F(\mathcal{G})$ with $R(f) = Q(f)^{\omega(1)}$.
\end{enumerate}
\end{corollary}

It would be interesting to prove some version of \cor{primitive_dichotomy} directly using the definition of $b(G)$, without appealing to such deep results about the structure of permutation groups.

\subsection{Imprimitive groups}

Starting with an arbitrary permutation group $G$, the first place to look for structure is in the orbits of $G$. We first observe that quantum speedups can be constructed out of any permutation group with many orbits.

\begin{proposition}
Let $G$ be a permutation group on $[n]$. Let $\mathcal{O} = \{O_1,O_2,\ldots,O_k\}$ be a partition of $[n]$ into $k$ orbits of $G$. Let $f$ be a (possibly partial) Boolean function with $\Dom(f) \subseteq \Sigma^k$ for some alphabet $\Sigma$. Then there exists a partial Boolean function $g$ that is symmetric under $G$ such that $Q(g) = Q(f)$ and $R(g) = R(f)$.
\end{proposition}

\begin{proof}
Define the promise for $g$ to consist of the strings $x$ such that for all $i \in [n]$, $x_i = y_j$ where $i \in O_j$ and $y \in \Dom(f)$. Naturally, in this case we define $g(x) = f(y)$. Then this problem is trivially symmetric under $G$, and moreover these two problems have the same (quantum or randomized) query complexity because $g$ is just $f$ with some inputs possibly repeated.
\end{proof}

By choosing $f$ to be some query problem with $Q(f) = O(1)$ but $R(f) = \omega(1)$ (e.g., Forrelation \cite{AA15}), one can construct a super-polynomial speedup out of any collection of permutation groups with $\omega(1)$ orbits.

Next, we observe that if $G$ restricted to any orbit is consistent with super-polynomial speedups, then so is $G$ as a whole.

\begin{proposition}
Let $G$ be a permutation group on $[n]$. Let $\mathcal{O} = \{O_1,O_2,\ldots,O_k\}$ be a partition of $[n]$ into $k$ orbits of $G$, and let $G|_{O_i}$ denote the restriction of $G$ acting only on $O_i$. Let $f$ be a (possibly partial) Boolean function that is symmetric under $G|_{O_i}$ with $\Dom(f) \subseteq \Sigma^{|O_i|}$ for some alphabet $\Sigma$ and orbit $i$. Then there exists a partial Boolean function $g$ that is symmetric under $G$ such that $Q(g) = Q(f)$ and $R(g) = R(f)$.
\end{proposition}

\begin{proof}
Define the promise for $g$ to consist of the strings $x$ that agree with some string $y \in \Dom(f)$ on $O_i$ and are some constant symbol elsewhere. Naturally, in this case we define $g(x) = f(y)$. Then this problem is trivially symmetric under $G$, and moreover these two problems have the same (quantum or randomized) query complexity because $g$ is just $f$ with extra inputs that do not affect the output.
\end{proof}

Thus, a collection of permutation groups that does not allow any super-polynomial quantum speedups must remain so when restricted to any subset of its orbits. For this reason, the main task is to understand which \emph{transitive} permutation groups (i.e., groups with a single orbit) allow super-polynomial quantum speedups. The following well-known embedding theorem shows that impritive transitive groups can be expressed as subgroups of a wreath product constructed from a nontrivial block system. See Theorem II.49 of \cite{Hul10} for a proof.

\begin{lemma}[Embedding theorem]\label{lem:embedding}
Let $G$ be a transitive imprimitive permutation group on domain $D$ with $\mathcal{B} = \{B_1, B_2, \ldots, B_k\}$ a nontrivial block system for $G$. Let $H$ denote the permutation group acting on $[k]$ induced by the action of $G$ on blocks $\mathcal{B}$. Let $\Stab_G(\{B_1\})|_{B_1}$ denote the permutation group induced by the action of $Stab_G(\{B_1\})$ restricted to the orbit $B_1$. Then the action of $G$ is permutation isomorphic to a subgroup of $H \wr \Stab_G(\{B_1\})|_{B_1}$ in imprimitive action, meaning that $G$ can be viewed as a subgroup under an appropriate relabeling that identifies $D$ with $[k] \times B_1$.
\end{lemma}

Applying \lem{embedding} recursively, one can always decompose a transitive imprimitive permutation group as a subgroup of an iterated wreath product of primitive groups. Thus, primitive groups are sometimes described as the ``building blocks'' of permutation groups. Note that such a decomposition is not unique in general, because a permutation group can have many nontrivial block systems.

The next proposition shows that if any terms in such a wreath product decomposition are consistent with super-polynomial speedups, then so is the group as a whole.

\begin{proposition}\label{prop:wreath_speedup}
Let $G$ be a permutation group on $[n]$ that is a subgroup of an iterated wreath product $G_1 \wr G_2 \wr \cdots \wr G_d$ in imprimitive action. Suppose $f$ is a (possibly partial) function symmetric under $G_i$ for some $i$. Then there exists a partial function $g$ symmetric under $G$ with $Q(g) = Q(f)$ and $R(g) = R(f)$.
\end{proposition}

\begin{proof}
Suppose $G_i$ acts on $n_i$ points, so that $n = \prod_{i=1}^d n_i$. It suffices to exhibit such a function when $G$ is the full wreath product, and each $G_j$ is the symmetric group $S_{n_j}$ for every $j \neq i$. Define the function $\triv_m$ to be a promise problem that only takes two inputs, $0^m$ and $1^m$, outputting $0$ on $0^m$ and $1$ on $1^m$. The block composition $\triv_{n_1n_2\cdots n_{i-1}} \circ f \circ \triv_{n_{i+1}n_{i+2}\cdots n_d}$ is symmetric under $G$, and has quantum and randomized query complexities identical to those of $f$.
\end{proof}

Less trivially, if a wreath product decomposition of a group is sufficiently deep, then the group is also consistent with large speedups. The following construction is based on the so-called ``$\kFaultDirectTrees$'' problem of \cite{ZKH12,Kim12}:

\begin{theorem}\label{thm:deep_wreath_speedup}
Let $G$ be a permutation group on $[n]$ that is a subgroup of an iterated wreath product of nontrivial permutation groups $G_1 \wr G_2 \wr \cdots \wr G_d$ in imprimitive action. Then there exists a function $f$ symmetric under $G$ such that $Q(f) = O(1)$ but $R(f) = \Omega(\log d)$.
\end{theorem}
\begin{proof}
Suppose $G_i$ acts on $n_i \ge 2$ points, so that $n = \prod_{i=1}^d n_i$. It suffices to show that there exists a super-polynomial speedup when $G$ is the full wreath product and each $G_i$ is the symmetric group $S_{n_i}$. In this case, $G$ can be described the group of symmetries of a depth-$d$ rooted tree, where all of the vertices at distance $i - 1$ from the root have exactly $n_i$ children. $G$ is the action of this group of symmetries on the leaves of the tree. We construct a problem symmetric under $G$ by taking a Boolean evaluation tree and adding a promise that gives an exponential speedup.

If $n_i = 2$ for all $i$, such a problem is already known: the $\kFaultDirectTrees$ problem is a promise problem symmetric under $G$ with quantum query complexity $O(1)$ \cite{Kim12} and randomized query complexity $\Omega(\log d)$ \cite{ZKH12}. In the simplest case, the problem is a restriction of the depth-$d$ binary $\NAND$ tree, but it can also be generalized to Boolean evaluation trees composed of other functions.

We take a slight generalization of the $\kFaultDirectTrees$ trees problem where the Boolean functions can have different fan-in at each depth. Specifically, we take the tree to be the block composition of functions $f_1 \circ f_2 \circ \cdots \circ f_d$ where the function $f_i$ at depth $i$ is defined by $f_i\colon S \to \{0,1\}$ where $S \subset \{0,1\}^{n_i}$ consists of the $n_i$-bit strings that have Hamming weight in $\{0,\lfloor n_i/2 \rfloor, \lceil n_i/2 \rceil, n_i\}$, and $f_i(x) = 0$ if and only if $x = 1^{n_i}$. Note that $f_i$ is essentially just a padded version of the binary $\NAND$ function. The proof that this problem has the desired query complexity bounds follows almost immediately from the proofs in \cite{Kim12, ZKH12}; for completeness, we provide the details in \app{wreath_tree}.
\end{proof}

We remark that \thm{deep_wreath_speedup} shows the existence of super-polynomial speedups under permutation groups with large base. Indeed, if we define $G = \underbrace{S_2 \wr S_2 \wr \cdots \wr S_2}_{d \text{ times}}$ to be the depth-$d$ iterated wreath product in imprimitive action of the symmetric group $S_2$, then $b(G) = 2^{d-1}$ even though $G$ acts on $2^d$ points.

Putting all of these together, we can say the following about collections of permutation groups that are inconsistent with super-polynomial speedups:

\begin{corollary}\label{cor:imprimitive_classification}
Let $\mathcal{G}$ be a collection of permutation groups. Suppose that for every $f \in F(\mathcal{G})$, we have $R(f) = Q(f)^{O(1)}$. Then the following must hold:
\begin{enumerate}[label=(\roman*)]
    \item Permutation groups in $\mathcal{G}$ have $O(1)$ orbits.
    \item The collection of restrictions of groups $G \in \mathcal{G}$ to subsets of their orbits is also inconsistent with super-polynomial quantum speedups.
    \item All wreath product decompositions of groups $G \in \mathcal{G}$ restricted to an orbit have $O(1)$ primitive factors.
    \item The collection of primitive factors that appear in such wreath product decompositions is well-shuffling. Equivalently, by \cor{primitive_shuffling_equivalence}, those primitive factors $G$ that act on $[n]$ have $b(G) = n^{\Omega(1)}$, and all but finitely many satisfy case (i) of \thm{liebeck} with $\ell p = O(1)$.
\end{enumerate}
\end{corollary}

Thus, we have formally shown that collections of permutation groups that do not allow super-polynomial quantum speedups must be constructed out of a constant number of well-shuffling primitive groups. As we have shown before, such primitive groups are essentially just extensions of permutation groups that correspond to $p$-uniform hypergraph symmetries.

\section{Exponential speedup for adjacency-list graph property testing}
\label{sec:testing}

Having shown that graph properties in the adjacency \emph{matrix} model cannot have exponential quantum speedup, we now turn to the adjacency \emph{list} model. Specifically, we describe a graph property testing problem in the adjacency list model for which there is an exponential quantum speedup. We work with graphs that are undirected and have degree at most $5$, while for convenience\footnote{This is simply for clarity of presentation; we use self-loops and double edges as markers, for nodes and edges respectively. However, one could also work with simple graphs only and get the same exponential separation. This can be achieved by using slightly more complicated marker structures, e.g., one could replace self-loops with an edge connected to a triangle, etc.} we allow the graphs to have self-loops and parallel edges (we count these with multiplicity toward the degree bound).

Recall that in property testing, the goal is to distinguish between some set of yes instances and the set of no instances consisting of all graphs that are $\eps$-far from the yes instances. More precisely we say that a graph $G=(V,E)$ (with maximum degree at most $5$) is $\eps$-far from the property $\mathcal{P}$, if for any yes instance $G'=(V,E')\in\mathcal{P}$ with $n$ vertices, the symmetric difference of the edge sets has size $|E \triangle E'|\geq \eps n$, where we count edges with multiplicity. Since we work with graph properties, we also require that $\mathcal{P}$ is invariant under permuting the vertices.

\begin{figure}[ht]
\centering
\begin{tikzpicture}[scale=0.96]
\SetGraphUnit{1.5}
\GraphInit[vstyle=Simple]
\tikzset{VertexStyle/.style = {shape = circle,fill = black,minimum size = 3.5pt,inner sep=1pt}}

\candy{(-8,0)}{1.2}
\node[text width=5cm, align=center] at (-2.5,0) {$\boldsymbol{=}$};
\node[text width=5cm, align=center] at (-5.6,0) {Body};
\node[text width=2cm, align=center, rotate = 90] at (-7.8,0) {Antenna};
\node[text width=2cm, align=center, rotate = 90] at (-3.5,0) {Antenna};

\Vertex[x=-1,y=1]{T11}
\Vertex[x=-1,y=-1]{T12}
\Vertex[x=-2,y=1.5]{T21}
\Vertex[x=-2,y=0.5]{T22}
\Vertex[x=-2,y=-0.5]{T23}
\Vertex[x=-2,y=-1.5]{T24}

\Vertex[x=0,y=0]{O}

\Vertex[x=1,y=1]{M11}
\Vertex[x=1,y=-1]{M12}
\Vertex[x=2,y=1.5]{M21}
\Vertex[x=2,y=0.5]{M22}
\Vertex[x=2,y=-0.5]{M23}
\Vertex[x=2,y=-1.5]{M24}

\Vertex[x=3.5,y=1.5]{H11}
\Vertex[x=3.5,y=0.5]{H12}
\Vertex[x=3.5,y=-0.5]{H13}
\Vertex[x=3.5,y=-1.5]{H14}
\Vertex[x=4.5,y=1]{H21}
\Vertex[x=4.5,y=-1]{H22}
\Vertex[x=5.5,y=0]{E}

\Vertex[x=6.5,y=1]{AR11}
\Vertex[x=6.5,y=-1]{AR12}
\Vertex[x=7.5,y=1.5]{AR21}
\Vertex[x=7.5,y=0.5]{AR22}
\Vertex[x=7.5,y=-0.5]{AR23}
\Vertex[x=7.5,y=-1.5]{AR24}

\Edges(O,M11,M21,H13,H22,E)
\Edges(O,M11,M22,H11,H21,E)
\Edges(O,M12,M23,H14,H22,E)
\Edges(O,M12,M24,H12,H21,E)

\Edges(M21, H11)
\Edges(M22, H14)
\Edges(M23, H12)
\Edges(M24, H13)

\Edges(T11, O)
\Edges(T12, O)
\Edges(T11, T21)
\Edges(T11, T22)
\Edges(T12, T23)
\Edges(T12, T24)
\Edges(T11, T21)
\Edges(T11, T21)

\Edges(E, AR11)
\Edges(E, AR12)
\Edges(AR11, AR21)
\Edges(AR11, AR22)
\Edges(AR12, AR23)
\Edges(AR12, AR24)

\doublebowtie{(-8,-6)}{1.2}
\node[text width=5cm, align=center] at (-2.5,-6) {$\boldsymbol{=}$};
\node[text width=5cm, align=center, rotate=90] at (-5.2,-6) {Body};
\node[text width=5cm, align=center, rotate=90] at (-6,-6) {Body};
\node[text width=2cm, align=center, rotate = 90] at (-7.8,-6) {Antenna};
\node[text width=2cm, align=center, rotate = 90] at (-3.5,-6) {Antenna};

\Vertex[x=-1,y=-5]{TT11}
\Vertex[x=-1,y=-7]{TT12}
\Vertex[x=-2,y=-4.5]{TT21}
\Vertex[x=-2,y=-5.5]{TT22}
\Vertex[x=-2,y=-6.5]{TT23}
\Vertex[x=-2,y=-7.5]{TT24}

\Vertex[x=0,y=-6]{OO}

\Vertex[x=1,y=-5]{MM11}
\Vertex[x=1,y=-7]{MM12}
\Vertex[x=2,y=-4.5]{MM21}
\Vertex[x=2,y=-5.5]{MM22}
\Vertex[x=2,y=-6.5]{MM23}
\Vertex[x=2,y=-7.5]{MM24}

\Vertex[x=3.5,y=-4.5]{HH11}
\Vertex[x=3.5,y=-5.5]{HH12}
\Vertex[x=3.5,y=-6.5]{HH13}
\Vertex[x=3.5,y=-7.5]{HH14}
\Vertex[x=4.5,y=-5]{HH21}
\Vertex[x=4.5,y=-7]{HH22}
\Vertex[x=5.5,y=-6]{EE}

\Vertex[x=6.5,y=-5]{ARAR11}
\Vertex[x=6.5,y=-7]{ARAR12}
\Vertex[x=7.5,y=-4.5]{ARAR21}
\Vertex[x=7.5,y=-5.5]{ARAR22}
\Vertex[x=7.5,y=-6.5]{ARAR23}
\Vertex[x=7.5,y=-7.5]{ARAR24}

\Edges(OO, MM11)
\Edges(MM11,MM21)
\Edges(MM11,MM22)
\Edges(OO,MM12)
\Edges(MM12,MM23)
\Edges(MM12,MM24)

\Edges(EE,HH21)
\Edges(EE,HH22)
\Edges(HH21,HH11)
\Edges(HH21,HH12)
\Edges(HH22,HH13)
\Edges(HH22,HH14)

\Edges(TT11, OO)
\Edges(TT12, OO)
\Edges(TT11, TT21)
\Edges(TT11, TT22)
\Edges(TT12, TT23)
\Edges(TT12, TT24)
\Edges(TT11, TT21)
\Edges(TT11, TT21)

\Edges(EE, ARAR11)
\Edges(EE, ARAR12)
\Edges(ARAR11, ARAR21)
\Edges(ARAR11, ARAR22)
\Edges(ARAR12, ARAR23)
\Edges(ARAR12, ARAR24)

\draw [line width=0.25mm] (MM21) to [out=-60,in=60] (MM24);
\draw [line width=0.25mm] (MM24) to [out=120,in=240] (MM22);
\draw [line width=0.25mm] (MM22) to [out=-60,in=60] (MM23);
\draw [line width=0.25mm] (MM23) to [out=120,in=240] (MM21);

\draw [line width=0.25mm] (HH11) to [out=-60,in=60] (HH12);
\draw [line width=0.25mm] (HH14) to [out=120,in=240] (HH12);
\draw [line width=0.25mm] (HH13) to [out=-60,in=60] (HH14);
\draw [line width=0.25mm] (HH13) to [out=120,in=240] (HH11);

\node[text width=3cm, align=center] at (5.5,2) {\rooot};
\node[text width=3cm, align=center] at (0,2) {\rooot};

\draw [->, line width=1pt] (0,1.5) -- (0,.25);
\draw [->, line width=1pt] (5.5,1.5) -- (5.5,.25);

\node[text width=3cm, align=center] at (5.5,-8) {\rooot};
\node[text width=3cm, align=center] at (0,-8) {\rooot};

\draw [->, line width=1pt] (0,-7.5) -- (0,-6.25);
\draw [->, line width=1pt] (5.5,-7.5) -- (5.5,-6.25);

\draw[<->, line width=1pt] (-2,-3) -- (-0.25,-3);

\node[text width=1.5cm, align=center] at (-1.125,-3.5) {$k$};
\node[text width=1.5cm, align=center] at (-1.125,-2.5) {\small Antenna};

\draw[<->, line width=1pt] (0.25,-3) -- (1.5,-3);
\node[text width=1.5cm, align=center] at (0.875,-3.5) {$k-1$};
\node[text width=1.5cm, align=center] at (0.875,-2.25) {\small Body interior};

\draw[<->, line width=1pt] (1.5,-3) -- (4,-3);
\node[text width=1.5cm, align=center] at (2.75,-3.5) {$2$};
\node[text width=1.5cm, align=center] at (2.725,-2.25) {\small Body weld};

\draw[<->, line width=1pt] (4,-3) -- (5.25,-3);

\node[text width=1.5cm, align=center] at (4.625,-3.5) {$k-1$};
\node[text width=1.5cm, align=center] at (4.625,-2.25) {\small Body interior};
\draw[<->, line width=1pt] (5.75,-3) -- (7.5,-3);

\node[text width=1.5cm, align=center] at (6.625,-3.5) {$k$};
\node[text width=1.5cm, align=center] at (6.625,-2.5) {\small Antenna};
\end{tikzpicture}
\caption{An illustration of a candy graph in the yes instance (top right) and our shorthand symbol for it (top left), and a double-bow-tie graph which appears in the particular no instances considered for our classical lower bound proof (bottom right) and our shorthand symbol for it (bottom left).}
\label{fig:modified_glued_tree}
\end{figure}
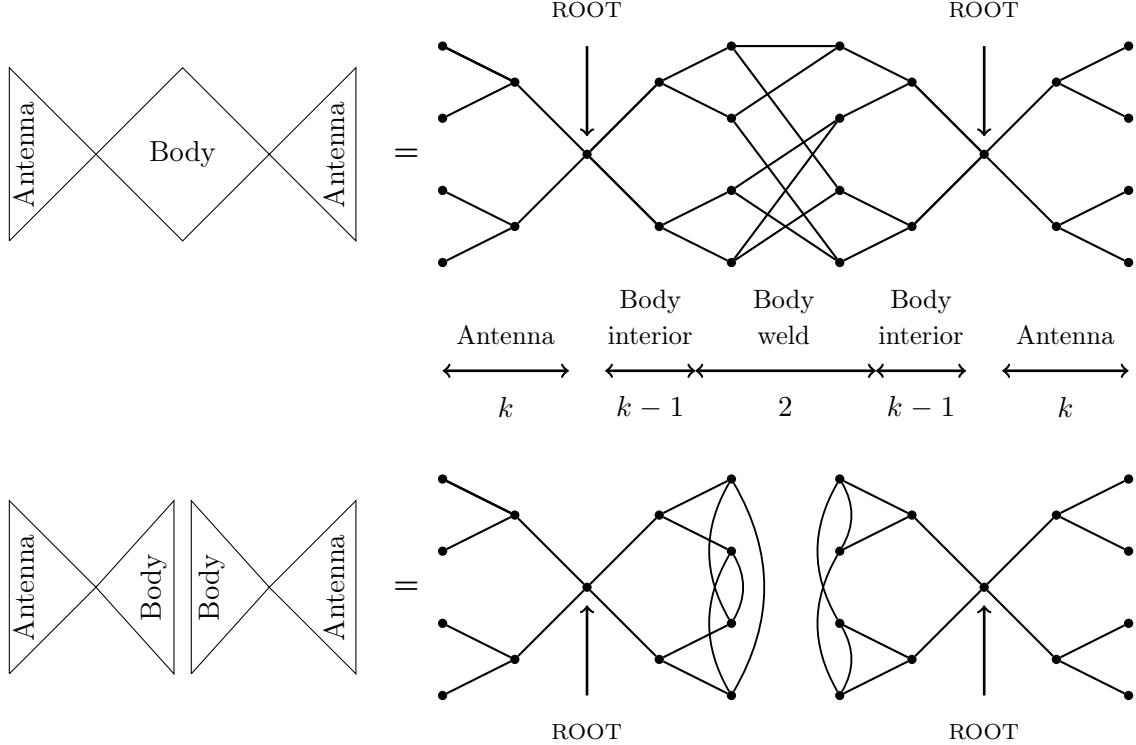

The graph property $\mathcal{P}_k$ that we want to test is the following. A graph has the property $\mathcal{P}_k$ if its single-edges form a graph that contains $2j(2^k-1)$ disconnected instances of ``candy'' graphs for some $j\in \mathbb{N}$, as shown in \fig{modified_glued_tree}. A candy graph is obtained from $4$ binary trees of depth $k$: two ``antenna'' trees $A_1,A_2$ and two ``body'' trees $B_1,B_2$. The candy graph is formed by first connecting the leaves of $B_1$ and $B_2$ by a collection\footnote{In the lower bound proof we assume that there is a single long cycle, but this is hard to test, so for the definition of the property we allow smaller cycles as well.} of alternating cycles containing all leaves of the body trees, and then merging the roots of $A_1-B_1$ and $A_2-B_2$. Additionally, in the full graph, roughly half of the candy (sub)graphs get a self-loop on exactly one of the roots, and the other candy graphs have either no self-loops or both roots get a self-loop.\footnote{Strictly speaking, a $2^{k}/(2^{k+1}-2)$ fraction of candy graphs should get exactly one self-loop and the remaining fraction should be split between getting two or zero self-loops, because there are $2\cdot 2^k$ weld vertices while only $2\cdot (2^k-2)$ interior vertices in each candy graph, cf.~the following main text. With a slight modification to the construction this could be balanced out, for example by replacing the roots by a path of length 5, and treating the middle vertex of this length-5 path as the new root.} There is also a double edge attached to every non-root vertex of the candy graphs, namely every body vertex is connected to a distinct antenna vertex by a double edge in the following way:
if the body vertex is a ``weld'' vertex (i.e., it was a leaf of a body tree), then it gets attached to an antenna vertex inside a candy graph that has exactly one root with self-loop, while every interior vertex (i.e., non-weld and non-root) gets attached to an antenna vertex in a candy graph where the parity of the number of self-loops is even.

\begin{figure}[ht]
\centering
\begin{tikzpicture}[circ/.style={shape=circle, inner sep=3pt, draw, node contents=}]
\SetGraphUnit{1.5}
\GraphInit[vstyle=Simple]
\tikzset{VertexStyle/.style = {shape = circle,fill = black,minimum size = 5pt,inner sep=1pt}}
\candy{(0,0)}{1}
\candy{(0,2.5)}{1}
\candy{(0,5)}{1}
\candy{(5,0)}{1}
\candy{(5,2.5)}{1}
\candy{(5,5)}{1}

\Vertex[x=5.5,y=2.5] {1};
\Vertex[x=2.5,y=2.5] {2};
\Vertex[x=3.5,y=5] {3};
\Vertex[x=6.5,y=2.5] {4};
\Vertex[x=7,y=2.5] {5};
\Vertex[x=8.5,y=5] {6};
\Vertex[x=0.5,y=0] {7};
\Vertex[x=2,y=2.5] {8};
\Vertex[x=5.5,y=0] {9};
\Vertex[x=7.5,y=0] {10};
\draw [-, double] (1) to [out=150,in=30] (2);
\draw [-, double] (3) to [out=-5,in=100] (4);
\draw [-, double] (5) to [out=80,in=-140] (6);
\draw [-, double] (7) to [out=80,in=-140] (8);
\draw [-, double] (9) to [out=45,in=135] (10);

\draw node () at (1, 5) [circ];
\draw node () at (1, 2.5) [circ];
\draw node () at (3, 0) [circ];
\draw node () at (3, 5) [circ];
\draw node () at (6, 0) [circ];
\draw node () at (6, 5) [circ];
\draw node () at (8, 0) [circ];

\end{tikzpicture}
\caption{Multiple candy (sub)graphs which together form a yes-instance graph with property $\mathcal{P}_k$. In the particular no instances we consider, the candy graphs are replaced by double-bow-tie graphs. Here, we only show six of them and only five of the many ``advice edges'' (indicated by double lines) that connect each body vertex to a distinct antenna vertex. The circles in the figure represent self-loops at the roots of the candy graphs, which provide advice about whether a body vertex is in the interior or the weld. An even parity of circles indicates interior, while an odd parity indicates weld.}
\label{fig:schematic_full_graph}
\end{figure}
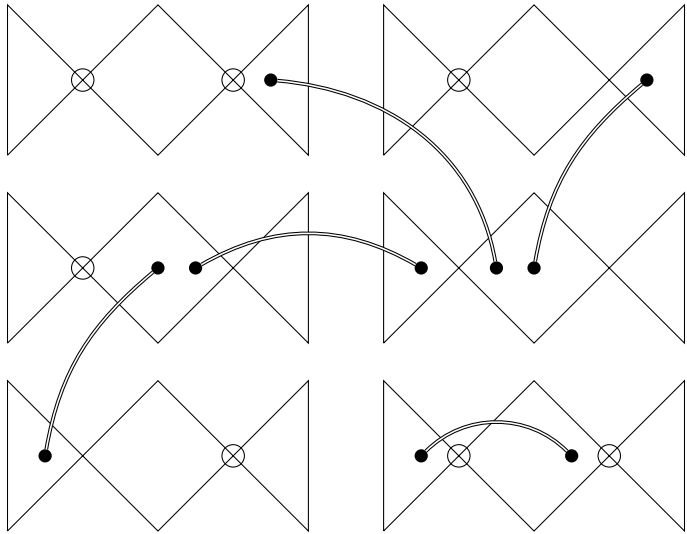

While we do not claim that testing the above graph property is particularly useful, we can prove that a quantum computer has an exponential advantage in testing this property. The intuition is the following: this graph property is exponentially difficult to test on a classical computer because the welded trees locally look like exponentially deep trees, unless the weld vertices can be distinguished. This is where quantum computers get their advantage: a quantum computer can ``read'' the advice hidden in the structure of the double edges very efficiently using the quantum walk algorithm of Childs et al.~\cite{CCD+03}. Once we can tell apart the weld vertices from the other vertices, testing the structure could be performed even on a classical computer.

We prove the following theorem in this section:

\begin{theorem}\label{thm:adjacency_list}
	In the adjacency list model, there exists a constant $\epsilon$ such that testing whether a bounded-degree graph has property $\mathcal{P}_k$ or is $\epsilon$-far from having it can be done by a quantum algorithm using $Q(\mathcal{P}_k) = \poly(k)$ queries, with constant success probability and perfect completeness. On the other hand, any classical randomized algorithm that can test the property with bounded two-sided error needs $R(\mathcal{P}_k) = \exp(\Omega(k))$ queries.
\end{theorem}

\begin{proof}[Proof outline]
	First we prove in \cref{subsec:effAdviceTest} that the property $\mathcal{P}_k$ can be tested with perfect completeness and constant success probability with $\poly(k)$ queries, given the proper advice. Then in \cref{subsec:MarkWeld} we show that for graphs having property $\mathcal{P}_k$ the advice can be computed on a quantum computer with probability $1$. This completes the proof of our efficient quantum tester.

	Finally, in \cref{subsec:HardToDistinguish} we show that there are two particular distributions of yes and no instances that are exponentially difficult to distinguish on a classical computer, where the considered no instances are constant-far from having property $\mathcal{P}_k$ with high probability. This rules out any sub-exponential-time classical property tester of $\mathcal{P}_k$ that succeeds with constant probability.
\end{proof}

\subsection{Efficient classical testability with advice} \label{subsec:effAdviceTest}

In this subsection, we describe how to test the property $\mathcal{P}_k$ efficiently with advice. First, we rigorously define this notion, which has a similar flavor to the complexity classes $\mathsf{NP}$ and $\mathsf{MA}$. Then, we proceed by showing an efficient tester with advice.

The required advice is precisely what our quantum algorithm can provide as per the above subsection: marking the weld vertices.
Formally, the advice is a bit associated with every vertex, and for yes instances the advice is supposed to mark weld vertices.

There are two major parts of our analysis. First, we show that the test always accepts yes instances if the advice is correctly provided. Then we show that if the test accepts a graph with high probability, then the graph must actually be close to a yes instance. The latter implies that no instances that are far from complying with the property must actually be rejected with high probability.

\subsubsection{Property testing with advice}

We say that a property $\mathcal{P}$ can be tested with advice\footnote{Note that this terminology differs from the standard use of ``advice'' in complexity theory, which typically refers to a string that can only depend on the input length.} in $Q$ queries if there is a tester $\mathcal{T}$ that is
\begin{itemize}[labelwidth=5.5em,leftmargin=!]
	\item[Complete.] For every instance $x\in \mathcal{P}$ there is a witness $w\in\{0,1\}^*$, such that $\mathcal{T}$ accepts $(x,w)$ with probability at least $2/3$ and makes at most $Q$ queries to $x$ and $w$.

	\item[Sound.] For every instance $x$ that is at least $\eps$-far from $\mathcal{P}$, and for any witness $w\in\{0,1\}^*$, $\mathcal{T}$ accepts $(x,w)$ with probability at most $1/3$.
\end{itemize}
Additionally, if for every yes instance there is a witness that makes the tester accept with certainty, we say that the tester has perfect completeness.

\subsubsection{Testing the binary trees, the weld, and the advice}

In this subsection, we describe our classical tester using the advice. We build increasingly more complex tests, such that passing them enforces more and more structure on the graph.
Our final tester ensures that any graph passing the test with high probability must be $\eps$-close to a yes instance. The query and time complexity of our final tester is $\poly(k/\eps)$.

As a first step we test the binary tree structures that are supposed to remain after removing every double edge, self-loop, and edge between marked (therefore supposedly weld) vertices.
Therefore, while testing the binary tree structures---i.e., in \alg{FindParent}, \alg{FindRootPath}, and the consistency check below---we ignore those edges (whenever we say that we ignore some edges we also exclude them from degree counts).
Also, if we encounter any vertex that has more than $4$ single-edge neighbors (excluding self-loops), then we immediately reject the graph, so we can also assume without loss of generality that the modified graph has degree at most $4$ in \alg{FindParent} and \alg{FindRootPath}.
Similarly, if \alg{FindParent} or \alg{FindRootPath} returns REJECT, we immediately reject the instance.

The basic primitive of our tester is a binary-tree tester described in \alg{FindParent}, based on non-backtracking walks.
The main idea of the tester is that for any non-root vertex $v$ at level $\ell$ (i.e., distance $\ell$ from the root) in a binary tree of depth $k$, any non-backtracking walk starting towards the parent of $v$ can continue for at least $k-\ell+2$ steps, while any other non-backtracking walk must terminate after $k-\ell$ steps. \linebreak[2]
Based on this observation we can find the parent of $v$ efficiently. \alg{FindParent} summarizes this procedure for finding the parent vertex, and adds some consistency checks that are helpful for catching potential corruptions of the structure.

\alg{FindRootPath} recursively applies \alg{FindParent} to find a path to the root.
Although this recursive structure may not be the most efficient way of finding such a root-path, it allows us to argue that a root-path found for a vertex should match the root-path of its parent, even in slightly corrupted trees.

\begin{algorithm}[ht]
	\caption{FindParent(v)}\label{alg:FindParent}
	\begin{algorithmic}[1]
		\STATEx {\bf input:} vertex $v$ (in a graph with maximum degree at most $5$)
		\STATEx {\bf output:} parent vertex $u$ (if the graph is a binary tree of depth $k$)
		\STATE {\bf if} $\deg(v)=4$ {\bf then return} $\mathsf{ROOT}$ \hfill($\star$ in a binary tree only the root has degree 4 $\star$)
		\STATE {\bf if} $\deg(v)\notin\{1,3,4\}$ {\bf then} $\mathsf{REJECT}$ \label{step:BannedDegrees} \hfill($\star$ every vertex is expected to have $1$, $3$, or $4$ neighbors $\star$)
		\STATE $r\leftarrow \emptyset$ \hfill($\star$ stores a root candidate that is any degree-$4$ vertex encountered $\star$)
		\STATE {\bf for each} $w \in \mathrm{adj}(v)$ {\bf do} \hfill($\star$ if we get here we know that the degree of $v$ is 1 or 3 $\star$)
		\STATE ~~~ $p_0^{(w)}\leftarrow v$ \hfill($\star$ the starting vertex of the non-backtracking walk $\star$)
		\STATE ~~~ $p_1^{(w)}\leftarrow w$ \hfill($\star$ the first vertex on the path in the direction towards $w$ $\star$)
		\STATE ~~~ $\ell^{(w)}\leftarrow 2k$ \hfill($\star$ upper bound on the length of a non-backtracking walk $\star$)
		\STATE ~~~ {\bf for each} $t=1,\ldots, 2k-1$ {\bf do} \hfill($\star$ the steps of the non-backtracking walk $\star$)
		\STATE ~~~~~~ {\bf if} $\deg(p_t^{(w)})=1$ {\bf then} \hfill($\star$ a leaf is found $\star$)
		\STATE ~~~~~~~~~ $\ell^{(w)}\leftarrow t$ \hfill($\star$ store the length of the path $\star$)
		\STATE ~~~~~~~~~ {\bf quit loop and goto line} \ref{line:endfor2} \hfill($\star$ quit the loop and end the walk $\star$)
		\STATE ~~~~~~ {\bf if} $\deg(p_t^{(w)})=4$ {\bf then} \hfill($\star$ a root candidate is found $\star$)
		\STATE ~~~~~~~~~ {\bf if} $r=\emptyset$ {\bf then} $r\leftarrow p_t^{(w)}$ \hfill($\star$ update root candidate $\star$)
		\STATE ~~~~~~~~~ {\bf else} $\mathsf{REJECT}$ \hfill($\star$ there is a unique root $\star$)
		\STATE ~~~~~~ $p_{t+1}^{(w)}\leftarrow$ a uniformly random vertex in $\mathrm{adj}(p_{t}^{(w)})\setminus\{p_{t-1}^{(w)}\}$ \hfill($\star$ make a random step $\star$)
		\STATE ~~~ {\bf end for}
		\STATE ~~~ {\bf if} $\deg(p_{\ell^{(w)}}^{(w)})>1$ {\bf then} $\mathsf{REJECT}$ \label{line:endfor2} \hfill($\star$ verify that a leaf is found $\star$)
		\STATE {\bf end for}
		\STATE $u\leftarrow \argmax_{w}(\ell^{(w)})$ \hfill($\star$ the longest path must go through the parent $\star$)
		\STATE {\bf if} $\exists w_1, w_2\in \mathrm{adj}(v)\setminus\{u\}$ such that $\ell^{(w_1)}\neq \ell^{(w_2)} $ {\bf then} $\mathsf{REJECT}$ \label{step:LengthCheck} \hfill($\star$ consistency check $\star$)
		\STATE {\bf if} $r\neq  \emptyset \wedge r\notin \{p_{1}^{(u)},p_{2}^{(u)},\ldots, p_{\ell^{u}}^{(u)}\} $ {\bf then} $\mathsf{REJECT}$ \label{step:NoRootSkip} \hfill($\star$ consistency check $\star$)
		\STATE {\bf return} $u$
	\end{algorithmic}
\end{algorithm}

\begin{algorithm}[ht]
	\caption{FindRootPath(v)}\label{alg:FindRootPath}
	\begin{algorithmic}[1]
		\STATEx {\bf input:} vertex $v$ (in a graph with maximum degree at most $5$)
		\STATEx{\bf output:} a path to the root (if the graph is a binary tree of depth $k$)
		\STATE $p_0\leftarrow v$ \hfill($\star$ $v$ is the first vertex on the path to the root $\star$)
		\STATE $\ell \leftarrow k$ \hfill($\star$ the path to the root cannot be longer than $k$ $\star$)
		\STATE{\bf for each} $t=1,\ldots, k$ {\bf do} \hfill($\star$ find a path to the root step-by-step $\star$)
		\STATE ~~~ $p_{t}\leftarrow \mathrm{FindParent}(p_{t-1})$ \hfill($\star$ find the parent of the current vertex $\star$)
		\STATE ~~~ {\bf if} $p_{t}=\mathsf{ROOT}$ {\bf then} \hfill($\star$ the root is found $\star$)
		\STATE ~~~~~~~~~ $\ell\leftarrow t$  \hfill($\star$ store the length of the path $\star$)
		\STATE ~~~~~~~~~ {\bf quit loop and goto} \ref{line:endfor2x} \hfill($\star$ store the length of the path and stop $\star$)
		\STATE {\bf end for}
		\STATE {\bf if} $\deg(p_{\ell})\neq 4$ {\bf then} $\mathsf{REJECT}$ \label{line:endfor2x} \hfill($\star$ verify that a root is found $\star$)
	    \STATE {\bf return}	$(p_0,p_1,\ldots,p_\ell)$
	\end{algorithmic}
\end{algorithm}

\paragraph{Consistency test.} During this test we ignore every double edge and self-loop, and also any edge between marked vertices.
Given vertex $v$, reject if $v$ is marked and has degree greater than $1$. Run \cref{alg:FindRootPath} to find a path to the root (degree-4 vertex).
If the root-path is longer than $k$, or any vertex (other than $v$) is marked on the root-path, then reject.
Also reject if $v$ is marked and the path is shorter than $k$.
Repeat this procedure $100$ times and reject unless always the same path is found.
If this consistency check or any individual run fails, then reject, otherwise accept.

\begin{definition}
	A vertex $v$ is \emph{consistent} if it passes the consistency test with probability at least $1/2$.
	For a consistent vertex, there is a \emph{consistent root} and a \emph{consistent root-path} to the consistent root that \alg{FindRootPath} finds with probability at least $0.99$.
\end{definition}

\begin{lemma}
	For any edge $e$ between two neighboring consistent vertices $u$ and $v$, one of the consistent root-paths, say $(u,p_1,\ldots,p_\ell)$, must be a concatenation of the edge $e$ and the other consistent root-path, i.e., $p_1=v$.
\end{lemma}
\begin{proof}
	First note that due to \cref{step:BannedDegrees}, any consistent vertex $u,v$ must have degree $1$, $3$, or $4$.
	If $v$ has degree $4$, then for any neighbor $u$ the consistent root-path must be $(u,v)$ due to \cref{step:NoRootSkip}.
	Also note that the statement is trivial if either $u$ or $v$ has degree $1$.

	Finally, we treat the case when both $u$ and $v$ have degree $3$.
	Towards a contradiction let us assume that the statement of the lemma does not hold.
	Let $v_c\neq u$ denote the neighbor of $v$ which is not its immediate parent on the consistent root-path, and similarly let $u_c\neq v$ denote the neighbor of $u$ which is not its immediate parent on the consistent root-path.
	Due to \cref{step:LengthCheck}, a non-backtracking walk from $u$ in the direction $u_c$ has typically (with probability at least $0.99$) the same length as a non-backtracking walk in the direction $v$.
	This also means that there is a fixed length that has high probability ($> 0.98$), which we denote by $\ell_u$.
	We denote by $\ell_u$ the analogous typical length starting from $u$ towards either $u_c$ or $v$.
	When we start a non-backtracking walk in the direction of $u$ we continue with a non-backtracking walk towards $u_c$ with probability at least $1/3$.
	This implies that $\ell_v>\ell_u$.
	Since the situation is symmetric we can also deduce that $\ell_v<\ell_u$ which is a contradiction.
\end{proof}

It is easy to see that for any consistent vertex $v$ all vertices on its consistent root-path are also consistent vertices. This is due to the recursive structure of \alg{FindRootPath}.

Consequently we can conclude that the graph $G_C$ induced by consistent vertices is a (sub-)binary forest. More precisely, every connected component can be injectively mapped into a graph obtained by merging the roots of two binary trees of depth $k$.
The mapping is defined via the root-paths, and also ensures that marked consistent vertices get mapped to leaves.
Moreover, for any $\delta>0$, if the consistency test is passed with probability $1-\delta$ for a uniformly random vertex $v\in V$, then there are at least $(1-2\delta)|V|$ consistent vertices, since by definition non-consistent vertices fail to pass the consistency check with probability at least $1/2$.

\paragraph{Weld-consistency test.} During this test we ignore every double edge and self-loop (additionally, when we call the consistency test as a subroutine, we also ignore edges between marked vertices).
Run the consistency test on $v$ and denote by $r$ the root found.
If the length of the root-path is less than $k$ (i.e., $v$ is not a ``leaf''), then accept.
Let $e$ denote the last edge of the root-path from $v$ to $r$.
\begin{itemize}
	\item If $v$ is not marked, then go to the root $r$ and perform $100$ non-backtracking random walks of length $k$ along $e$.
	If any walk terminates before $k$ steps or any marked vertex is found, then reject, otherwise accept.
	\item If $v$ is marked, then reject unless exactly $2$ of its neighbors are also marked.
	Run the consistency test on both marked neighbors, and denote the roots $r'$ and $r''$. Reject if $r'\neq r''$ or $r=r'$.
	Perform a non-backtracking random walk from $r$ along all 4 edges of length $k$; if any walk stops before $k$ steps or not exactly $2$ out of the $4$ end-points is marked, then reject.
	Also reject if the walk along $e$ does not end in a marked vertex. Reject unless both marked end-points have exactly two marked neighbors; traverse to a random marked neighbor of both marked end-points, run a consistency check, and report its root.
	If the consistency check fails or either of the found roots does not match $r'$, then reject.
	Repeat the process of this paragraph $100$ times.
\end{itemize}

\begin{definition}
	A vertex $v$ is \emph{weld-consistent} if it passes the above test with probability at least $1/2$.
\end{definition}
We already established that consistent vertices form a (sub-)binary forest, where some leaves at level $k$ may be marked.
The weld-consistency test might remove some of the marked vertices, but ensures that the (sub-)binary trees are paired up via marked-marked vertex edges.

\begin{lemma}\label{lem:weld-pairs}
	All weld-consistent marked vertices in a (sub-)binary tree $t_1$ are connected to at most 2 consistent marked vertices in another (sub-)binary tree $t_2$ of consistent vertices.
\end{lemma}
\begin{proof}
Indeed, for any weld-consistent vertex $v$, its consistent root $r$ has the property that a non-backtracking walk of length $k$ along the $4$ outgoing edges after traversing the $2$ marked leaves finds the same root $r'$ with high probability.
We call $r'$ the \emph{consistent-root-pair} of $r$. It is easy to see that all consistent marked neighbors of $v$ must have the same consistent root as the consistent-root-pair of $r$.
\end{proof}

By \cref{lem:weld-pairs} we have established that weld-consistent vertices form pairs of (sub-)binary trees, that are connected by edges between marked vertices at level $k$, with each marked vertex having at most $2$ marked neighbors. Moreover, these (sub-)binary candy graphs are isomorphic to an induced subgraph of an actual candy graph. Also at any root, at most two of the branches contain marked leaves. We call such a graph a (sub-)binary candy ensemble.

Similarly as before, for any $\delta>0$, if the weld-consistency test is passed with probability $1-\delta$ for a uniformly random vertex $v\in V$, then there are at least $(1-2\delta)|V|$ weld-consistent vertices. This way we can ensure that there are many weld-consistent vertices and that they form a (sub-)binary candy ensemble. By the following test we also ensure that a typical connected component is close to being a complete candy graph.

\paragraph{Completeness test.} During this test we still ignore every double edge and self-loop.
Given vertex $v$, run the weld-consistency test, and then go to the root $r$.
From $r$, along each edge perform a non-backtracking walk. Reject if any walk fails to find a leaf (degree-1 vertex) or a marked vertex (with 2 marked neighbors) after exactly $k$ steps, or any non-degree-$3$ vertex is encountered other than a leaf.
Also reject if there are not exactly $2$ marked vertices among the $4$ reached end-points.
Traverse to a random marked neighbor of a randomly chosen marked end-point, find its root-path, and let $r'$ denote the root that is found.
Also run the weld-consistency check for every vertex on every sampled path and reject if any of the vertices fail the weld-consistency check.
Go to $r'$ and run the same test, except for traversing a marked-marked edge at the end.
Repeat the whole process described in this paragraph $\Theta(k/\eps)$ times.

\begin{lemma}
	If $v$ is a weld-consistent vertex within a (sub-)binary candy graph induced by fewer than $(1-\eps) (2^{k+3}-6)$ weld-consistent vertices, then the completeness test rejects with high probability.
\end{lemma}
\begin{proof}
	Let us consider the connected component of $v$ formed by weld-consistent vertices. We already established that this connected component can be embedded into a candy graph. If at least an $\eps$-fraction of the vertices of the candy graph are missing from the connected component of $v$, then we must encounter a non-weld consistent vertex or a degree-deficient vertex (i.e., less than degree-3 interior or weld vertex) with probability $\Omega(\eps/k)$ in each of the $\Theta(k/\eps)$ runs (since by the pigeonhole principle there is a level set of the candy graph, where at least an $\Omega(\eps/k)$-fraction of the vertices are missing). This implies the statement of the lemma.
\end{proof}

If there are more than $\eps |V|$ consistent vertices that are in a sub-binary tree of weld-consistent vertices of size at most $(1-\eps)(2^{k+1}-1)$, then the completeness check fails with probability at least $\eps$. Therefore, we can see that if the graph passes the completeness check with probability at least $1-\eps$, then the graph is $O(\eps)$-close to the desired structure of having a disjoint union of welded trees, when disregarding self-loops and double edges.
(One can actually form the desired structure by removing non-consistent vertices, and then completing the potentially incomplete trees missing $O(\eps)$ edges.)

\paragraph{Advice test.} Given vertex $v$, run the completeness test, and reject unless it has precisely one double edge or has $4$ single-edge neighbors, with the potential addition of a self-loop (root vertex).
If $v$ has a double edge neighbor $u$, then run the completeness test on $u$.
Reject unless precisely one of $u$ and $v$ appear to be in a branch with non-marked leaves in the completeness test.
Let $a$ denote the one in the non-marked branch, and $b$ the other.
Check the parity of the number of loops of $a$'s root $r$ and the root-pair $r'$ of $r$ found during the completeness test.
Reject if the parity does not match the marking of $b$.

\begin{definition}
	A vertex $v$ is \emph{advice-consistent} if it passes the advice test with probability at least $1/2$.
\end{definition}
We can see that if the graph passes the advice test with probability at least $1-\eps$, then at most an $O(\eps)$-fraction of the weld-consistent vertices have an incorrect advice edge.
Thus, if the number of vertices is a multiple of $2 \cdot (2^k-1)\cdot(2^{k+3}-6)$, then we can just take the induced subgraph of the advice-consistent vertices and complete it using $O(\eps)$ edges to get to a yes instance with property~$\mathcal{P}_k$.\footnote{Note that it is possible, that one also needs to remove an $O(\eps)$-fraction of the candy-components, if there is a surplus of candies with even/odd number of marker loops.}

\paragraph{The final test.} Verify that the number of vertices is a multiple of $2 \cdot (2^k-1)\cdot(2^{k+3}-6)$, then run the advice test $O(1/\eps)$ times.
If we have a yes instance with the correct advice, then the test always accepts, so the tester has perfect completeness.
On the other hand, if the test accepts with probability at least $1/3$, then we know that the graph must be $\eps$-close to a yes instance. Therefore, any $\eps$-far no instance will be rejected with probability at least $2/3$.

\subsection{Marking the weld vertices using a quantum computer} \label{subsec:MarkWeld}

In this subsection, we describe how to efficiently mark weld vertices in a graph $G\in\mathcal{P}_k$. More precisely, given a vertex $v$ in $G$ the task is to tell whether it is a weld vertex or not, in query and time complexity $\poly(k)$. If the graph $G$ does not have property $\mathcal{P}_k$ we do not place any restriction on the output, so we will always assume that $G\in\mathcal{P}_k$. The procedure that we describe below has success probability $1$.

Let us denote by $W$ the set of weld vertices in a graph $G$ with property $\mathcal{P}_k$. Given a vertex $v$, we decide whether or not $v\in W$ as follows:
\begin{enumerate}[label=\arabic*.)]
	\item Along every non-double edge adjacent to $v$, start a non-backtracking (random) walk, which only traverses single edges. If a degree-$3$ vertex (i.e., a leaf with a double edge) is encountered within the first $k$ steps of any of these walks (including $v$ itself), then conclude $v\notin W$.
	\item Now we know that $v$ is a non-root body vertex. We traverse the double edge to the antenna vertex $a$ in antenna $A$. From now on we completely ignore double edges (and even exclude them from degree counts), and use \alg{FindRootPath} to find the root $r_a$ of the antenna $A$.
	\item From the root $r_a$ we run the quantum walk algorithm of Childs et al.~\cite{CCD+03} to find the other root $r_b$ in the candy graph. We conclude $v\in W$ if exactly one of $r_a,r_b$ has a self-loop, otherwise we conclude $v\not\in W$.
\end{enumerate}

Now we briefly prove the correctness of the above procedure. For the first step, note that the degree-$3$ vertices are precisely the leaves of the antennas. Moreover, if a non-backtracking walk is started from an antenna vertex towards the leaves, then it will always find a leaf within $k$ steps. On the other hand, any body vertex is at least $k+1$ steps from an antenna leaf, if only single-edges can be traversed.

The second step relies on the correctness of \alg{FindRootPath}, which we already analyzed in the previous section.

The third step deserves a bit more explanation. We would like to use the quantum algorithm of Childs et al.~\cite{CCD+03} on the body part of the candy graph. As we discuss in the next paragraph, for this it suffices to construct an adjacency list ``oracle'' for a graph, where all vertices have bounded degree (say, $3$) and the body welded tree structure containing $r_a$ is one of the connected components. We construct this ``oracle'' starting from the input oracle and removing double edges. Then we disconnect the antennae by treating the roots (vertices with at least $4$ non-double edges) in a special way. From a root we start a non-backtracking (random) walk of length $k$ along each edge. If we find a leaf after $k$ steps, then we know that the initial edge leads to the antenna vertex and therefore we remove the corresponding edge adjacent to the root.\footnote{Strictly speaking this results in a corrupted adjacency list oracle, because we did not delete the edge from the neighbor list of the vertex in the antenna connected to the root. This could be easily fixed with a bit of caution, but we note that the edge also gets automatically deleted if we use the block-encoding framework~\cite{GSLW18}.}

Once we have the adjacency list oracle for the modified graph, we can construct a block-encoding of the modified adjacency matrix $A$ divided by the maximum degree ($=3$) using standard techniques (see for example~\cite{GSLW18}). Once we have a block-encoding we can also perform Hamiltonian simulation~\cite{BCK15,LC19} efficiently (or implement polynomials of $A$~\cite{GSLW18}), which in turn enables implementing the algorithm of~\cite{CCD+03}. Since the quantum walk algorithm of~\cite{CCD+03} finds the other root in the tree with probability $\Omega(1/\poly(k))$, and the success probability can also be computed,
we can make the algorithm succeed with probability $1$ using amplitude amplification.

Once we have found the other root in the candy graph we can easily compute the parity of the number of self-loops, which in turn correctly identifies the weld vertices, since $G\in\mathcal{P}_k$. As all steps of the above procedure succeed with certainty, we get the following.

\begin{lemma}\label{lem:Mark}
	The above quantum algorithm given any vertex $v$ in a graph $G\in\mathcal{P}_k$ outputs a bit which is $1$ iff $v$ is a weld vertex. Moreover, the algorithm has query and time complexity $\poly(k)$.
\end{lemma}

\subsection{Classical lower bound} \label{subsec:HardToDistinguish}

We show that our problem is classically hard by showing that it is hard to distinguish between a graph chosen randomly from $\mathcal{G}_1$---a particular ``difficult'' subset of yes instances in $\mathcal{P}_k$---and a graph chosen randomly from $\mathcal{G}_2$---a particular ``difficult'' subset of no instances. $\mathcal{G}_1$ is the set of yes instances containing $2^k-1$ candy subgraphs, all of which are welded along a single cycle, such that exactly half of the candy subgraphs with an even number of self-loops have no self-loops. $\mathcal{G}_2$ contains the same graphs as $\mathcal{G}_1$ except that each candy graph is modified to take on the shape of a ``double-bow-tie'': the weld vertices within a ``bow-tie'' are connected by an arbitrary (single) cycle, and exactly half of the roots have a self-loop.

We prove that classically at least $2^{\Omega(k)}$ queries are needed to distinguish between random instances from $\mathcal{G}_1$ and $\mathcal{G}_2$. With a slight abuse of notation, we also use $\mathcal{G}_i$ to mean a random graph sampled from the set $\mathcal{G}_i$. More precisely, we define the sampling procedures for $\mathcal{G}_i$ as follows:
\begin{enumerate}[label=\textbf{\alph{*}.},ref=\textbf{\alph{*}}]
    \item Call a graph a \emph{depth-$k$ root-joined binary tree} if it is obtained by identifying the root vertices of two binary  trees of depth $k$, one of which is called antenna and the other body. Create $4(2^k-1)$ depth-$k$ root-joined binary trees, and arrange them in $2(2^k-1)$ pairs.
    \item For the first $2^{k}$ pairs mark one of the roots with a self-loop, and for the last $2^{k-1}-1$ pairs mark both roots with a self-loop. \label{it:self-loops}
    \item Pick a random bijection between body interior vertices and antenna vertices that are in a pair with exactly one self-loop,
    and pick another random bijection between body weld vertices and antenna vertices that are in a pair with an even number of self-loops. Connect vertices along the bijections with a double edge. \label{it:double-edges}
    \item
    \begin{enumerate}[label=$G_\arabic{*}$:]
    	\item In each pair, join the body weld vertices in the two root-joined binary trees of the pair by a single random cycle that alternates between vertices in the two trees.
    	\item In each pair, join the body weld vertices within both root-joined binary trees by a single random cycle.
    \end{enumerate}
\end{enumerate}
When an adjacency-list oracle is sampled for the above graphs, then vertex labels are randomly permuted, as well as the adjacency lists of every vertex.\footnote{Note that the above procedure samples from graph isomorphism classes non-uniformly; indeed, classes with richer automorphism groups tend to be sampled more often.}

First, we verify that a random no instance in $\mathcal{G}_2$ is indeed far from the yes instances in expectation. Let us call a graph in $\mathcal{G}_i$ \emph{reduced} if all its advice edges and self-loops have been removed (i.e., in the above procedure we ignore steps \ref{it:self-loops}\ and \ref{it:double-edges}). Observe that reduced graphs in $\mathcal{G}_1$ are bipartite (simply alternately color each layer in each welded tree). On the other hand, we prove that reduced graphs in $\mathcal{G}_2$ are typically far from even being bipartite.

Second, we argue that welded and self-welded trees are hard to distinguish by a classical randomized algorithm that makes few queries. Then, it follows that graphs from $\mathcal{G}_1$ and $\mathcal{G}_2$ are themselves hard to distinguish since they are based on the welded and self-welded trees, respectively, but are otherwise the same. We argue this formally by defining a fixed simulator $\mathcal{G}_s$ which we show behaves like both $\mathcal{G}_1$ and $\mathcal{G}_2$ when only sub-exponentially many queries are used.

\subsubsection{A random graph in \texorpdfstring{$\mathcal{G}_2$}{G2} is typically far from any yes instance}

We show that the reduced version of $\mathcal{G}_2$ is typically constant-far from reduced graphs in $\mathcal{G}_1$. It is easy to see that this also implies that (non-reduced) $\mathcal{G}_2$ is typically constant-far from $\mathcal{G}_1$.

\begin{lemma}~\label{lem:far_from_bipartite}
With probability at least $1-\exp(-\Omega(2^{k}))$, we need to remove at least $2^{k}\cdot(2^{k}-1)/16$ edges from a random reduced graph in $\mathcal{G}_2$ to make it bipartite.
\end{lemma}
\begin{proof}
Let us consider one self-welded component in a reduced graph $\mathcal{G}_2$, and in particular its induced subgraph $B$, spanned by the vertices in the weld layer and the layer adjacent to it, so that $m = 3\cdot 2^{k-1}$ is the total number of vertices in $B$. We show the following:

\begin{claim}
With probability at least $1-\exp(-\Omega(m))$, we need to remove at least $m/96$ edges from $B$ to make it bipartite.
\end{claim}

This implies the statement of the lemma because there are $4\cdot(2^{k}-1)$ copies of $B$ in any graph in $\mathcal{G}_2$. If all of them are far from bipartite, in the sense of needing to remove at least $m/96$ edges, then the total number of edges that need to be removed is at least $2^{k}\cdot(2^{k}-1)/16$. On the other hand, the probability of any of them not being far from bipartite is at most $4\cdot(2^{k}-1) \cdot \exp(-\Omega(2^{k})) = \exp(-\Omega(2^k))$.

\begin{proof}
This can be argued following the strategy of~\cite[Lemma 7.4]{GR97}.
The sampling procedure in terms of $B$ can be described as follows: take $2^{k-1}$ disjoint paths of length $2$, enumerate all paths and give labels $1,2,\ldots, 2^k$ to the endpoints of the paths, then sample a random permutation $\pi$ of $[2^k]$, and connect $\pi^{-1}(1)$ to $\pi^{-1}(2)$, $\pi^{-1}(2)$ to $\pi^{-1}(3)$, $\ldots$, and $\pi^{-1}(2^k)$ to $\pi^{-1}(1)$.
Observe that $B$ can be alternatively sampled by fixing a cycle of length $2^k$, and randomly matching its vertices as follows:
sample a random permutation $\pi$ of $[2^k]$, then connect with a length-2 path $\pi^{-1}(1)$ to $\pi^{-1}(2)$, $\pi^{-1}(2)$ to $\pi^{-1}(3)$, $\ldots$, and $\pi^{-1}(2^k)$ to $\pi^{-1}(1)$. It is easy to see that for a given $\pi$ the graphs sampled in the two different ways are isomorphic (via the permutation of the cycle vertices $\pi$).

It is more convenient to work with the second sampling procedure.
Consider a particular partition $(U_1,U_2)$ of the cycle vertices in $B$ and any partition $(V_1,V_2)$ of all vertices in $B$ that extends $(U_1,U_2)$ in the sense that $U_i\subset V_i$. We say an edge violates $(V_1,V_2)$ if it connects two vertices from the same $V_i$. Let $c := 2m/3$ denote the number of cycle vertices. We separate the analysis into two cases:

\begin{enumerate}[labelwidth=4em,leftmargin=!]
    \item[Case 1.] There are at least $c/64$ violating cycle edges.

    \item[Case 2.] There are fewer than $c/64$ violating cycle edges. Then assume without loss of generality that $|U_1|\leq |U_2|$. If $|U_1| \leq c/2 - c/128$, then there are at most $c-c/64$ non-violating cycle edges, so there are at least $c/64$ violating cycle edges, which contradicts our assumption. Hence $|U_1| \geq c/2 - c/128\geq 3c/8$.

    To randomly match vertices in $U_1$, we can take an arbitrary unmatched vertex $u$ in $U_1$ and connect it randomly to a cycle vertex $w$ by a path of length $2$. Note that if $w$ is in $U_2$, then we must have one violating edge on the length-2 path with respect to any $(V_1,V_2)$ partition. We continue until all the vertices in $U_1$ are matched.

    At each step, the probability of $w$ being in $U_2$ is at least $1/2$, so the probability that fewer than $c/64$ violating edges are created in a total of at least $3c/16$ steps is at most
    \begin{equation}
        \sum_{i=0}^{c/64-1} \binom{3c/16}{i} \cdot  2^{i-3c/16}
        \leq 2^{-11c/64}   \sum_{i=0}^{c/64} \binom{3c/16}{i}
        \leq 2^{\frac{3c}{16}H(16/3/64)-\frac{11c}{64}}
        \leq 2^{-0.11c},
    \end{equation}
    where $H(p) \coloneqq -p\log(p) - (1-p)\log(1-p)$ is the binary entropy function, and the second inequality follows from \cite[Corollary 22.9]{jukna2011ExtremalCombi2} stating
    \begin{align}\label{eq:entropyBinom}
    \forall\, 0 < h \leq n/2 :\,\, \sum_{j=0}^h\binom{n}{j} &\leq 2^{n\cdot H(h/n)}.
    \end{align}

    Note that this bounds the probability that \textit{any} partition $(V_1,V_2)$ extending $(U_1,U_2)$ has fewer than $c/64$ violating edges.
\end{enumerate}

We call a partition $(V_1,V_2)$ of all vertices in $B_1$ ``bad'' if it has fewer than $c/64$ violating edges. Then by the union bound
\begin{equation}
\begin{aligned}
    \Pr(\exists \text{ bad partition } (V_1,V_2)) &\leq \sum_{(U_1,U_2)} \Pr(\exists \text{ bad partition extending } (U_1,U_2))\\
    &\leq \#\{(U_1,U_2): \text{Case }2\}\times 2^{-0.11c},
\end{aligned}
\end{equation}
where $(U_1,U_2)$ denotes a partition of the cycle edges and $\#\{(U_1,U_2): \text{Case }2\}$ is the number of such partitions with fewer than $c/64$ violating cycle edges. For each $i<c/64$, each partition which has $i$ violating cycle edges is determined by the choice of those edges. Therefore
\begin{equation}
    \#\{(U_1,U_2): \text{Case }2\} = \sum_{i=0}^{c/64-1}\binom{c}{i} \leq 2^{0.09c}.
\end{equation}
Therefore $\Pr(\exists \text{ bad partition } (V_1,V_2)) \leq 2^{-0.02c} = \exp(-\Omega(c)) = \exp(-\Omega(m))$ as desired.
\end{proof}
This completes the proof of the lemma.
\end{proof}

\subsubsection{Hardness of distinguishing welded and self-welded trees}

Let $\mathcal{B}$ be a classical randomized algorithm. We consider the difficulty of $\mathcal{B}$ winning four different games by querying an oracle that provides black-box access to either a self-welded or welded tree. To be clear, in this subsection, by a ``self-welded tree'' we mean two binary trees of depth $k$ where the weld vertices in each are connected to themselves by a single cycle; by a ``welded tree'' we mean two binary trees of depth $k$ where the weld vertices in each are connected to weld vertices in the other by a single cycle. A self-welded tree looks like a candy graph without the antenna and a welded tree looks like a double-bow-tie graph without the antenna (cf.~\fig{modified_glued_tree}).

The difficulty of winning can be quantified by bounding the winning probability by a function of $t$, the number of queries that $\mathcal{B}$ makes. Throughout, the winning probability is over three or four sources of randomness: the internal randomness of $\mathcal{B}$, the cycle forming the (self-)weld, the vertex labelings, and also the random choice of the starting vertex in two of the games. These four games are described in the four lemmas below. They relate to the difficulty of distinguishing welded from self-welded trees because, unless they are won, the welded and self-welded trees both look just like a large (effectively infinite) binary tree to $\mathcal{B}$.

We can, without loss of generality, assume that $\mathcal{B}$ does not query labels which have already been queried and that each query to a vertex label returns the labels of \textit{all} neighbors of that vertex. For notational convenience, we shall use the concept of the ``knowledge graph'' of $\mathcal{B}$, as introduced in \cite[Section 7]{GR97}. The knowledge graph of $\mathcal{B}$ is the graph that it sees, which consists of all vertices whose labels it has either queried or has been returned by the oracle in answer to those queries, as well as the edges between them. The knowledge graph can change every time $\mathcal{B}$ makes a query. In the following, we make the further assumption that $\mathcal{B}$ can only query labels in its knowledge graph. This assumption shall later be justified by \lem{random_goes_to_untapped}.

\begin{lemma}\label{lem:self_weld_entrance}
Let Game A be that of finding a cycle in a self-welded tree given the label of \rooot. If $\mathcal{B}$ uses $t\leq 2^{k-1}$ queries, then its probability of winning Game A is at most $O(t^2 \cdot 2^{-k/2})$.
\end{lemma}

\begin{proof}
As explained in \cite{CCD+03}, the winning probability can be expressed as the probability that a random embedding $\pi$ of a random rooted binary tree $T$,
with $t$ vertices\footnote{Strictly speaking, $T$ has $O(t)$ vertices as we are assuming the oracle answers each query to a vertex label with all labels of its neighbors. But, because the lemma is only up to big-$O$, this is unimportant.}
and root at \rooot, into a random self-welded tree $G$ contains a cycle, where the randomness in the embedding corresponds to the random vertex labelings, and the randomness in the rooted binary tree corresponds to the internal randomness of $\mathcal{B}$. The terms ``random embedding'' and ``rooted binary tree'' are defined in \cite[Sec.~IV]{CCD+03}.

Therefore, it suffices to show that for an arbitrary fixed $T$, the probability that its random embedding $\pi(T)$ in a random $G$ contains a cycle is at most $O(t^2 \cdot 2^{-k/2})$. We denote this probability by $\Ex_{G}[P^{G}(T)]$ to match the notation of \cite{CCD+03}. Our proof is very similar to that of \cite[Lemma 8]{CCD+03}.

The probability that $\pi(T)$ contains a cycle is the same as the probability that there exist vertices $a\neq b$ in $T$ such that $\pi(a)=\pi(b)$.

First note that the probability that $\pi(T)$ ever goes to a depth smaller than $k/2$ after reaching a weld vertex is at most $t^2\cdot 2^{-k/2-1}$, as $\pi$ must embed $T$ to the left $k/2+1$ times starting from some vertex in $T$ and there are at most $t$ paths from root to leaf in $T$, each containing at most $t$ vertices.

Now fix any pair of vertices $a\neq b$ in $T$. We shall bound the probability that $\pi(a)=\pi(b)$. Let $P$ be the path in $T$ from $a$ to $b$, and $c$ the vertex in $T$ closest to the root on $P$. Let $P_1$ and $P_2$ be the paths in $T$ from $c$ to $a$ and $b$ respectively. The vertices at depths $k/2+1, \dots, k$ divide into $2^{k/2}$ complete binary trees of depth $k/2$ which we denote $S_{1},\dots,S_{2^{k/2}}$.

We have already considered the case that $\pi(P_i)$ ever goes to depth smaller than $k/2$ after reaching a weld vertex, so we can assume it does not. In this case, $\pi(a)=\pi(b)$ must lie in one of the subtrees $S_j$ and one of $\pi(P_1)$ or $\pi(P_2)$ must go from a weld vertex to a vertex in $S_j$. The probability of the latter occurring is at most $2^{k/2}/(2^{k}-t) \leq 2\cdot 2^{-k/2}$ for $t\leq 2^{k-1}$ due to the randomness of the self-weld. There are $\binom{t}{2}$ choices of $a,b$, so the probability of finding a cycle is at most $t^2\cdot 2^{-k/2}$.

Overall, we have shown
\begin{equation}
    \Ex_{G}[P^G(T)] \leq t^2 \cdot 2^{-k/2}+t^2\cdot 2^{-k/2-1} = O(t^2\cdot 2^{-k/2})
\end{equation}
as claimed.
\end{proof}

\begin{lemma}\label{lem:self_weld_random}
Let Game B be that of finding \rooot\ or a cycle in a self-welded tree given the label of a uniformly random starting vertex. If $\mathcal{B}$ uses $t\leq 2^{k-1}$ queries, then its probability of winning is at most $O(t^2\cdot 2^{-k/4})$.
\end{lemma}

\begin{proof}
Again, we can bound this probability by bounding the probability that a random embedding $\pi$ of an arbitrary fixed rooted binary tree $T$ contains \rooot\ or a cycle. The only difference in setup is that the root is now at the random vertex in the self-welded tree and the probability is also over this randomness.
First, the probability of the random vertex being at depth less than $3k/4$ is at most $2^{3k/4}\cdot 2^{-k} = 2^{-k/4}$. We may suppose that $\mathcal{B}$ wins in this case.

Therefore, consider when the random vertex is at depth more than $3k/4$. Then, similarly to the proof of the Lemma 6, we can bound the probability that $\pi(T)$ ever goes to depth smaller than $k/2$ by $t^2\cdot 2^{-k/4-1}$.

The probability of finding a cycle is again $t^2\cdot 2^{-k/2}$ by the same argument as before. Therefore, the overall probability of $\mathcal{B}$ finding the \rooot\ or a cycle starting from a random vertex is $O(t^2\cdot 2^{-k/4})$.
\end{proof}

\begin{lemma}\label{lem:normal_weld_entrance}
Given a \rooot, let Game C be that of finding the other \rooot\ or a cycle in a welded tree. If $\mathcal{B}$ uses $t\leq 2^{k-1}$ queries, then its probability of winning is at most $O(t^2\cdot 2^{-k/2})$.
\end{lemma}
\begin{proof}
The proof is given by~\cite[Lemma 8]{CCD+03} and is very similar to that of \lem{self_weld_entrance}.
\end{proof}

\begin{lemma}\label{lem:normal_weld_random}
Let Game D be that of finding a \rooot\ or a cycle in a welded tree given the label of a uniformly random starting vertex. If $\mathcal{B}$ uses $t\leq 2^{k-1}$ queries, then its probability of winning is at most $O(t^2\cdot 2^{-k/4})$.
\end{lemma}

\begin{proof}
The proof is essentially the same as that of \lem{self_weld_random} and \lem{normal_weld_entrance}.
\end{proof}

\subsubsection{Hardness of distinguishing \texorpdfstring{$\mathcal{G}_1$}{G1} and \texorpdfstring{$\mathcal{G}_2$}{G2}}

Let $\mathcal{A}$ be a classical randomized algorithm that distinguishes $\mathcal{G}_1$ from $\mathcal{G}_2$. Our strategy for proving the classical lower bound is to consider how the ``knowledge graph'' of $\mathcal{A}$ changes as it interacts with oracles that provide black-box descriptions of a random graph in $\mathcal{G}_1$ or $\mathcal{G}_2$. Recall that the knowledge graph of $\mathcal{A}$ is the graph that it sees. In the following, we refer to vertices in the knowledge graph at each step as ``known'' and vertices outside it as ``unknown''.

We shall show that, regardless of whether we are given a graph from $\mathcal{G}_1$ or $\mathcal{G}_2$, the knowledge graph of $\mathcal{A}$ looks the same at each step with high probability if only a small number of queries have been made. To be more precise, we shall define a graph simulator $\mathcal{G}_s$ such that the knowledge graph of $\mathcal{A}$ when it interacts with either $\mathcal{G}_1$ or $\mathcal{G}_2$ looks the same as when it interacts with $\mathcal{G}_s$ with all but exponentially small probability, if only a polynomial number of queries are made.

To the benefit of $\mathcal{A}$, we assume that $\mathcal{A}$ has knowledge of which labels are antenna, \rooot, or body labels. We also suppose that all reached \rooot\ labels have been queried for free so that they and their neighbors all belong to the knowledge graph of $\mathcal{A}$ at the outset.
We can assume that $\mathcal{A}$ does not query labels that have already been queried and that each query to a vertex label returns the labels of all neighbors of that vertex. Then, at each step, $\mathcal{A}$ must choose to perform one of the following Actions:

\begin{enumerate}[label=\Roman*.]
    \item Query an unknown body label.

    \item Query an unknown antenna label.

    \item Query a known body label.

    \item Query a known antenna label.
\end{enumerate}

For ease of reference, we call the candy subgraphs of $\mathcal{G}_1$ and double-bow-tie subgraphs of $\mathcal{G}_2$ their  ``base graphs''. At each step of $\mathcal{A}$, we say that a base graph has been \textit{tapped} if any one of its vertices, other than a \rooot, is already in the knowledge graph of $\mathcal{A}$.

At each step of $\mathcal{A}$, we assume to its benefit that:
\begin{enumerate}
    \item[A1.] $\mathcal{A}$ wins if, upon taking Action I or Action II, the label of a vertex in a tapped base graph is revealed.

    \item[A2.] $\mathcal{A}$ wins if, upon taking Actions III or IV, the vertex whose label is queried is connected by an advice edge to a vertex in a tapped base graph.

    \item[A3.] When $\mathcal{A}$ performs Action II, the entire half-antenna in which the queried antenna label resides is revealed (i.e., all vertex labels within the antenna as well as edges between them become known).

    \item[A4.] When $\mathcal{A}$ performs Action IV, the entire half-antenna in which the queried antenna label resides is revealed.
\end{enumerate}

The point of making these assumptions is simply to make it easier to analyze how the knowledge graph of $\mathcal{A}$ changes. We consider A1 and A2 because when they do not occur, $\mathcal{A}$ is essentially restricted to local traversal in each base graph. We consider A3 and A4 so that we do not need to unduly worry about how $\mathcal{A}$ traverses the antenna, which is irrelevant for distinguishing $\mathcal{G}_1$ from $\mathcal{G}_2$.

\begin{lemma}\label{lem:random_goes_to_untapped}
If $\mathcal{A}$ uses $t$ queries, then its  probability of winning via either A1 or A2 above is at most $O(t^2\cdot 2^{-k})$.
\end{lemma}

\begin{proof}
In $t$ queries, at most $2t$ base graphs can be tapped. However, there are $2^k$ base graphs in $\mathcal{G}_1$ and $\mathcal{G}_2$. Therefore, the probability of winning via A1 at each step is at most $O(t/2^{k})$, so the probability of winning via A1 at any step is $O(t^2/2^k)$. The same bound holds for winning via A2 due to the randomness in the connection of advice edges.
\end{proof}

Let us consider a graph simulator $\mathcal{G}_s$ that can respond to the different Actions of $\mathcal{A}$. As mentioned previously, we shall argue that its responses are similar to the responses of a random graph from either $\mathcal{G}_1$ and $\mathcal{G}_2$ when the number of queries is not too large. The responses of $\mathcal{G}_s$ depend on the Actions of $\mathcal{A}$ as follows:
\begin{enumerate}[label=\Roman*.]
    \item $\mathcal{G}_s$ returns a label chosen uniformly at random from antenna labels not currently in the knowledge graph. $\mathcal{G}_s$ also returns three labels chosen uniformly at random from body labels not currently in the knowledge graph.

    \item $\mathcal{G}_s$ returns an entire half-antenna with labels chosen uniformly at random from antenna labels not currently in the knowledge graph. $\mathcal{G}_s$ also returns one body label chosen uniformly at random from body labels not currently in the knowledge graph.

    \item $\mathcal{G}_s$ does one of the following two case-dependent actions.

    \begin{enumerate}
        \item If the queried body label is known because it is connected to a known antenna vertex, $\mathcal{G}_s$ returns three labels chosen uniformly at random from body labels not currently in the knowledge graph.

        \item If the queried body label is known because it is connected to a known body vertex, $\mathcal{G}_s$ returns two labels chosen uniformly at random from body labels not currently in the knowledge graph as well as a label chosen uniformly at random from antenna labels not currently in the knowledge graph.
    \end{enumerate}

    \item $\mathcal{G}_s$ does one of the following two case-dependent actions.

    \begin{enumerate}
        \item If the queried antenna label is known because it is connected to a known body vertex, $\mathcal{G}_s$ returns an entire half-antenna with labels chosen uniformly at random from antenna labels not currently in the knowledge graph.

        \item If the queried antenna label is known because it is connected to a known antenna vertex,
        $\mathcal{G}_s$ returns a label chosen uniformly at random from body labels not currently in the knowledge graph.
    \end{enumerate}
\end{enumerate}
Note that in cases III and IV, $\mathcal{G}_s$ also returns vertex labels that are consistent with its knowledge graph when querying a known label, but we omitted describing this for convenience. In particular, in case IV(b), $\mathcal{G}_s$ returns antenna labels consistent with the relevant half-antenna which must be entirely known by assumptions A3 and A4.

Now, suppose that $\mathcal{A}$ makes at most $t \leq 2^{k-1}$ queries to either $\mathcal{G}_1$ or $\mathcal{G}_2$. Then \lem{random_goes_to_untapped} says that the probability of it winning, at any step, via A1 or A2 is at most $O(t^2\cdot 2^{-k})$. We henceforth assume that neither A1 nor A2 ever occurs. In this case, we explain why $\mathcal{G}_s$ responds like $\mathcal{G}_1$ and $\mathcal{G}_2$ when $\mathcal{A}$ performs each of the following Actions:
\begin{enumerate}[label=\Roman*.]
    \item When $\mathcal{A}$ queries unknown body label, because A1 does not occur, the corresponding body vertex as well as its neighbors lie in two untapped base graphs. Therefore, the labels returned are outside the knowledge graph of $\mathcal{A}$ and so are distributed at random among those not in its knowledge graph, just like those returned by $\mathcal{G}_s$ in response to Action I.

    \item When $\mathcal{A}$ queries an unknown antenna label, because A1 does not occur, the corresponding entire half-antenna (cf.~A3) and the neighboring body vertex lie in two untapped base graphs. Therefore, the labels returned are outside the knowledge graph of $\mathcal{A}$ and so are distributed at random among those not in its knowledge graph, just like those returned by $\mathcal{G}_s$ in response to Action II.

    \item  When $\mathcal{A}$ queries a known body label, there are two cases, depending on how the body vertex is known.
    \begin{enumerate}
        \item If it is known because it is connected to a known antenna vertex by an advice edge, then it must be in an untapped base graph because A2 does not occur. Hence $\mathcal{G}_s$ behaves the same as $\mathcal{G}_i$.
        \item If it is known because it is connected to a known body vertex, then \lem{self_weld_entrance} and \lem{self_weld_random} (respectively \lem{normal_weld_entrance} and \lem{normal_weld_random}) imply that $\mathcal{G}_s$ behaves the same as $\mathcal{G}_2$ (respectively $\mathcal{G}_1$), except with probability at most $O(t^2\cdot 2^{-k/4})$. \lem{self_weld_entrance} and \lem{normal_weld_entrance} address the case when the queried body vertex is connected to the root of a base graph. \lem{self_weld_random} and \lem{normal_weld_random} address the case when it is connected to a random vertex in the body of a base graph.
    \end{enumerate}

    \item When $\mathcal{A}$ queries a known antenna label, there are two cases, depending on how the antenna vertex is known.

    \begin{enumerate}
        \item If it is known because it is connected to a known body vertex by an advice edge, then it must be in an untapped base graph because A2 does not occur. Hence $\mathcal{G}_s$ behaves the same as $\mathcal{G}_i$ (cf.~A4).

        \item If it is known because it is connected to a known antenna vertex, then the entire half-antenna must already be known (cf.~A4), and the body vertex connected to it must be in an untapped base graph because A2 does not occur. Hence $\mathcal{G}_s$ behaves the same as $\mathcal{G}_i$.
    \end{enumerate}
\end{enumerate}

Therefore, the probability of $\mathcal{A}$ distinguishing $\mathcal{G}_i$ from $\mathcal{G}_s$ is at most $O(t^2\cdot 2^{-k/4})$ for each $i=1, 2$. Therefore, the probability of distinguishing $\mathcal{G}_1$ from $\mathcal{G}_2$ is also at most $O(t^2\cdot 2^{-k/4})$. Now, \lem{far_from_bipartite} implies that there exists a constant $\epsilon$ such that for all $k$ sufficiently large, the probability that $\mathcal{G}_2$ is  $\epsilon$-\textit{close} to $\mathcal{G}_1$ is at most $0.1$. Let $\mathcal{G}_2'$ be a random graph from those in $\mathcal{G}_2$ that are $\epsilon$-far from $\mathcal{G}_1$. Then the probability of distinguishing $\mathcal{G}_2'$ from $\mathcal{G}_2$ is at most $0.1$. Therefore, the probability of distinguishing $\mathcal{G}_2'$ from $\mathcal{G}_1$ is at most $0.1+O(t^2\cdot 2^{-k/4})$. Hence, to distinguish $\mathcal{G}_1$ from $\mathcal{G}_2'$ with probability at least $2/3$ requires $t=\exp(\Omega(k))$ queries. This establishes the desired classical lower bound in \thm{adjacency_list}.

\section{Open problems}
\label{sec:open}

We conclude by discussing a few open problems.

We proved that $p$-uniform hypergraph properties have at most a power $3p$ separation between quantum and randomized query complexity. However, we do not know if this is tight.

\begin{open}
What is the largest possible separation between $Q(f)$ and $R(f)$ for $p$-uniform hypergraph properties $f$? That is, what is the largest $k$ for which there exists such an $f$ with $R(f) = \Omega(Q(f)^{k})$?
\end{open}

We remark that the problem is open even for the case $p = 1$ of fully symmetric functions, where the best upper bound is $k \le 3$ due to Chailloux \cite{Cha18}, and the best lower bound is $k \ge 2$ for the $\OR$ function \cite{Gro96}. For larger $p$, it is at least possible to exhibit functions with a power separation of $k \ge p / 2$ by appealing to \prop{small_base}, choosing the function $f$ in \prop{small_base} to be Forrelation \cite{AA15}, and using an upper bound of $b(G) = O(m)$ for $G$ the action of $S_m$ on $\binom{m}{p}$ points (Theorem 3.2 of \cite{Hal12}). Thus there must be some dependence on $p$.

In \cor{imprimitive_classification}, we showed that well-shuffling permutation groups must be constructed out of a constant number of primitive groups with sufficiently large minimal base size. We conjecture that a converse holds, which would imply a complete dichotomy regarding which permutation groups allow super-polynomial quantum speedups and which do not:

\begin{conjecture}\label{conj:shuffling_dichotomy}
Let $\mathcal{G}$ be a collection of permutation groups that satisfies conditions (i)--(iv) of \cor{imprimitive_classification}. Then $\mathcal{G}$ is well-shuffling, and thus $R(f) = Q(f)^{O(1)}$ for every $f \in F(\mathcal{G})$.
\end{conjecture}

\conj{shuffling_dichotomy} is based on our intuition that it is challenging to find an interesting family of permutation groups that satisfies these conditions, but that cannot be constructed out of well-shuffling primitive groups using the transformations that preserve the well-shuffling property. For example, even small subgroups of a wreath product of well-shuffling groups can remain well-shuffling: the product of of two well-shuffling permutation groups $G_1 \times G_2$ (as defined in \defn{group_product}) is well-shuffling, as shown in \thm{group-product}. But $G_1 \times G_2$ can be viewed as a subgroup of $G_1 \wr G_2$ in imprimitive action, and is potentially much smaller than the full wreath product. Based on this intuition, we expect that it is possible to prove \conj{shuffling_dichotomy} essentially in the same way we classified the well-shuffling primitive groups, via reduction to the symmetric and alternating groups.

In this form, \conj{shuffling_dichotomy} would also imply that the well-shuffling property completely characterizes whether a collection of permutation groups allows super-polynomial quantum speedups or not. It would be interesting to prove such a characterization directly, without appealing to the structure of such groups.

In \prop{small_base}, we showed how to construct super-polynomial quantum speedups out of any sufficiently small permutation group. However, the construction requires large alphabets (indeed, even larger than the number of input symbols). Might it be possible to construct such functions over a smaller alphabet, or even a Boolean alphabet? Note that some permutation groups cannot be realized as the group of symmetries of a function with Boolean inputs \cite{Kis98}.

Finally, while we constructed a separation between classical and quantum graph property testing in the adjacency list model, it remains unclear whether such a separation exists for graph properties of practical interest. For example, it remains open whether there could be an exponential quantum speedup for testing bipartiteness \cite{ACL11}. Going beyond that particular example, it would be interesting to investigate the possibility of a super-polynomial quantum speedup for testing \emph{monotone} graph properties (i.e., those that remain satisfied when adding edges to a yes instance).

\section*{Acknowledgements}

We thank Scott Aaronson and Carl Miller for many helpful discussions. SP also thanks Zak Webb for many related discussions. Part of this work was done while visiting the Simons Institute for the Theory of Computing. We gratefully acknowledge the Institute's hospitality.

AMC and DW acknowledge support from the Army Research Office (grant W911NF-20-1-0015); the Department of Energy, Office of Science, Office of Advanced Scientific Computing Research, Quantum Algorithms Teams and Accelerated Research in Quantum Computing programs; and the National Science Foundation (grant CCF-1813814). AG acknowledges funding provided by Samsung Electronics Co., Ltd., for the project ``The Computational Power of Sampling on Quantum Computers''. Additional support was provided by the Institute for Quantum Information and Matter, an NSF Physics Frontiers Center (NSF Grant PHY-1733907). WK acknowledges support from a Vannevar Bush Fellowship from the US Department of Defense.

\appendix

\section{Proof of the quantum minimax lemma}
\label{app:minimax}

We prove \lem{minimax}, which we restate below.

\minimax*

\begin{proof}
By \cite{BSS03}, there is a finite bound $B$ expressible in terms of
$n$ and $|\Sigma|$ on the necessary size of the work space register for a quantum
algorithm computing $f$ with error at most $\epsilon$. This means the quantum query
algorithms we deal with can be assumed without loss
of generality to have work space size $B$. A quantum algorithm making $T$ queries
can be represented as a sequence of $T$ unitary matrices of size
upper bounded by $B$; this can be arranged as a finite vector of complex numbers.
It is not hard to see that the set of all
such valid quantum algorithms is a compact set.

For a quantum algorithm $Q$, let $\err(Q,x)$ denote the error $Q$ makes
when run on input $x\in\Dom(x)$; this is $\Pr[Q(x)\ne f(x)]$, where
$Q(x)$ is the random variable for the measured output of $Q$ when run on $x$.
We note that $\err(Q,x)$ is a continuous function of $Q$.
Let $v_Q$ be the vector in $\bR^{|\Dom(f)|}$ defined by
$v_Q[x]\coloneqq\err(Q,x)$. Then $v_Q$ is a continuous function of $Q$.
Further, let $V$ be the set of all such vectors $v_Q$ for valid quantum
algorithms $Q$ which make at most $\Q_\epsilon(f)-1$ queries.
Since the set of such valid quantum algorithms is compact and since $v_Q$
is continuous in $Q$, we conclude that $V$ is compact.
Furthermore, we claim that $V$ is convex: this is because for any two
quantum algorithms $Q$ and $Q'$, there is a quantum algorithm $Q''$ that
behaves like their mixture (in terms of its error on each input $x$).

Next, let $\Delta\subseteq\bR^{|\Dom(f)|}$ be the set of all probability
distributions over $\Dom(f)$. Then $\Delta$ is also convex and compact.
Finally, define $\alpha\colon V\times\Delta\to\bR$ by
$\alpha(v,\mu)\coloneqq\bE_{x\leftarrow\mu}v[x]=\sum_{x\in\Dom(f)}\mu[x]v[x]$.
Then $\alpha$ is continuous in each coordinate, and is \emph{saddle}:
that is, $\alpha(\cdot,\mu)$ is convex for each $\mu\in\Delta$
(indeed, it is linear),
and $\alpha(v,\cdot)$ is concave for each $v\in V$ (indeed, it is also linear).
A standard minimax theorem (e.g., \cite{Sio58}) then gives us
\[\min_{v\in V}\max_{\mu\in\Delta}\alpha(v,\mu)
=\max_{\mu\in\Delta}\min_{v\in V}\alpha(v,\mu).\]
For the left-hand side, it is clear that the maximum over $\mu$ (once the vector $v$
has been chosen) is the same as the maximum over $x\in\Dom(f)$ of $v[x]$.
This makes the left-hand side the minimum over $v\in V$ of $\|v\|_{\infty}$,
or equivalently, the minimum worst-case error of quantum algorithms making at most
$\Q_\epsilon(f)-1$ queries. By the definition of $\Q_\epsilon(f)$, this minimum
must be strictly greater than $\epsilon$ (or else $\Q_\epsilon(f)$ would be smaller).
Hence the left-hand side is strictly greater than $\epsilon$.

Looking at the right-hand side, we get a single distribution
$\mu$ such that every quantum algorithm $Q$ making at most $\Q_\epsilon(f)-1$
queries must make error greater than $\epsilon$ against $\mu$, as desired.
\end{proof}

\section{Quantum speedup for deep wreath product symmetries}
\label{app:wreath_tree}

Here, we complete the proof of \thm{deep_wreath_speedup}. We define $\kFaultDirectTrees$ almost exactly as in Definition 2 of \cite{Kim12}, with the addition that we allow the Boolean evaluation tree to be an arbitrary block composition $f_1 \circ f_2 \circ \cdots \circ f_d$ of (possibly different) direct Boolean functions $f_i$, where $f_i$ has $n_i$ inputs. As long as each $f_i$ is a symmetric function, then $\kFaultDirectTrees$ is symmetric under the iterated wreath product $S_{n_1} \wr S_{n_2} \wr \cdots \wr S_{n_d}$ of symmetric groups because for any vertex $v$, the quantity $\kappa(v)$ defined in Definition 2 of \cite{Kim12} does not depend on the ordering of the children of $v$.

Recall that we choose $f_i\colon S \to \{0,1\}$ where $S \subseteq \{0,1\}^{n_i}$ consists of the $n_i$-bit strings that have Hamming weight in $\{0,\lfloor n_i/2 \rfloor, \lceil n_i/2 \rceil, n_i\}$, and $f_i(x) = 0$ if and only if $x = 1^{n_i}$. The function $f_i$ is defined to be trivial on $x$ if $x = 0^{n_i}$ or $1^{n_i}$; otherwise $f_i$ is defined to be fault on $x$. If $f_i(x)$ is trivial on $x$, then all input bits of $x$ are defined to be strong. Otherwise, if $f_i(x)$ is fault on $x$, the bits $j$ such that $x_j = 1$ are defined to be weak, and the bits $j$ such that $x_j = 0$ are defined to be strong.

Using the notation from Definition 2.1 of \cite{ZKH12}, we define a span program $P_i$ for $f$ by $v_j = \frac{1}{\sqrt{n_i}} \in \mathbb{R}^1$ and $\chi_j = \bar{x}_j$ for all $j \in [n_i]$. We now verify that $P_i$ satisfies the key properties that are needed to prove the quantum query upper bound in \cite{Kim12}. We have $r_0 = \frac{1}{\sqrt{n_i}}(1, 1, \ldots, 1) \in \mathbb{C}^{n_i}$ and $C = 1$ as in Definition 2.2 of \cite{ZKH12}, which gives the following:
\begin{itemize}
    \item $\textsc{wsize}_1(P_i, x) = 1$ if input $x$ makes $f_i$ trivial. If $x = 0^n$, choosing $\vec{w} = r_0$ suffices to show this. Otherwise, if $x = 1^n$, then we are forced to choose $\vec{w} = r_0$.
    \item $\textsc{wsize}_1(P_i, x) \le 3$ if input $x$ makes $f_i$ fault. Suppose $x = 0^{\lfloor n_i / 2 \rfloor}1^{\lceil n_i / 2 \rceil}$. Choosing $w_j = \frac{\sqrt{n_i}}{\lfloor n_i / 2 \rfloor}$ for $j \in [\lfloor n_i / 2 \rfloor]$ gives witness size $\frac{n_i}{\lfloor n_i / 2 \rfloor} \le 3$ (tight for $n_i = 3$). Similarly, for $x = 0^{\lceil n_i / 2 \rceil}1^{\lfloor n_i / 2 \rfloor}$, we take $w_j = \frac{\sqrt{n_i}}{\lceil n_i / 2 \rceil}$ for $j \in [\lceil n_i / 2 \rceil]$ to get witness size $\frac{n_i}{\lceil n_i / 2 \rceil} \le 2$.
    \item For $\textsc{wsize}_s(P_i, x)$, $s_j$ do not affect the witness size, where the $j$th input bit is weak. This is because the only weak input bits are those where $x_j = 1$ and $f_i$ is fault on $x$. Then $f_i(x) = 1$, and we have $\chi_j = 0$, so we are forced to take $w_j = 0$.
\end{itemize}

Now that these three conditions are satisfied, the proof of Theorem 2 of \cite{Kim12} goes through without modifications, and shows that this version of $\kFaultDirectTrees$ has quantum query complexity $O(1)$.

It remains to lower bound the randomized query complexity of $\kFaultDirectTrees$ as we have defined it. Suppose we further promise that at every vertex in the tree corresponding to $f_i$, the child subtrees are partitioned into two sets, one of size $\lfloor n_i / 2 \rfloor$ and one of size $\lceil n_i / 2 \rceil$, such that all of the subtrees in one part have the same values at corresponding leaves. Then with this additional promise, the problem remains at least as hard as the $1\textsc{-Fault Nand Trees}$ problem: if we always take the partition to be the $\lfloor n_i / 2 \rfloor$ left vertices and the $\lceil n_i / 2 \rceil$ right vertices, then it is in fact equivalent to $1\textsc{-Fault Nand Trees}$ where some inputs are repeated. \cite{ZKH12} proved that depth-$d$ $1\textsc{-Fault Nand Trees}$ have randomized query complexity at least $\Omega(\log d)$, which completes our proof.

\phantomsection\addcontentsline{toc}{section}{References}
\renewcommand{\UrlFont}{\ttfamily\small}
\let\oldpath\path
\renewcommand{\path}[1]{\small\oldpath{#1}}
\newcommand{\etalchar}[1]{$^{#1}$}
\newcommand{\lName}{1}\newcommand{\arxiv}[1]{
  \href{https://arxiv.org/abs/#1}{\ttfamily{arXiv:#1}}\?}\newcommand{\arXiv}[1]{
  \href{https://arxiv.org/abs/#1}{\ttfamily{arXiv:#1}}\?}\def\?#1{\if.#1{}\else#1\fi}\providecommand{\multiletter}[1]{#1}\renewcommand{\multiletter}[1]{#1}\DeclareRobustCommand{\dutchPrefix}[2]{#2}\providecommand{\dutchPrefix}[2]{#2}\renewcommand{\dutchPrefix}[2]{#2}\newcommand{\skp}[3]{#2}\newcommand{\focs
  }[1]{\if\lName1\skp{ }{Proceedings of the #1 {IEEE} Symposium on Foundations
  of Computer Science ({FOCS})}{ }\else{FOCS}\fi}\newcommand{\stoc
  }[1]{\if\lName1\skp{ }{Proceedings of the #1 {ACM} Symposium on the Theory of
  Computing ({STOC})}{ }\else{STOC}\fi}\newcommand{\soda }[1]{\if\lName1\skp{
  }{Proceedings of the #1 {ACM-SIAM} Symposium on Discrete Algorithms
  ({SODA})}{ }\else{SODA}\fi}\newcommand{\stacs }[1]{\if\lName1\skp{
  }{Proceedings of the #1 Symposium on Theoretical Aspects of Computer Science
  ({STACS})}{ }\else{STACS}\fi}\newcommand{\itcs }[1]{\if\lName1\skp{
  }{Proceedings of the #1 Innovations in Theoretical Computer Science
  Conference (ITCS)}{ }\else{ITCS}\fi}\newcommand{\fsttcs }[1]{\if\lName1\skp{
  }{Proceedings of the #1 International Conference on Foundations of Software
  Technology and Theoretical Computer Science (FSTTCS)}{
  }\else{FSTTCS}\fi}\newcommand{\ccc }[1]{\if\lName1\skp{ }{Proceedings of the
  #1 {IEEE} Conference on Computational Complexity ({CCC})}{
  }\else{CCC}\fi}\newcommand{\isit }[1]{\if\lName1\skp{ }{Proceedings of the #1
  {IEEE} International Symposium on Information Theory ({ISIT})}{
  }\else{ISIT}\fi}\newcommand{\colt }[1]{\if\lName1\skp{ }{Proceedings of the
  #1 Conference On Learning Theory (COLT)}{ }\else{COLT}\fi}\newcommand{\nips
  }[1]{\if\lName1\skp{ }{Advances in Neural Information Processing Systems #1
  ({NIPS})}{ }\else{NIPS}\fi}\newcommand{\aistats }[1]{\if\lName1\skp{
  }{Proceedings of the #1 International Conference on Artificial Intelligence
  and Statistics ({AISTATS})}{ }\else{AISTATS}\fi}\newcommand{\icml
  }[1]{\if\lName1\skp{ }{Proceedings of the #1 International Conference on
  Machine Learning (ICML)}{ }\else{ICML}\fi}\newcommand{\icalp
  }[1]{\if\lName1\skp{ }{Proceedings of the #1 International Colloquium on
  Automata, Languages, and Programming (ICALP)}{
  }\else{ICALP}\fi}\newcommand{\esa }[1]{\if\lName1\skp{ }{Proceedings of the
  #1 Annual European Symposium on Algorithms (ESA)}{
  }\else{ESA}\fi}\newcommand{\tqc }[1]{\if\lName1\skp{ }{Proceedings of the #1
  Conference on the Theory of Quantum Computation, Communication, and
  Cryptography (TQC)}{}\else{TQC}\fi}\newcommand{\jacm }{\if\lName1\skp{
  }{Journal of the ACM}{ }\else{J. ACM}\fi}\newcommand{\acmta }{\if\lName1\skp{
  }{ACM Transactions on Algorithms}{ }\else{{ACM} Tr.
  Alg}\fi}\newcommand{\acmtct }{\if\lName1\skp{ }{ACM Transactions on
  Computation Theory}{ }\else{ACM Tr. Comp. Th.}\fi}\newcommand{\jams
  }{\if\lName1\skp{ }{Journal of the AMS}{ }\else{J. AMS}\fi}\newcommand{\pams
  }{\if\lName1\skp{ }{Proceedings of the AMS}{ }\else{Proc.
  AMS}\fi}\newcommand{\linalgappl }{\if\lName1\skp{ }{Linear Algebra and its
  Applications}{ }\else{Lin. Alg. \& App.}\fi}\newcommand{\jalgo
  }{\if\lName1\skp{ }{Journal of Algorithms}{ }\else{J.
  Alg.}\fi}\newcommand{\jcss }{\if\lName1\skp{ }{Journal of Computer and System
  Sciences}{ }\else{J. Comp. Sys. Sci.}\fi}\newcommand{\cc }{\if\lName1\skp{
  }{Computational Complexity}{ }\else{Comp. Comp.}\fi}\newcommand{\algor
  }{\if\lName1\skp{ }{Algorithmica}{ }\else{Alg.}\fi}\newcommand{\comb
  }{\if\lName1\skp{ }{Combinatorica}{ }\else{Comb.}\fi}\newcommand{\cacm
  }{\if\lName1\skp{ }{Communications of the ACM}{ }\else{Comm.
  ACM}\fi}\newcommand{\sigart }{\if\lName1\skp{ }{SIGART Bulletin}{
  }\else{SIGART Bull.}\fi}\newcommand{\sigactn }{\if\lName1\skp{ }{SIGACT
  News}{ }\else{SIGACT News}\fi}\newcommand{\eatcsbul }{\if\lName1\skp{
  }{Bulletin of the {EATCS}}{ }\else{Bull. {EATCS}}\fi}\newcommand{\siamrev
  }{\if\lName1\skp{ }{SIAM Review}{ }\else{SIAM Rev.}\fi}\newcommand{\siamjc
  }{\if\lName1\skp{ }{SIAM Journal on Computing}{ }\else{SIAM J.
  Comp.}\fi}\newcommand{\siamjo }{\if\lName1\skp{ }{SIAM Journal on
  Optimization}{ }\else{SIAM J. Opt.}\fi}\newcommand{\siamjdm }{\if\lName1\skp{
  }{SIAM Journal on Discrete Mathematics}{ }\else{SIAM J. Disc.
  Math.}\fi}\newcommand{\siamjnum }{\if\lName1\skp{ }{SIAM Journal on Numerical
  Analysis}{ }\else{SIAM J. Num. Anal.}\fi}\newcommand{\siamjmathanal
  }{\if\lName1\skp{ }{SIAM Journal on Mathematical Analysis}{ }\else{SIAM J.
  Math. Anal.}\fi}\newcommand{\discmath }{\if\lName1\skp{ }{Discrete
  Mathematics}{ }\else{Disc. Math.}\fi}\newcommand{\das }{\if\lName1\skp{
  }{Discrete Applied Mathematics}{ }\else{Disc. App.
  Math.}\fi}\newcommand{\amatstat }{\if\lName1\skp{ }{Annals of Mathematical
  Statistics}{ }\else{Ann. Math. Stat.}\fi}\newcommand{\rms }{\if\lName1\skp{
  }{Russian Mathematical Surveys}{ }\else{Russ. Math.
  Surv.}\fi}\newcommand{\invmath }{\if\lName1\skp{ }{Inventiones Mathematicae}{
  }\else{Inv. Math.}\fi}\newcommand{\jnumber }{\if\lName1\skp{ }{Journal of
  Number Theory}{ }\else{J. Num. Th.}\fi}\newcommand{\toc }{\if\lName1\skp{
  }{Theory of Computing}{ }\else{Th. Comp.}\fi}\newcommand{\cjtcs
  }{\if\lName1\skp{ }{Chicago Journal of Theoretical Computer
  Science}{}\else{Chic. J. Th. Comp. Sci.}\fi}\newcommand{\quantum
  }{\if\lName1\skp{ }{{Quantum}}{ }\else{Quant.}\fi}\newcommand{\cmp
  }{\if\lName1\skp{ }{Communications in Mathematical Physics}{ }\else{Comm.
  Math. Phys.}\fi}\newcommand{\jmp }{\if\lName1\skp{ }{Journal of Mathematical
  Physics}{ }\else{J. Math. Phys.}\fi}\newcommand{\rspa }{\if\lName1\skp{
  }{Proceedings of the Royal Society A}{ }\else{Proc. Roy. Soc.
  A}\fi}\newcommand{\qic }{\if\lName1\skp{ }{Quantum Information and
  Computation}{ }\else{Quant. Inf. \& Comp.}\fi}\newcommand{\physrev
  }{\if\lName1\skp{ }{Physical Review}{ }\else{Phys. Rev.}\fi}\newcommand{\pra
  }{\if\lName1\skp{ }{Physical Review A}{ }\else{Phys. Rev.
  A}\fi}\newcommand{\prb }{\if\lName1\skp{ }{Physical Review B}{ }\else{Phys.
  Rev. B}\fi}\newcommand{\pre }{\if\lName1\skp{ }{Physical Review E}{
  }\else{Phys. Rev. E}\fi}\newcommand{\prx }{\if\lName1\skp{ }{Physical Review
  X}{ }\else{Phys. Rev. X}\fi}\newcommand{\prl }{\if\lName1\skp{ }{Physical
  Review Letters}{ }\else{Phys. Rev. Lett.}\fi}\newcommand{\njp
  }{\if\lName1\skp{ }{New Journal of Physics}{ }\else{New J.
  Phys.}\fi}\newcommand{\prapp }{\if\lName1\skp{ }{Physical Review Applied}{
  }\else{Phys. Rev. Appl.}\fi}\newcommand{\physrep }{\if\lName1\skp{ }{Physics
  Reports}{ }\else{Phys. Rep.}\fi}\newcommand{\rmp }{\if\lName1\skp{ }{Reviews
  of Modern Physics}{ }\else{Rev. Mod. Phys. }\fi}\newcommand{\phystoday
  }{\if\lName1\skp{ }{Physics Today}{ }\else{Phys.
  Today}\fi}\newcommand{\physics }{\if\lName1\skp{ }{Physics}{
  }\else{Phys.}\fi}\newcommand{\nature }{\if\lName1\skp{ }{Nature}{
  }\else{Nat.}\fi}\newcommand{\natcomm }{\if\lName1\skp{ }{Nature
  Communications}{ }\else{Nat. Comm.}\fi}\newcommand{\natphys }{\if\lName1\skp{
  }{Nature Physics}{ }\else{Nat. Phys.}\fi}\newcommand{\npjqi }{\if\lName1\skp{
  }{npj Quantum Information}{ }\else{npj Quant. Inf.}\fi}\newcommand{\scirep
  }{\if\lName1\skp{ }{Scientific Reports}{ }\else{Sci.
  Rep.}\fi}\newcommand{\science }{\if\lName1\skp{ }{Science}{
  }\else{Sci.}\fi}\newcommand{\jpa }{\if\lName1\skp{ }{Journal of Physics A:
  Mathematical and Theoretical}{ }\else{J. Phys. A}\fi}\newcommand{\ijtp
  }{\if\lName1\skp{ }{International Journal of Theoretical Physics}{
  }\else{Int. J. Th. Phys.}\fi}\newcommand{\jmo }{\if\lName1\skp{ }{Journal of
  Modern Optics}{ }\else{J. Mod. Opt.}\fi}\newcommand{\jstatph
  }{\if\lName1\skp{ }{Journal of Statistical Physics}{ }\else{J. Stat.
  Phys.}\fi}\newcommand{\pnas }{\if\lName1\skp{ }{Proceedings of the National
  Academy of Sciences}{ }\else{PNAS}\fi}\newcommand{\lncs }{\if\lName1\skp{
  }{Lecture Notes in Computer Science}{ }\else{L. Notes Comp.
  Sci.}\fi}\newcommand{\lnai }{\if\lName1\skp{ }{Lecture Notes in Artificial
  Intelligence}{ }\else{L. Notes Art. Int.}\fi}\newcommand{\lnm
  }{\if\lName1\skp{ }{Lecture Notes in Mathematics}{ }\else{L. Notes
  Math.}\fi}\newcommand{\tams }{\if\lName1\skp{ }{Transactions of the American
  Mathematical Society}{ }\else{Trans. AMS}\fi}\newcommand{\ieeetit
  }{\if\lName1\skp{ }{{IEEE} Transactions on Information Theory}{ }\else{{IEEE}
  Trans. Inf. Th.}\fi}\newcommand{\iscs }{\if\lName1\skp{ }{International
  Series in Computer Science}{ }\else{Int. Ser. Comp.
  Sci.}\fi}\newcommand{\tocl }{\if\lName1\skp{ }{Theory of Computing Library}{
  }\else{Th. Comp. Lib.}\fi}


\begin{thebibliography}{GSLW19}

\bibitem[AA14]{AA14}
Scott Aaronson and Andris Ambainis.
\newblock \href{http://dx.doi.org/10.4086/toc.2014.v010a006}{The need for
  structure in quantum speedups}.
\newblock {\em \toc}, 10:133--166, 2014.
\newblock \arXiv{0911.0996}.

\bibitem[AA15]{AA15}
Scott Aaronson and Andris Ambainis.
\newblock \href{http://dx.doi.org/10.1145/2746539.2746547}{Forrelation: A
  problem that optimally separates quantum from classical computing}.
\newblock In {\em \stoc{47th}}, pages 307--316, 2015.
\newblock \arxiv{1411.5729}.

\bibitem[AB16]{AB16}
Scott Aaronson and Shalev Ben{-}David.
\newblock \href{http://dx.doi.org/10.4230/LIPIcs.CCC.2016.26}{Sculpting quantum
  speedups}.
\newblock In {\em \ccc{31st}}, pages 26:1--26:28, 2016.
\newblock \arXiv{1512.04016}.

\bibitem[ACL11]{ACL11}
Andris Ambainis, Andrew~M. Childs, and Yi-Kai Liu.
\newblock
  \href{http://dx.doi.org/http://dx.doi.org/10.1007/978-3-642-22935-0_31}{Quantum
  property testing for bounded-degree graphs}.
\newblock In {\em Proceedings of the 15th International Workshop on
  Randomization and Computation (RANDOM)}, volume 6845 of {\em Lecture Notes in
  Computer Science}, pages 365--376. Springer, 2011.
\newblock \arxiv{1012.3174}.

\bibitem[ACL{\etalchar{+}}19]{ACL+19}
Scott Aaronson, Nai-Hui Chia, Han-Hsuan Lin, Chunhao Wang, and Ruizhe Zhang.
\newblock On the quantum complexity of closest pair and related problems.
\newblock \arxiv{1911.01973}, 2019.

\bibitem[AK02]{AK02}
Noga Alon and Michael Krivelevich.
\newblock \href{http://dx.doi.org/10.1137/S0895480199358655}{Testing
  $k$-colorability}.
\newblock {\em SIAM Journal on Discrete Mathematics}, 15(2):211--227, 2002.

\bibitem[Amb05]{Amb05}
Andris Ambainis.
\newblock \href{http://dx.doi.org/10.4086/toc.2005.v001a003}{Polynomial degree
  and lower bounds in quantum complexity: Collision and element distinctness
  with small range}.
\newblock {\em \toc}, 1(1):37--46, 2005.
\newblock \arxiv{quant-ph/0305179}.

\bibitem[AS04]{AS04}
Scott Aaronson and Yaoyun Shi.
\newblock \href{http://dx.doi.org/10.1145/1008731.1008735}{Quantum lower bounds
  for the collision and the element distinctness problems}.
\newblock {\em \jacm}, 51(4):595--605, July 2004.

\bibitem[BBC{\etalchar{+}}01]{BBC+01}
Robert Beals, Harry Buhrman, Richard Cleve, Michele Mosca, and Ronald
  {\dutchPrefix{Wolf}{d}}e~Wolf.
\newblock \href{http://dx.doi.org/10.1145/502090.502097}{Quantum lower bounds
  by polynomials}.
\newblock {\em \jacm}, 48(4):778--797, 2001.
\newblock Earlier version in FOCS'98. \arxiv{quant-ph/9802049}.

\bibitem[BCK15]{BCK15}
Dominic~W. Berry, Andrew~M. Childs, and Robin Kothari.
\newblock \href{http://dx.doi.org/10.1109/FOCS.2015.54}{Hamiltonian simulation
  with nearly optimal dependence on all parameters}.
\newblock In {\em \focs{56th}}, pages 792--809, 2015.
\newblock \arxiv{1501.01715}.

\bibitem[BdW02]{BdW02}
Harry Buhrman and Ronald de~Wolf.
\newblock \href{http://dx.doi.org/10.1016/S0304-3975(01)00144-X}{Complexity
  measures and decision tree complexity: a survey}.
\newblock {\em \toc}, 288(1):21--43, 2002.

\bibitem[Ben16]{Ben16}
Shalev Ben{-}David.
\newblock \href{http://dx.doi.org/10.4230/LIPIcs.TQC.2016.7}{The structure of
  promises in quantum speedups}.
\newblock In {\em \tqc{11th}}, pages 7:1--7:14, 2016.
\newblock \arxiv{1409.3323}.

\bibitem[BKT18]{BKT18}
Mark Bun, Robin Kothari, and Justin Thaler.
\newblock \href{http://dx.doi.org/10.1145/3188745.3188784}{The polynomial
  method strikes back: {T}ight quantum query bounds via dual polynomials}.
\newblock In {\em \stoc{50th}}, pages 297--310, 2018.
\newblock \arxiv{1710.09079}.

\bibitem[B{\v S}13]{BS13}
Aleksandrs Belovs and Robert {\v S}palek.
\newblock \href{http://dx.doi.org/10.1145/2422436.2422474}{Adversary lower
  bound for the k-sum problem}.
\newblock In {\em \itcs{4th}}, pages 323--328, 2013.
\newblock \arxiv{1206.6528}.

\bibitem[BSS03]{BSS03}
Howard Barnum, Michael Saks, and Mario Szegedy.
\newblock \href{http://dx.doi.org/10.1109/CCC.2003.1214419}{Quantum query
  complexity and semi-definite programming}.
\newblock In {\em \ccc{18th}}, pages 179--193, 2003.

\bibitem[CCD{\etalchar{+}}03]{CCD+03}
Andrew~M. Childs, Richard Cleve, Enrico Deotto, Edward Farhi, Sam Gutmann, and
  Daniel~A. Spielman.
\newblock \href{http://dx.doi.org/10.1145/780542.780552}{Exponential
  algorithmic speedup by quantum walk}.
\newblock In {\em \stoc{35th}}, pages 59--68, 2003.
\newblock \arxiv{quant-ph/0209131}.

\bibitem[Cha18]{Cha18}
Andr{\'e} Chailloux.
\newblock \href{http://dx.doi.org/10.4230/LIPIcs.ITCS.2019.19}{A note on the
  quantum query complexity of permutation symmetric functions}.
\newblock In {\em \itcs{10th}}, pages 19:1--19:7, 2018.
\newblock \arxiv{1810.01790}.

\bibitem[Cle04]{Cle04}
Richard Cleve.
\newblock \href{http://dx.doi.org/10.1016/j.ic.2004.04.001}{The query
  complexity of order-finding}.
\newblock {\em Information and Computation}, 192(2):162--171, 2004.
\newblock \arxiv{quant-ph/9911124}.

\bibitem[DM12]{DM12}
John~D. Dixon and Brian Mortimer.
\newblock \href{https://books.google.com/books?id=1SPjBwAAQBAJ}{{\em
  Permutation Groups}}, volume 163 of {\em Graduate Texts in Mathematics}.
\newblock Springer New York, 2012.

\bibitem[GGR98]{GGR98}
Oded Goldreich, Shafi Goldwasser, and Dana Ron.
\newblock \href{http://dx.doi.org/10.1145/285055.285060}{Property testing and
  its connection to learning and approximation}.
\newblock {\em \jacm}, 45(4):653--750, 1998.

\bibitem[GR97]{GR97}
Oded Goldreich and Dana Ron.
\newblock \href{http://dx.doi.org/10.1145/258533.258627}{Property testing in
  bounded degree graphs}.
\newblock In {\em \stoc{29th}}, page 406–415, 1997.

\bibitem[Gro96]{Gro96}
Lov~K. Grover.
\newblock \href{http://dx.doi.org/10.1145/237814.237866}{A fast quantum
  mechanical algorithm for database search}.
\newblock In {\em \stoc{28th}}, page 212–219, 1996.
\newblock \arxiv{quant-ph/9605043}.

\bibitem[GSLW19]{GSLW18}
András Gilyén, Yuan Su, Guang~Hao Low, and Nathan Wiebe.
\newblock \href{http://dx.doi.org/10.1145/3313276.3316366}{Quantum singular
  value transformation and beyond: exponential improvements for quantum matrix
  arithmetics}.
\newblock In {\em \stoc{51st}}, pages 193--204, 2019.
\newblock \arxiv{1806.01838}.

\bibitem[Hal12]{Hal12}
Zolt\'an Halasi.
\newblock \href{http://dx.doi.org/10.1556/SScMath.49.2012.4.1222}{On the base
  size for the symmetric group acting on subsets}.
\newblock {\em Studia Scientiarum Mathematicarum Hungarica}, 49, 12 2012.

\bibitem[Hul10]{Hul10}
Alexander Hulpke.
\newblock Notes on computational group theory, 2010.
\newblock \url{https://www.math.colostate.edu/~hulpke/CGT/cgtnotes.pdf}.

\bibitem[Juk11]{jukna2011ExtremalCombi2}
Stasys Jukna.
\newblock \href{http://dx.doi.org/10.1007/978-3-642-17364-6}{{\em Extremal
  Combinatorics}}.
\newblock Texts in Theoretical Computer Science. Springer, 2011.

\bibitem[Ker13]{Ker13}
Adalbert Kerber.
\newblock \href{http://dx.doi.org/10.1007/978-3-662-11167-3}{{\em Applied
  finite group actions}}, volume~19 of {\em Algorithms and Combinatorics}.
\newblock Springer Science \& Business Media, 2013.

\bibitem[Kim12]{Kim12}
Shelby Kimmel.
\newblock \href{http://dx.doi.org/10.1007/978-3-642-31594-7_47}{Quantum
  adversary (upper) bound}.
\newblock In {\em \icalp{12th}}, volume 7391 of {\em Lecture Notes in Computer
  Science}, pages 557--568, 2012.
\newblock \arxiv{1101.0797}.

\bibitem[Kis98]{Kis98}
Andrzej Kisielewicz.
\newblock
  \href{http://dx.doi.org/https://doi.org/10.1006/jabr.1997.7198}{Symmetry
  groups of {B}oolean functions and constructions of permutation groups}.
\newblock {\em Journal of Algebra}, 199(2):379 -- 403, 1998.

\bibitem[Kut05]{Kut05}
Samuel Kutin.
\newblock \href{http://dx.doi.org/10.4086/toc.2005.v001a002}{Quantum lower
  bound for the collision problem with small range}.
\newblock {\em Theory of Computing}, 1(1):29--36, 2005.
\newblock \arxiv{quant-ph/0304162}.

\bibitem[LC19]{LC19}
Guang~Hao Low and Isaac~L. Chuang.
\newblock \href{http://dx.doi.org/10.22331/q-2019-07-12-163}{Hamiltonian
  simulation by qubitization}.
\newblock {\em \quantum}, 3:163, 2019.
\newblock \arxiv{1610.06546}.

\bibitem[Lie84]{Lie84}
Martin~W. Liebeck.
\newblock \href{http://dx.doi.org/10.1007/BF01193603}{On minimal degrees and
  base sizes of primitive permutation groups}.
\newblock {\em Archiv der Mathematik}, 43(1):11--15, 1984.

\bibitem[MdW16]{MdW13}
Ashley Montanaro and Ronald de~Wolf.
\newblock \href{http://dx.doi.org/10.4086/toc.gs.2016.007}{{\em A survey of
  quantum property testing}}.
\newblock Number~7 in Graduate Surveys. Theory of Computing Library, 2016.
\newblock \arxiv{1310.2035}.

\bibitem[Sho94]{Shor94}
Peter~W. Shor.
\newblock \href{http://dx.doi.org/10.1137/S0097539795293172}{Algorithms for
  quantum computation: Discrete logarithms and factoring}.
\newblock In {\em \focs{35th}}, pages 124--134, 1994.
\newblock \arxiv{quant-ph/9508027}.

\bibitem[Sim97]{Sim97}
Daniel~R. Simon.
\newblock \href{http://dx.doi.org/10.1137/S0097539796298637}{On the power of
  quantum computation}.
\newblock {\em \siamjc}, 26(5):1474--1483, 1997.

\bibitem[Sio58]{Sio58}
Maurice Sion.
\newblock \href{http://dx.doi.org/10.2140/pjm.1958.8.171}{On general minimax
  theorems}.
\newblock {\em Pacific Journal of Mathematics}, 8(1):171--176, 1958.

\bibitem[Ver98]{Ver98}
Nikolai~K. Vereshchagin.
\newblock \href{http://dx.doi.org/10.1016/S0304-3975(98)00071-1}{Randomized
  {B}oolean decision trees: Several remarks}.
\newblock {\em Theoretical Computer Science}, 207(2):329--342, 1998.

\bibitem[Yao77]{Yao77}
Andrew Chi-Chih Yao.
\newblock \href{http://dx.doi.org/10.1109/SFCS.1977.24}{Probabilistic
  computations: {T}oward a unified measure of complexity}.
\newblock In {\em \focs{18th}}, pages 222--227, 1977.

\bibitem[Zha15]{Zha13}
Mark Zhandry.
\newblock \href{http://dx.doi.org/10.26421/QIC15.7-8}{A note on the quantum
  collision and set equality problems}.
\newblock {\em \qic}, 15(7\&8):557--567, 2015.
\newblock \arxiv{1312.1027}.

\bibitem[ZKH12]{ZKH12}
Bohua Zhan, Shelby Kimmel, and Avinatan Hassidim.
\newblock \href{http://dx.doi.org/10.1145/2090236.2090258}{Super-polynomial
  quantum speed-ups for {B}oolean evaluation trees with hidden structure}.
\newblock In {\em \itcs{3rd}}, pages 249--265, 2012.
\newblock \arxiv{1101.0796}.

\end{thebibliography}
\end{document}